\documentclass[letterpaper]{article} % DO NOT CHANGE THIS
\usepackage{aaai24}
\usepackage{times}  % DO NOT CHANGE THIS
\usepackage{helvet}  % DO NOT CHANGE THIS
\usepackage{courier}  % DO NOT CHANGE THIS
\usepackage[hyphens]{url}  % DO NOT CHANGE THIS
\usepackage{graphicx} % DO NOT CHANGE THIS
\usepackage{natbib}  % DO NOT CHANGE THIS AND DO NOT ADD ANY OPTIONS TO IT
\usepackage{caption} % DO NOT CHANGE THIS AND DO NOT ADD ANY OPTIONS TO IT
\nocopyright %-- Your paper will not be published if you use this command
\newcommand{\primarykeywords}{
(7) Multi-Agent Planning;
}

\usepackage[switch, modulo]{lineno}
\title{A Conflict-Aware Optimal Goal Assignment Algorithm for Multi-Robot Systems}

\author {
    Aakash and Indranil Saha
}
\affiliations {
Department of Computer Science and Engineering\\
Indian Institute of Technology Kanpur\\
\{aakashp, isaha\}@cse.iitk.ac.in
}

\usepackage{comment}  % aks
\usepackage{subcaption} %aks
\usepackage[ruled,vlined,linesnumbered]{algorithm2e}  %aks
\usepackage{amsmath} % aks
\usepackage{amsfonts} % aks
\usepackage{amssymb} % aks
\usepackage{amsthm} % aks 2
\newtheorem{problem}{Problem} % aks
\newtheorem{theorem}{Theorem} % aks
\newtheorem*{theorem*}{Theorem} % aks new
\newtheorem{remark}{Remark} % aks
\newtheorem{lemma}{Lemma} % aks
\usepackage{multicol} % aks
\usepackage{multirow} % aks
\usepackage{soul}
\usepackage{booktabs}
\usepackage{siunitx}
\usepackage[table]{xcolor}

\allowdisplaybreaks

\SetKwInput{KwInput}{Input}                % Set the Input
\SetKwInput{KwOutput}{Output}              % set the Output
\SetKwInput{KwGlobal}{Global}              % Set the Global variables
\SetKwProg{myproc}{procedure}{}{end}  

\newcommand{\ustar}{\Upsilon^{*}\xspace}
\newcommand{\tstar}{\theta^{*}\xspace}
\newcommand{\unext}{\Upsilon_{next}\xspace}

\begin{document}

\maketitle

%%%%%%%%%%%%%%%%%%%%%%%%%%%%%%%%%
\label{sec:abstract}

\begin{abstract}

The fundamental goal assignment problem for a multi-robot application aims to assign a unique goal to each robot while ensuring collision-free paths, minimizing the total movement cost.
%NP-hardness poses a tremendous challenge to designing scalable solutions for this practically significant problem.
A plausible algorithmic solution to this NP-hard problem involves an iterative process that integrates a task planner to compute the goal assignment while ignoring the collision possibilities among the robots and a multi-agent path-finding algorithm to find the collision-free trajectories for a given assignment.
This procedure involves a method for computing the next best assignment given the current best assignment.
% \st{whose corresponding collision-free trajectories turn out to be of higher cost}.
A naive way of computing the next best assignment, as done in the state-of-the-art solutions, becomes a roadblock to achieving scalability in solving the overall problem.
To obviate this bottleneck, we propose an efficient conflict-guided method to compute the next best assignment. 
% Our method keeps track of the robot-goal assignments that led to an increase in cost while resolving robot-robot collisions during the path planning stage and enforces additional constraints while computing the next best assignment to avoid generating similar assignments.
Additionally, we introduce two more optimizations to the algorithm --- first for avoiding the unconstrained path computations between robot-goal pairs wherever possible, and the second to prevent duplicate constrained path computations for multiple robot-goal pairs. 
We extensively evaluate our algorithm for up to a hundred robots on several benchmark workspaces.
The results demonstrate that the proposed algorithm achieves nearly an order of magnitude speedup over the state-of-the-art algorithm, showcasing its efficacy in real-world scenarios.

\end{abstract}
\section{Introduction}
\label{sec-intro}

% Several multi-robot applications such as warehouse management~\cite{TKDKK21,Chen21,DasNS21}, disaster response~\cite{Tian09}, precision agriculture~\cite{Gonzalez-de-Santos17}, mail and goods delivery~\cite{GrippaBWB19}, etc. require the robots to visit specific locations in the workspace to perform some designated tasks. 
% These applications lead to the fundamental goal assignment problem for multi-robot systems: 
% % Given the initial locations of a set of robots and a set of goal locations, assign each robot to a goal so that the total cost of movements by the robots to their designated goal locations is minimized. 
% %
% Given the initial locations of a set of robots and a set of goal locations, assign each robot to a goal so that the assigned paths are collision-free and the total cost of movements by the robots to their designated goal locations is minimized. 
% %
% This problem is also referred to as \emph{anonymous multi-agent path finding (AMAPF)} problem~\cite{stern2019mapf}.

A fundamental problem related to operating a multi-robot system is the \emph{anonymous multi-agent pathfinding (AMAPF)} problem~\cite{stern2019mapf}. 
In this problem, the initial locations of a set of robots and a set of goal locations are given, and the aim is to assign each robot to a goal such that the assigned paths are collision-free and the total cost or makespan of the trajectories of the robots to their designated goal locations is minimized. 
Though this problem is at the core of many multi-robot applications such as warehouse management~\cite{TKDKK21,Chen21,DasNS21}, disaster response~\cite{Tian09}, precision agriculture~\cite{Gonzalez-de-Santos17}, mail and goods delivery~\cite{GrippaBWB19}, etc., 
% the computational intractability of the problem~\cite{Yu_LaValle_2013} poses a major hindrance in developing a scalable solution for the problem.
its computational intractability~\cite{Yu_LaValle_2013} poses a major hindrance in developing a scalable solution.

To deal with the computational hardness of the AMAPF problem, two main approaches have been studied in the literature. The first approach is a decoupled one, where the task assignment problem and the path planning problem are solved sequentially. These methods~\cite{Turpin-RSS-13, 6630671, DBLP:journals/arobots/TurpinMMK14} are scalable but generally considers only makespan as the objective function and do not provide any guarantee of optimality. The second approach integrates the task planner and the motion planner, and, through their iterative interaction, find the optimal collision-free paths for the robots upon convergence. These methods~\cite{ ma2016optimal,cbsta}, though guarantee optimality, suffer from the lack of scalability. A major outstanding question %in the literature 
on multi-agent planning is whether there could be a scalable algorithm that can also guarantee the optimality of the collision-free paths for the agents
in terms of the total cost.

In this paper, we present an optimal algorithm for the AMAPF problem with minimization of the total cost as the objective, which is scalable to a large number of robots. 
Our algorithmic solution is inspired by the design of CBS-TA~\cite{cbsta}, which involves an iterative process that integrates a task planner based on Hungarian algorithm~\cite{kuhn1955hungarian} to compute the goal assignment while ignoring the collision possibilities among the robots and the state-of-the-art multi-agent path-finding algorithm CBS~\cite{cbs_journal} to find the collision-free trajectories for a given assignment.
This procedure is based on a method for computing the next best assignment given the current best assignment.
% whose corresponding collision-free trajectories turn out to be of higher cost.
However, CBS-TA's naive way of computing the next best assignment using the standard algorithm~\cite{murty1968algorithm, chegireddy1987algorithms} becomes the major bottleneck to achieving scalability in solving the overall problem.

In this paper, we target this prime roadblock to achieve scalability in solving an AMAPF problem optimally. Towards this goal, we propose a conflict-guided method to compute the next best assignment. Our method keeps track of the partial robot-goal assignments that led to an increase in the cost while resolving robot-robot collisions during the collision-free multi-agent path-finding stage. 
We use these cost-increasing partial robot-goal assignments to formulate constraints and use them to postpone computing new assignments containing the same partial robot-goal assignments. This postponement is enforced until there is a possibility of computing robot-goal assignments with a lower cost. Afterwards, the postponed assignments are reconsidered to find collision-free paths.

While postponement of the computation of the inefficient assignments has been our major algorithmic contribution, we also introduce two other powerful optimizations to the basic integrated task assignment and path planning algorithm. 
First, in our algorithm, we incorporate the mechanism for computing an assignment without computing the independent paths between all robot goal pairs, as introduced in~\cite{our1stPaper}. The standard assignment computation algorithm Hungarian method requires the costs for all robot-goal pairs a priori. \cite{our1stPaper} provides an algorithm that can compute an optimal assignment while computing only a few independent paths between robot-goal pairs in a demand-driven way. We adapt their mechanism in computing the assignments in such a way that we compute the paths between the robot-goal pairs minimally while computing an assignment, and the paths computed for one assignment computation can be reused in the subsequent assignment computations.
Second, we introduce a path memoization mechanism to prevent duplication in constrained path computations for multiple robot-goal pairs.

We implement our algorithm in Python and evaluate it through extensive experimentation on several standard benchmark workspaces. As a baseline, we use CBS-TA~\cite{cbsta}, which is considered to be the state-of-the-art for solving the AMAPF problem optimally. We evaluate our algorithm for up to $100$ robots. Experimental results demonstrate that our algorithm scales well with the number of robots and outperforms CBS-TA by an order of magnitude. We also evaluate the efficacy of the individual optimizations through an ablation study, which confirms that all three optimizations introduced in this paper contribute significantly to the overall performance of the algorithm.

\section{Problem}
\label{sec:problem}

\subsection{Preliminaries}
\label{problem:prelims}

\noindent
\textbf{Notations.} Let $\mathbb{N}$ represent the set of natural numbers and 
$\mathbb{R}$ represent the set of real numbers. For a natural number $n \in \mathbb{N}$, let $[n]$ denote the set $\{1, 2, 3, \ldots, n\}$.

\smallskip
\noindent
\textbf{Workspace.}
A workspace $WS$ is a $2$D rectangular space which is divided by grid lines into square-shaped cells.
Each cell can be addressed using its coordinates.
In general, a workspace consists of a set $O$ of cells that are occupied by obstacles.
Mathematically, $ WS = \langle dimension, O \rangle $, where $dimension$ is a tuple of the number of cells along the coordinate axes. 

\smallskip
\noindent
\textbf{Motion Primitives.}
In a $2$D workspace, we assume that a robot can move in $4$ directions 
% (North, South, East and West) 
(Up, Down, Left and Right) 
from its current location while respecting the workspace boundaries.
It also has the option to stay in its current cell. 
% The cost of each motion primitive is $1$ unit, which represents both the delay and the energy consumed in executing it. 
Each motion primitive incurs a cost of $1$~unit, indicative of the time required for its execution.
%\isb{Perhaps we should only talk about delay? For a ground robot, why should the cost of the stay primitive be the same as that of a movement primitive?}

\subsection{Problem Definition}
\label{problem:prob_def}

In a typical multi-robot application, robots must complete a set of tasks within a designated workspace.
These tasks are associated with specific locations within the workspace, referred to as \emph{goal locations}.
The robots must navigate to their respective goal locations to complete their assigned tasks.
A \emph{collision-free robot-goal assignment} is defined as the allocation of a unique goal to each robot, such that the resulting paths are devoid of collisions (robot-robot collisions or robot-obstacle collisions).
% The \emph{cost} of such an assignment represents the total cost of the (collision-free) movement of all the robots needed to reach their respective goal locations. 
The \emph{cost} of such an assignment represents the sum of the costs incurred by each robot to reach their respective goal locations.
An \emph{optimal collision-free robot-goal assignment} is the one that minimizes the cost.
%
% Our aim is to find which goal should be assigned to which robot so that the total cost of movement of all the robots due to the resultant assignment is minimized.
%
We now define the problem formally.

%\isb{Should we define collision-free assignment and total cost of an assignment?}

% \begin{problem}[Collision-Free Goal Assignment with Optimal Total Cost]
% \begin{problem}[Optimal Collision-Free Goal Assignment]
\begin{problem}
\label{prob1}
Consider a multi-robot application with a grid-based workspace $WS$, the set $S$ of start locations of the robots, and the set $F$ of goal locations as inputs.
% Let $R = |S|$ and $G = |F|$ denote the number of robots and goals, respectively.
%
Each robot can be assigned to at most one goal, and each goal can be served by at most one robot. 
% Let $\mathtt{cost}(i, j)$ denote the cost of movement between $s_i \in S$ and $f_j \in F$, where $i \in [R]$ and $j \in [G]$. 
Find a collision-free robot-goal assignment for the multi-robot application such that the cost of the assignment is minimized.
\end{problem}

% \begin{problem}
% \label{prob1}
% Consider a multi-robot application in a grid-based workspace $WS$, where the set $S$ of start locations of the robots and the set $F$ of goal locations are given as inputs.
% Let $R = |S|$ and $G = |F|$ denote the number of robots and goals, respectively.

% Each robot can be assigned to at most one goal, and each goal can be served by at most one robot. 
% Let $\mathtt{cost}(i, j)$ denote the cost of movement between $s_i \in S$ and $f_j \in F$, where $i \in [R]$ and $j \in [G]$. 
% Find a robot-goal assignment for the multi-robot application such that the total cost of the collision-free movement of all the robots to reach their respective goals is minimized.
% \end{problem}

\subsection{Example}
\label{problem:example}

% \begin{figure}[!ht]
%         \centering
%         \subfloat[]{\includegraphics[scale=0.6]{images_ex/1_ws.pdf} \label{example:prob-ws} }
%         \hspace*{0.7cm}
%         \subfloat[]{\raisebox{0.60cm}{\includegraphics[scale=0.65]{images_ex/1_b_problem_statement.pdf} } \label{example:prob-output} }   
        
%     \caption{ An example problem }
%     \label{fig:example-prob}
% \end{figure}

\begin{figure}%[!ht]
    % first subfigure
    \begin{subfigure}{0.234\textwidth}
    \centering
        \includegraphics[scale=0.4]{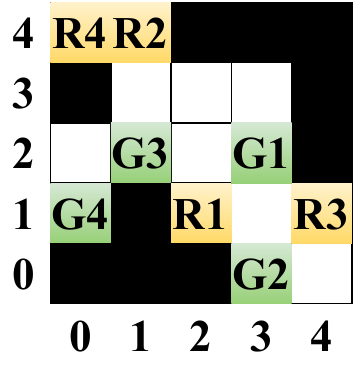}
        \caption{Workspace}
        \label{fig:ws}
    \end{subfigure}
    \hfill
    % second subfigure
    \begin{subfigure}{0.234\textwidth}
    \centering
        \raisebox{0.3cm}{\includegraphics[scale=0.45]{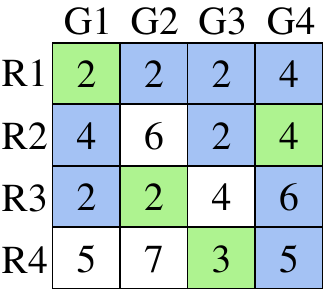}}
        \caption{Final Assignment}
        \label{fig:output}
    \end{subfigure}

    % outer figure attributes
    \caption{ An example problem }
    \label{fig:example-prob}
\end{figure}

\noindent 
Consider a multi-robot application in the $5\times5$ workspace as illustrated in Figure~\ref{fig:example-prob}(\subref{fig:ws}). 
It consists of four robots (R1,~R2, R3, and R4) and four goals (G1, G2, G3,~and~G4).
The black-colored cells denote the obstacles.
We aim to determine an assignment of robots to goals that optimizes the total cost of movement of all the robots while guaranteeing collision-free paths.
Figure~\ref{fig:example-prob}(\subref{fig:output}) displays a cost matrix associated with the multi-robot application, which is an outcome of our approach. 
Here, the colored cells reflect the actual costs, while the white cells display the heuristic costs (computed using the Manhattan distance) corresponding to the robot-goal pairs.
The matrix cells in green color depict the assignment that has an optimal total cost while ensuring collision-free paths.
In the next section, we present an efficient approach to achieve this aim.

\section{Algorithm}
\label{sec:algo}

%%%%%%%%%%%%%%%%%%%%%%%%%%%%%%%%%%%%%%%%%%%%%%%%%%%%%%%%%%%%%%%%%%%%%%%%%%%%%%%%%%%%%%%%%%%%%%%%%%%%%%%%%%%%%%%%%%%%%%%%%%%%
% <----- ALGO 1 version 1 (oct 31, 2023) ----->
%%%%%%%%%%%%%%%%%%%%%%%%%%%%%%%%%%%%%%%%%%%%%%%%%%%%%%%%%%%%%%%%%%%%%%%%%%%%%%%%%%%%%%%%%%%%%%%%%%%%%%%%%%%%%%%%%%%%%%%%%%%%

\begin{algorithm*}[!htb] 
    \DontPrintSemicolon
    \caption{ Optimal Collision-Free Goal Assignment for Multi-Robot Systems }
    \label{algo1}
    \fontsize{9pt}{10pt}\selectfont
    \BlankLine
    \BlankLine
    \KwGlobal{ $OPEN$, $MemPath$, $root\_gen$ }
    \begin{multicols}{2}
    
\vspace*{1px}

% 1st procedure BEGINS -------------------->

\myproc{$\mathtt{solve\_goal\_assignment}${ ( $WS$, $S$, $F$ )}} 
{
    $ root\_gen \gets 0 $  \;

    % $ R = |S| $, \ \ $ G = |F| $\; 

    % computing heuristic costs
    \For{$ i = 1$ \KwTo $|S|$  \label{algo1:compute_hcosts}  }
    {
        \For{$ j = 1$ \KwTo $|F|$}
        {
            $C(i)(j) = \mathtt{get\_Manhattan\_distance}(S(i), F(j))$\;      \label{algo1:calc_Manhattan_dist}
            % $T(i)(j) = \text{`$h$'}$, \label{algo1:init_type_of_cost} \ \ 
            % $P(i)(j) = [\ ]$\; \label{algo1:init_path}
        }
    }
    
    \smallskip
    $ \langle C, \theta \rangle \gets \mathtt{get\_first\_assignment}(C)  $   \;  \label{algo1:first_asgn}
    \smallskip
    
    \If{ $\theta$ is not $None$ }
    {
        $ latest\_root\_cost \gets \mathtt{ create\_root }(\theta)  $   \;  \label{algo1:create_first_root}
        % $ latest\_root\_cost \gets \Upsilon.cost  $  \;
        % $ \text{insert} \ \  \Upsilon  \ \  \text{to}  \ \  OPEN $  \;
    }
    
    \medskip
    \While{$OPEN$ is not empty   \label{algo1:while_begins}}
    {
        \smallskip
        % $ \PD \gets \text{node with lowest solution cost from} \ \ OPEN  $  \;
        $ \ustar \gets \text{extract best node from} \ OPEN  $  \;   \label{algo1:best_node_from_open}
        \smallskip

        $ conflict \gets \mathtt{ get\_first\_conflict }(\ustar.path) $  \;  \label{algo1:find_conflict}

        \If{ $conflict$ is $None$ \label{algo1_no_conflict1} }  
        {
            % \Return $ \langle  \ustar.path, \ustar.cost \rangle $  \;
            \Return $ \ustar.path $  \;
        }  \label{algo1_no_conflict2}

        $ edge1 = ( conflict.robot_1, \ustar.M(conflict.robot_1) )  $  \;  \label{algo1:acc_conf_begin}
        $ edge2 = ( conflict.robot_2, \ustar.M(conflict.robot_2) )  $  \;
        
        $ Acc\_conf(\ustar.root\_id).\mathtt{add}( edge1, edge2 )  $  \;  \label{algo1:acc_conf_end}

        \smallskip
        $ Constraints \gets \mathtt{ create\_constraints }(conflict)  $  \;  \label{algo1:create_constraints}
        
        %invoke CBS
        $ \mathtt{create\_child\_nodes}( \ustar, Constraints, WS, S, F )  $  \;   \label{algo1:create_child_nodes}

        \medskip
        % next_P processing
        $ \unext \gets \text{peep at next best node within} \ OPEN  $  \;  \label{algo1:peep}  \smallskip

        % processing only for CBS nodes
        \If{ $\unext.root$ = $False$  \label{algo1:cbs_node} }
        {
            \smallskip
            % registering (storing) conflict and cost jump
            \If{ $\unext.cost > \unext.parent.cost $ and $\unext.conf\_reg = False$  }
            {
                \smallskip
                $ cost_{inc} = \unext.cost - \unext.root\_cost $ \;

                $ conf\_rec = \langle Acc\_conf( \unext.root\_id ), cost_{inc} \rangle $  \;
                
                \smallskip
                \If{ $ conf\_rec $ not in $Conflict\_rec $ }
                {
                    $ Conflict\_rec.\mathtt{add}( conf\_rec ) $  \;
                    % $ modified = True $  \;
                }

                $ \unext.conf\_reg \gets True $  \;

            }  \label{algo1:conf_reg_ends}

            \smallskip
            % next asgn computation
            \If{ $\unext.cost > latest\_root\_cost $  \label{algo1:next_best_asgn_begin} }
            {
                \smallskip
                $ \langle C, \theta \rangle \gets \mathtt{get\_next\_assignment}(C, $  $\hspace*{2.8cm} Conflict\_rec )  $   \;
                \smallskip
            
                \If{ $\theta$ is not $None$ }
                {
                    $ latest\_root\_cost \gets \mathtt{ create\_root }(\theta)  $   \;  \label{algo1:create_next_root}
                    % $ \langle root\_gen, \Upsilon \rangle \gets \mathtt{ create\_root }(\theta, root\_gen)  $   \;  \label{algo1:create_next_root}
                    % $ latest\_root\_cost \gets \Upsilon.cost  $  \;
                    % $ \text{insert} \ \  \Upsilon  \ \  \text{to}  \ \  OPEN $  \;
                    
                    % $ \Rnext \gets new \ cbsta \ node $  \; 
                    % $ \Rnext.root \gets True $   \;
                    % $ \Rnext.parent \gets None $  \;
                    % $ \Rnext.constraints \gets \emptyset  $  \;
                    
                    % $ \Rnext.M \gets \adollar.M $  \;
                    % $ \Rnext.solution \gets \adollar.solution $  \;
                    % $ \Rnext.cost \gets \adollar.cost $  \;
                    % $ \Rnext.root\_cost \gets \adollar.cost $  \;

                    % \smallskip
                    % $ root\_gen \gets root\_gen + 1 $  \;
                    % $ \Rnext.root\_id \gets root\_gen $  \;
                    % $ latest\_root\_cost \gets \Rnext.cost $  \;
            
                    % \smallskip
                    % $ \text{insert} \ \  \Rnext  \ \  \text{to}  \ \  OPEN $  \;
                                    
                }
            }       \label{algo1:next_best_asgn_end}
        }

    }  \label{algo1:while_ends} %while OPEN 

} 

% end procedure 1 <-------- PROC 1 ENDS -------->

\columnbreak
\vspace*{1px}

% new procedure to create a root node (oct 12, 2023)

\myproc{$\mathtt{create\_root}${ ( $\theta$ ) }  \label{algo1:create_root_begin} } 
{
    $ root\_gen \gets root\_gen + 1 $  \;
    
    \smallskip
    $ \Upsilon \gets new \ search \ node $   \; 
    $ \Upsilon.root \gets True $  \;
    
    $ \Upsilon.root\_id \gets root\_gen $  \;
    $ \Upsilon.parent \gets None $  \;
    $ \Upsilon.constraints \gets \emptyset  $  \;

    \smallskip
    $ \Upsilon.M \gets \theta.M $, \label{algo1:create_root_match}  \ \ 
    $ \Upsilon.path \gets \theta.path $  \;
    
    $ \Upsilon.cost \gets \theta.cost $,  \ \ 
    $ \Upsilon.root\_cost \gets \theta.cost $  \;
    % $ latest\_root\_cost \gets \R.cost $  \;

    \smallskip
    $ \text{insert} \ \  \Upsilon  \ \  \text{to}  \ \  OPEN $  \;

    % \Return $ \langle root\_gen, \Upsilon \rangle $  \;
    \Return $ \Upsilon.cost $  \;
}  \label{algo1:create_root_end}

\BlankLine
\BlankLine

% Procedure CREATE-CHILD-NODES - BEGINS
\myproc{$\mathtt{create\_child\_nodes}${($\ustar$, $Constraints$, $WS$, $S$, $F$)}}
{
    \smallskip
    \For{each robot $r \in Constraints$}
    {
        $ \Upsilon \gets new \ search \ node $  \; 
        $ \Upsilon.root \gets False $   \; 
        $ \Upsilon.root\_id \gets \ustar.root\_id $  \;
        $ \Upsilon.root\_cost \gets \ustar.root\_cost $  \;
        $ \Upsilon.parent = \ustar $  \;

        \smallskip
        $ \Upsilon.constraints \gets \ustar.constraints \ + $ \hspace*{2.34cm} $ Constraints(r)  $  \;

        $ \Upsilon.M \gets \ustar.M $,  \ \ 
        $ \Upsilon.path \gets \ustar.path $  \;
        $ g = \Upsilon.M(r) $, \ \  $ c = \Upsilon.constraints(r) $   \;

        % CBS share
        \smallskip
        % \If{ $ MemPath(r)(g)(c)$ does not exist  \label{algo1:memo_begin} }
        % {
        %     \smallskip
        %     $ MemPath(r)(g)(c) \gets \mathtt{ASTAR}(WS, S(r), F(g), c) $
        % }

        % $ \Upsilon.path(r) \gets MemPath(r)(g)(c)  $  \;  \label{algo1:memo_end}
        $ \Upsilon.path(r) \gets \mathtt{get\_constrained\_path}(r, g, c)  $  \;  
        
        $ \Upsilon.cost \gets \ustar.cost - \mathtt{compute\_cost}(\ustar.path(r)) + \hspace*{2.63cm} \mathtt{compute\_cost}(\Upsilon.path(r))  $  \;

        $ \text{insert} \ \  \Upsilon  \ \  \text{to}  \ \  OPEN $  \;

    }
}

\BlankLine \BlankLine

% Procedure GET-CONSTRAINED-PATH - BEGINS
\myproc{$\mathtt{get\_constrained\_path}${($r$, $g$, $c$)}  \label{algo1:memo_begin}  }
{
    \smallskip
    \If{ $ MemPath(r)(g)(c)$ does not exist  }
    {
        \smallskip
        $ MemPath(r)(g)(c) \gets \mathtt{ASTAR}(WS, S(r), F(g), c) $
    }

    \Return $ MemPath(r)(g)(c) $  \;  \label{algo1:memo_end}

}

\end{multicols}
\BlankLine
\end{algorithm*}

%%%%%%%%%%%%%%%%%%%%%%%%%%%%%%%%%%%%%%%%%%%%%%%%%%%%%%%%%%%%%%%%%%%%%%%%%%%%%%%%%%%%%%%%%%%%%%%%%%%%%%%%%%%%%%%%%%%%%%%%%%%%
% NEXT BEST ASSIGNMENT in separate ALGORITHM block 
%%%%%%%%%%%%%%%%%%%%%%%%%%%%%%%%%%%%%%%%%%%%%%%%%%%%%%%%%%%%%%%%%%%%%%%%%%%%%%%%%%%%%%%%%%%%%%%%%%%%%%%%%%%%%%%%%%%%%%%%%%%%

\begin{algorithm}[!htb]  % replaced [t] temporarliy [!htb]
    \DontPrintSemicolon
    \caption{Conflict-Aware Assignment}
    \label{algo2}
    % \begin{multicols}{2}

    \fontsize{9pt}{9.85pt}\selectfont
    
    \BlankLine
    \KwGlobal{ $ASGN\_OPEN$, $ASGN\_POST$ }

% 1st procedure BEGINS <-------------------->
\BlankLine

% Procedure GET-FIRST-ASSIGNMENT - BEGINS
\myproc{$\mathtt{get\_first\_assignment}${ ( $C$ ) }}
{

    $ \theta \gets new \ asgn \ node $   \;

    $ \theta.O \gets \emptyset $,   \hspace*{0.4cm}
    $ \theta.I  \gets \emptyset $   \;

    % \smallskip
    $ \langle C, M, path \rangle \gets \mathtt{compute\_assignment}(C, \theta.O, \theta.I )  $  \;  \label{algo2:solve-first-asgn}
    % \smallskip

    $ \theta.M \gets M $,   \ \
    $ \theta.path \gets path $  \;
    $ \theta.cost \gets \mathtt{compute\_cost}( path ) $   \;

    % \smallskip
    % \For{\text{each robot key} $m$ \text{in} $M$ }  
    % {
        % $ \theta.path(m) \gets path(m)(M(m)) $  \tcp{\color{blue} CHECK} 
    % }

    $ \text{insert} \ \  \theta  \ \  \text{to}  \ \  ASGN\_OPEN $  \;

    \Return $ \langle C, \theta \rangle $ \;  % can remove PATH from here

}

% \BlankLine
\BlankLine

\myproc{$\mathtt{get\_next\_assignment}${($C$, $Conflict\_rec$)}}
{
    % \smallskip
    % \isb{What is $\mathcal{G}$?}
    % \is{Could we present this algorithm without using $ASSIGN\_OPEN$?}

    $ \tstar \gets \text{extract least cost node from} \ ASGN\_OPEN  $  \;

    \If{ $ \tstar $ does not exist  }
    {
        \Return $ None $ 
    }

    $ \mathcal{R} \gets \text{derive set of robots from} \ \tstar.M  $  \; 
    % \smallskip
    
    % \If{ $modified = True$ }
    % {
    $ \text{sort } Conflict\_rec \text{ in decreasing order of } cost_{inc} $ \;   \label{algo2:sort_conflicts}
    % $ modified \gets False $ \;
    % }
    % \smallskip
    
    % creating custom order for looping
    \tcp{create custom order for looping}
    \For{ each tuple $t \in Conflict\_rec $  \label{algo2:custom_order_begin}  }
    {
        \For{each edge $(r, g) \in t.Acc\_conf $}
        {
            \If{ $ r \notin ordered\_robots $  }
            {
                $ ordered\_robots.\mathtt{add}(r) $  \; 
                $ \mathcal{R}.\mathtt{remove}(r) $  \;
            }
        }
    }

    select any robot $r\prime \in \mathcal{R}$ \;
    
    \For{ each remaining robot $r \in (\mathcal{R} - r\prime) $}
    {
        $ ordered\_robots.\mathtt{add}(r) $  \; 
    }    \label{algo2:custom_order_end}

    % \For{$ i \gets 1$ \KwTo $R - 1$}
    \For{ each robot $ r \in ordered\_robots $  \label{algo2:looping} }
    {
        \If{ $ r \notin \tstar.I $  }
        {
            $ \theta \gets new \ asgn \ node $   \;   
            $ \theta.O \gets \tstar.O  \cup \{r, \tstar.M(r) \}  $  \;
            
            % $ index = \mathtt{indexOf}(r, ordered\_robots) $   \;
            % $ \qdollar.I \gets \pdollar.I  \cup \{j, \pdollar.M(j) : j < i \}  $  \;
            $ \theta.I \gets \tstar.I  \cup \{u, \tstar.M(u) : ind_u < ind_r \}  $  \;   \label{algo2:build_include_set}

            \smallskip
            %scanning Q.I for conflicts
            $ lb_{cf\_asgn} \gets \tstar.cost $  \;    \label{algo2:postpone_begin}
            \For{ each tuple $ t \in Conflict\_rec $ }
            {
                \If{ $ t.Acc\_conf \subseteq \theta.I $ }
                {
                    $  lb_{cf\_asgn} \gets lb_{cf\_asgn} + t.cost_{inc} $  \;
                    exit loop  \;
                }
            }

            \If{ $ lb_{cf\_asgn} > \tstar.cost $  }
            {
                $ ASGN\_POST.\mathtt{add}(\langle lb_{cf\_asgn}, \theta.O, \theta.I \rangle)  $   
            }   \label{algo2:postpone_end}
            \Else
            {
                $ \langle C, M, path \rangle \gets \mathtt{compute\_assignment}(C, \theta.O, \theta.I)$  \label{algo2:solve-next-asgn} 

                % \smallskip
                \If{ $ M \ne \emptyset $  }
                {
                    $ \theta.M \gets M $,  \  \
                    $ \theta.path \gets path $   \;
                    $ \theta.cost \gets \mathtt{compute\_cost}( path ) $   \;
    
                    % \For{\text{each robot key} $m$ \text{in} $M$ }  
                    % {
                        % $ \theta.path(m) \gets path(m)(M(m)) $   
                    % }
    
                    $ \text{insert} \ \  \theta  \ \  \text{to}  \ \  ASGN\_OPEN $  \;  \label{algo2:insert-into-asgn-open} 
    
                }
            }  % Else
        }
        
    }  % For

    % 2nd leg of logic to bring back postponed assignments 

    $ best\_cost_{avail} \gets \text{min cost from } ASGN\_OPEN $  \;  \label{algo2:revoke_begin}
    $ best\_cost_{post} \gets \text{min cost from } ASGN\_POST $  \;

    % \smallskip
    \While{ $ best\_cost_{avail} > best\_cost_{post} $ }
    {
        $\Delta \gets \mathtt{min}( ASGN\_POST )  $ \;
        
        $ \theta \gets new \ asgn \ node $   \;   
        $ \theta.O \gets \Delta.O $, \ \ $ \theta.I \gets \Delta.I $  \;

        repeat lines~\ref{algo2:solve-next-asgn}-\ref{algo2:insert-into-asgn-open}   \;

        $ best\_cost_{avail} \gets \text{min cost from } ASGN\_OPEN $ \\
        $ best\_cost_{post} \gets \text{min cost from } ASGN\_POST $ 
    
    }  \label{algo2:revoke_end}
    
    $ \tstar \gets \text{select least cost node from} \ ASGN\_OPEN $  \;
    
    % \smallskip
    \Return $ \langle C, \tstar \rangle $ \;  % can remove PATH from here

}

% \end{multicols}
\end{algorithm}

In this section, we present the details of our algorithm that efficiently computes optimal collision-free robot-goal assignment for multi-robot systems. 
It learns the sets of conflicting robot-goal paths that inevitably lead to an increase in the assignment cost.
% It utilizes the learned conflicting paths to postpone the (computationally expensive) computation of assignments that include them. 
This enables the algorithm to postpone the expensive computation of several assignments that contain those learned conflicts. 

%%%%%%%%%%%%%%%%%%%%%%%%%%%%%%%%%%%%%%%%%%%%%%%%%%%%%%%%%%%%%%%%%%%%%%%%%%%%%%%%%%%%%%%%%%%%
\subsection{Algorithm Description}
\label{algo:algo_desc}
%%%%%%%%%%%%%%%%%%%%%%%%%%%%%%%%%%%%%%%%%%%%%%%%%%%%%%%%%%%%%%%%%%%%%%%%%%%%%%%%%%%%%%%%%%%%

We present our goal assignment algorithm formally in two parts, namely, Algorithms~\ref{algo1} and \ref{algo2}. 
While Algorithm~\ref{algo1} captures our complete approach, Algorithm~\ref{algo2} exclusively exhibits the procedures related to conflict-aware computation of assignments.
Our algorithm and notation draw inspiration from the state-of-the-art CBS-TA~\cite{cbsta}.
% On the high level, 
Like CBS-TA,
our approach builds a search forest to look for a goal assignment that is collision-free and has optimal cost. 
Typically, a search forest consists of more than one tree, each rooted at its respective root
denoting a specific assignment. 
An assignment can potentially contain  two or more robots whose paths conflict with each other.
Let us denote the location of a robot $r_i$ at timestep $\tau$ by $loc_{r_i}^\tau$.
A conflict between two robots $r_i$ and $r_j$ can be either a 'vertex conflict' (i.e., $\exists$ time $\tau$ : $loc_{r_i}^\tau = loc_{r_j}^\tau$) or an 'edge conflict' (i.e., $\exists$ time $\tau$ : $loc_{r_i}^\tau = loc_{r_j}^{\tau+1}$ and $loc_{r_i}^{\tau+1} = loc_{r_j}^{\tau}$).
Each tree may contain a set of nodes representing the assignment (specific to its root) with additional vertex or edge constraints that result from mitigation of conflicts.
%
% We employ the following three priority queues: 
% (a) $OPEN$, which stores the nodes belonging to the trees of the forest, prioritizing them based on the cost due to the assignment and any associated constraints, 
% (b) $ASGN\_OPEN$, which stores the nodes representing the computed assignments, and prioritizing them based on the cost of assignment, and 
% (c) $ASGN\_POST$, which stores meta data about the assignments whose computations are postponed, prioritizing them based on the lower bound cost of the postponed assignments (discussed later). 
%
We use a priority queue $OPEN$ to store the nodes belonging to the trees of the forest, prioritizing them based on the cost due to the assignment and any associated constraints.

The first procedure $\mathtt{solve\_goal\_assignment}$ acts as the \emph{main} module that invokes other procedures to solve the optimal collision-free goal assignment problem. 
It accepts the workspace $WS$, the set $S$ of start locations of the robots, and the set $F$ of goals locations as inputs.
As a first step, it computes the heuristic cost (Manhattan distance) for each robot-goal pair and
finds an initial robot-goal assignment (Algo~\ref{algo1}: Line~\ref{algo1:compute_hcosts}-\ref{algo1:first_asgn}).
This assignment has an optimal cost (based on the actual costs for the assigned robot-goal pairs), but may consists of collisions among the robots. 
On the existence of an initial assignment, the algorithm creates a corresponding root node (of the search forest) through the $\mathtt{create\_root}$ procedure (Algo~\ref{algo1}: Line~\ref{algo1:create_first_root}).
The $\mathtt{create\_root}$ procedure (Algo~\ref{algo1}: Lines~\ref{algo1:create_root_begin}-\ref{algo1:create_root_end}) also initializes the root node's attributes, including the robot-goal assignment (denoted by $M$ in Algo~\ref{algo1}: Line~\ref{algo1:create_root_match}).

Algorithm~\ref{algo1} performs a set of steps iteratively in the process to find an optimal collision-free goal assignment (Algo~\ref{algo1}: Lines~\ref{algo1:while_begins}-\ref{algo1:while_ends}). 
% It extracts the best node from $OPEN$ and checks for conflict in the associated set of assigned paths of the robots (Algo~\ref{algo1}: Lines~\ref{algo1:best_node_from_open}-\ref{algo1:find_conflict}). 
It extracts the best node from $OPEN$ and checks the associated assigned paths for conflict (Algo~\ref{algo1}: Lines~\ref{algo1:best_node_from_open}-\ref{algo1:find_conflict}). 
% Let us denote the location of a robot $r_i$ at timestep $t$ by $loc_{r_i}^t$.
% A conflict between two robots $r_i$ and $r_j$ can be either a 'vertex conflict' (i.e., $\exists$ time $t$ : $loc_{r_i}^t = loc_{r_j}^t$) or an 'edge conflict' (i.e., $\exists$ time $t$ : $loc_{r_i}^t = loc_{r_j}^{t+1}$ and $loc_{r_i}^{t+1} = loc_{r_j}^{t}$).
In the absence of conflict, the algorithm reports the current set of assigned paths as output (Algo~\ref{algo1}: Lines~\ref{algo1_no_conflict1}-\ref{algo1_no_conflict2}).  
However, in the event of a conflict, the algorithm follows the well-known `Conflict-Based Search' (CBS) framework \cite{cbs_conference, cbs_journal} to derive the constraints from the conflict and generate the child nodes (Algo~\ref{algo1}: Lines~\ref{algo1:create_constraints}-\ref{algo1:create_child_nodes}).

\smallskip
\noindent
Let us now review our key algorithmic contributions.

\subsubsection{Conflict-Aware Assignment Computation.}

With a focus on finding an optimal collision-free goal assignment, our approach, like CBS-TA, may require to compute a sequence of next best assignments apart from the optimal assignment. 
It is possible to enumerate the $K$ best solutions in various domains, including goal assignment~\cite{eppstein2016encyclopedia}.
However, instead of inefficiently generating a set of $K$ best assignments, we draw inspiration from the existing approach of~\cite{murty1968algorithm, chegireddy1987algorithms} and devise a novel method of computing the next best assignment on demand.   

Unlike CBS-TA, which computes the next best assignment immediately on the expansion of a root node (representing the current best assignment),
our approach computes the next best assignment only when continuing the search within the forest  
would lead to a cost that exceeds the largest cost among all the roots (denoting the largest assignment cost in the forest).  
We propose a novel \emph{conflict-aware assignment computation} mechanism for multi-robot systems that efficiently computes the next best assignment given the current best. 
It encompasses the following three steps.

\begin{enumerate}
    \item \textit{Conflict accumulation:}
    % Let us denote the path of a robot $r$ to its assigned goal $g$ by the edge $(r, g)$.
    Alternatively, we define the path of a robot $r$ to its assigned goal $g$ as an edge denoted by $(r, g)$.
    When there is a conflict between two robots (i.e., between two paths or edges), 
    Algorithm~\ref{algo1} cumulatively accumulates the pair of conflicting edges in the related tree's basket $Acc\_conf$ (Algo~\ref{algo1}: Lines~\ref{algo1:acc_conf_begin}-\ref{algo1:acc_conf_end}).

    \item \textit{Conflict registration:}
    After the creation of child nodes in adherence to CBS, Algorithm~\ref{algo1} \emph{peeps at (and not extract)} the tentative next best node $\unext$ within $OPEN$ (Algo~\ref{algo1}: Line~\ref{algo1:peep}). 
    If $\unext$ is a non-root node with a cost greater than its parent's cost, and if its `conflict registered' flag ($conf\_reg$) is not set, 
    the algorithm registers (\emph{learns}) the set of accumulated conflicting edges ($Acc\_{conf}$) with corresponding cost increment ($cost_{inc}$) in the conflict record ($Conflict\_rec$) (Algo~\ref{algo1}: Lines~\ref{algo1:cbs_node}-\ref{algo1:conf_reg_ends}). 
    The $conf\_reg$ flag also assists in determining the next best node by distinguishing and prioritizing the nodes having identical costs (see Remark~\ref{rem:select_best_node}).

    \item \textit{Assignment computation:} 
    If the cost of the tentative next best node appears to exceed the largest cost among the roots, we proceed to compute the next best assignment (i.e., a new root) (Algo~\ref{algo1}: Lines~\ref{algo1:next_best_asgn_begin}-\ref{algo1:next_best_asgn_end}).
    % We enhance the existing approach~\cite{murty1968algorithm, chegireddy1987algorithms} for this purpose and present it in Algorithm~\ref{algo2}.
    % Through Algorithm~\ref{algo2}, we propose a complete redesign of the established approach of~\cite{murty1968algorithm, chegireddy1987algorithms}
    % for the computation of next best assignment, given the current best assignment.
    % \color{blue}
    Drawing inspiration from the established approach of~\cite{murty1968algorithm, chegireddy1987algorithms}, 
    we propose a novel routine through Algorithm~\ref{algo2} that uses conflict-guided approach
    for the computation of next best assignment of robots to goals, given the current best assignment
    in multi-robot application.
    % \color{black}
    %
    In this algorithm, we employ the following two priority queues: 
    (a)~$ASGN\_OPEN$, which stores the nodes representing the computed assignments, prioritizing them based on the cost of assignment, and 
    (b)~$ASGN\_POST$, which stores meta data about the assignments whose computations are postponed, prioritizing them based on the lower bound cost of the postponed assignments (discussed later).
    We formulate a conflict-guided approach in which
    we postpone the computation of an assignment when it is guaranteed to include a set of conflicting edges that will unavoidably escalate the assignment's cost. 
    And the postponement should happen with the largest inevitable cost increase.
    To achieve this, Algorithm~\ref{algo2} sorts the tuples (of $\langle$set of conflicting edges, cost increment$\rangle$) in $Conflict\_rec$ in descending order of cost increments (Algo~\ref{algo2}: Line~\ref{algo2:sort_conflicts}).
    
    The existing literature considers iteration over the set of robots in a sequential order to partition the solution space of the assignments~\cite{murty1968algorithm, chegireddy1987algorithms, cbsta}.  
    In a major divergence from the previous works, we create a custom ordering of robots, which aids in maximizing the number of postponed assignments.
    We place the robots that indulge in any conflict ahead of the non-conflicting robots in the custom order (Algo~\ref{algo2}: Lines~\ref{algo2:custom_order_begin}-\ref{algo2:custom_order_end}).
    % \isb{When you mention a line number, would it be good to mention the algorithm as there are two algorithms, e.g., Algo 2: Lines 9-15?}

    Using the new custom order, we perform the conventional steps of building the `omit' and the `include' sets of edges, which are 
    used while determining an assignment (Algo~\ref{algo2}: Lines~\ref{algo2:looping}-\ref{algo2:build_include_set}).
    However, before going ahead with the assignment computation, we scan its `include' set to check for the presence of one of the learned sets of conflicting edges. 
    If a learned set is indeed a subset of the `include' set, we postpone the computation of the assignment, and save a corresponding entry in the priority queue $ASGN\_POST$ (Algo~\ref{algo2}: Lines~\ref{algo2:postpone_begin}-\ref{algo2:postpone_end}). 
    The first element $lb_{cf\_asgn}$ of the tuple saved in $ASGN\_POST$ denotes the lower bound cost of the collision-free assignment. 
    Let $\tstar$ be the current best assignment and $t$ be the tuple in $Conflict\_rec$ such that $t.Acc\_conf \subseteq$ `include' set of an assignment not yet computed. 
    Then, the lower bound cost of this assignment can be computed as:
    \begin{align}
        lb_{cf\_asgn} \gets \tstar.cost + t.cost_{inc}. \nonumber
    \end{align}
    After the assignments are computed or postponed using the custom order, we revoke the postponement and compute those assignments whose $lb_{cf\_asgn}$ are less than the minimum cost among the computed assignments stored in $ASGN\_OPEN$ (Algo~\ref{algo2}: Lines~\ref{algo2:revoke_begin}-\ref{algo2:revoke_end}). 

\end{enumerate}

% \color{blue} [REMOVE]
% Another priority queue, namely, $ASGN\_OPEN$ stores the nodes representing the computed assignments, and prioritizing them based on the cost of assignment. 
% \color{black}

\begin{remark} \label{rem:select_best_node}
In Algorithm~\ref{algo1}, two instances require determining the best node from the set of available nodes in the priority queue $OPEN$ (Algo~\ref{algo1}: Lines~\ref{algo1:best_node_from_open} and~\ref{algo1:peep}).
This selection hinges on two criteria, prioritized as follows:
(a)~cost, where lower values are favored, and
(b)~the `conflict registered' flag ($conf\_reg$), with a preference for ``False'' over ``True''.
The second criterion accelerates the process of conflict registration, meaning that conflicts are learned more rapidly.
\end{remark}

\begin{remark}
    In Algorithm~\ref{algo1}, Line~\ref{algo1:peep}, we only \emph{peep} at the next best node within the priority queue $OPEN$ without extracting it. 
    The reason behind this look-ahead approach is as follows.
    There can be numerous nodes with the same cost. 
    If conflict registration were to occur only when a node with a higher cost is selected for processing, the nodes with lower costs would need to be processed first. 
    This would delay the utilization of valuable conflict information.
    Therefore, looking ahead to the next best node allows proactive conflict registration, optimizing the conflict-aware assignment computation.
    
\end{remark}

\subsubsection{Heuristic Distance-Based Assignment Computation.}

% We apply the technique introduced in \cite{our1stPaper} to efficiently compute assignments, avoiding the need to compute the majority of the actual robot-goal paths (Algo~\ref{algo2}: Lines~\ref{algo2:solve-first-asgn} and~\ref{algo2:solve-next-asgn}). 

The goal assignment algorithm introduced in~\cite{our1stPaper} can compute a robot-goal assignment efficiently (without avoiding robot-robot collisions) by avoiding the computation of a majority of robot-goal paths. 
We adapt their solution for our problem by making the following customizations
to create a new procedure $\mathtt{compute\_assignment}$ (Algo~\ref{algo2}: Lines~\ref{algo2:solve-first-asgn} and~\ref{algo2:solve-next-asgn}).
We provide two additional inputs, namely, an `include' set and an `omit' set, to the aforesaid algorithm to further increase its efficiency.
% During the computation of an assignment, the robots and goals present in the `include' set are eliminated form the cost matrix
Since the robots and goals that are a part of the `include' set are compulsorily included in the assignment, we eliminate the rows and columns corresponding to them from the cost matrix.
% In the cost matrix, we assign an infinite cost to those robot-goal pairs that are present in the `omit' set. 
Additionally, we assign an infinite cost to the robot-goal pairs present in the `omit' set. 
An invocation of the $\mathtt{compute\_assignment}$ procedure can leave 
the cost matrix in a state where it has a mix of heuristic and actual costs.
We preserve and forward the modified state of the cost matrix in the subsequent invocations to prevent duplicate path computations for the robot-goal pairs.

\subsubsection{Path Memoization.}

% Many robot-goal paths could remain preserved between two or more assignments.
% Across multiple robot-goal assignments, several robots may retain their assigned paths.
Several robots may retain their assigned paths between two or more robot-goal assignments.
This implies that the conflicts and the corresponding constraints can remain consistent across many assignments.
To leverage this consistency, we introduce `Path Memoization', which involves caching the paths computed under specific constraints for future reuse. 
Thus, in cases where CBS requires the computation of a robot-goal path under previously encountered constraints, we utilize the memoized path instead of recomputing it (Algo~\ref{algo1}: Lines~\ref{algo1:memo_begin}-\ref{algo1:memo_end}).

% \isb{The order of these two optimization could be changed?}

% \isb{Should we not mention which part of the code is impacted by this optimization?}

\subsection{Example}

% running example

\begin{figure}[t]
    % 1st subfigure
    \begin{subfigure}{0.4\columnwidth}
    \centering
        \includegraphics[scale=0.4]{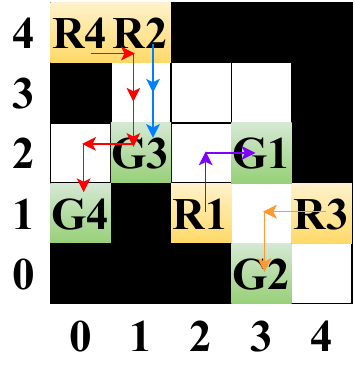} 
        \caption{Initial Assignment}
        \label{run_ex:first-asgn}
    \end{subfigure}
    % \hspace*{2.5cm}
    % 2nd subfigure
    \begin{subfigure}{0.6\columnwidth}
    \centering
        \includegraphics[scale=0.7]{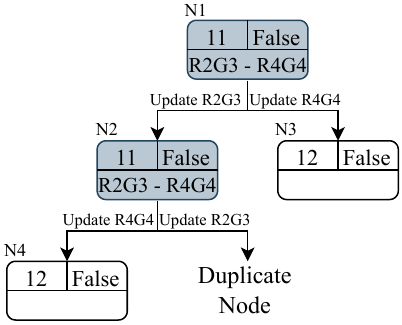}    
        \caption{CBS on Initial Assignment}
        \label{run_ex:first_cbs} 
    \end{subfigure}
    %\hspace*{1.4cm}
    \vspace{0.5cm}
    % 3rd subfigure
    \begin{subfigure}{0.42\columnwidth}
    \centering
        \includegraphics[scale=0.35]{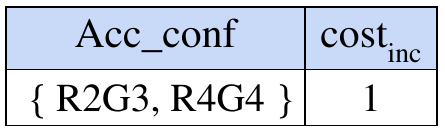} 
        \vspace{1cm}
        \caption{Conflict Registration - I}
        \label{run_ex:acc_conf}
    \end{subfigure}
    %\hspace*{0.46cm}
    % 4th subfigure
    \begin{subfigure}{0.58\columnwidth}
    \centering
        \includegraphics[scale=0.7]{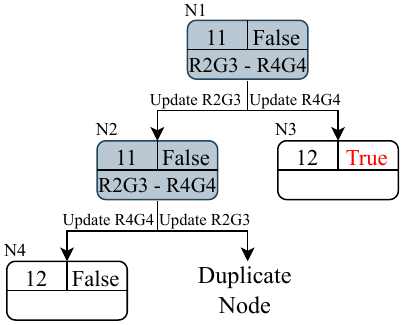} 
        \caption{Conflict Registration - II}
        \label{run_ex:register_conf}   
    \end{subfigure}
    % \hspace*{2cm}
    % % 5th subfigure
    % \begin{subfigure}{0.234\textwidth}
    % \centering
    %     \includegraphics[scale=0.4]{images_ex/6_custom_order.pdf} 
    %     \caption{CBS on First Assignment}
    %     \label{run_ex:custom_order} 
    % \end{subfigure}
    % \hfill
    % % 6th subfigure
    % \begin{subfigure}{0.234\textwidth}
    % \centering
    %     \includegraphics[scale=0.4]{images_ex/7_postpone_a.pdf} 
    %     \caption{CBS on First Assignment}
    %     \label{run_ex:postpone_a} 
    % \end{subfigure}
    % \hfill
    % 7th subfigure
    \vspace{0.5cm}
    \begin{subfigure}{0.6\columnwidth}
    % \centering
        \includegraphics[scale=0.31]{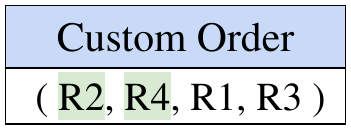} 
        \includegraphics[scale=0.31]{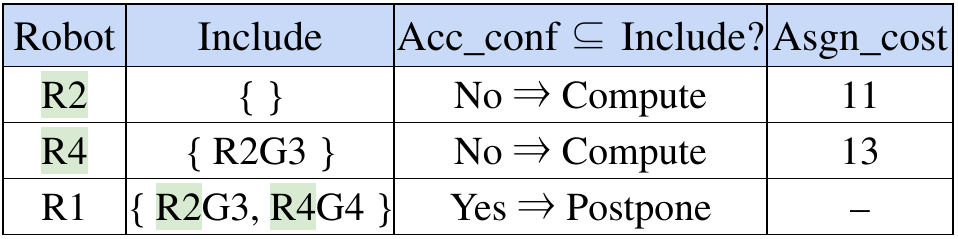} 
        \includegraphics[scale=0.31]{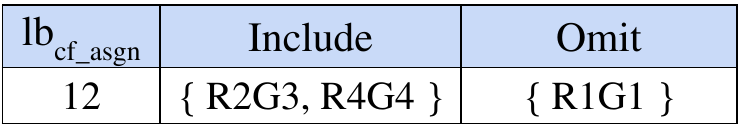} 
        \caption{Assignment Postponement}
        \label{run_ex:postpone_b}   
    \end{subfigure}
    % \hspace*{1.7cm} 
    % 8th subfigure
    \begin{subfigure}{0.35\columnwidth}
    \centering
        \includegraphics[scale=0.4]{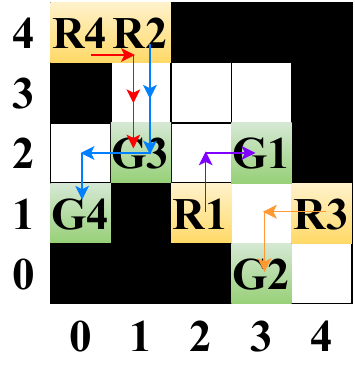} 
        \caption{Final Assignment}
        \label{run_ex:second_asgn} 
    \end{subfigure}

    % outer figure attributes
    % \caption{ Illustration of Algorithm~\ref{algo1} on the problem \hspace*{1.58cm} introduced in Figure~\ref{fig:example-prob} }
    \caption{ Applying Algorithm~\ref{algo1} on Problem in Figure~\ref{fig:example-prob}(\subref{fig:ws}) }
    \label{fig:run_ex_1}
\end{figure}

% Example description 
In Figure~\ref{fig:run_ex_1}, we illustrate the execution of Algorithm~\ref{algo1} on the multi-robot goal assignment problem introduced in Figure~\ref{fig:example-prob}.
Specifically, Figure~\ref{fig:run_ex_1}(\subref{run_ex:first-asgn}) showcases the initial robot-goal assignment that has an optimal total cost of $11$. 
The arrows visually map out the respective paths.
At timestep $3$, a conflict arises between the paths of robots R2 and R4 within cell $(1, 2)$. 
Algorithm~\ref{algo1} employs CBS to resolve this conflict, generating a search tree as depicted in Figure~\ref{fig:run_ex_1}(\subref{run_ex:first_cbs}).
The search tree originates from the root node `N1' which represents the initial assignment devoid of constraints. 
Each node within this tree encapsulates three pieces of information: 
(a)~the assignment's cost, potentially with collisions, under specific constraints, 
(b)~value of `conflict registered' flag ($conf\_reg$) which aids in distinguishing and prioritizing nodes having identical costs, and 
(c)~the first conflict in the node's associated paths. 
The grey color of nodes `N1' and `N2' signify that their processing is complete.
The tentative next best node for processing seems to be `N3' (or `N4') with a cost of $12$.  
Thus, Figure~\ref{fig:run_ex_1}(\subref{run_ex:first_cbs}) displays that state of the search tree from which the next node cannot be chosen for processing without incurring an increase in the cost of assignment. 
Having met the hurdle of cost increment, Algorithm~\ref{algo1} registers the set of accumulated conflicting paths ($Acc\_{conf}$) with corresponding cost increment ($cost_{inc}$) (Figure~\ref{fig:run_ex_1}(\subref{run_ex:acc_conf})).
The value of $cost_{inc}$ represents the disparity between the cost of the tentative next best node and that of the corresponding root.
Subsequently, the $conf\_reg$ flag of `N3' is set to ``True'' (Figure~\ref{fig:run_ex_1}(\subref{run_ex:register_conf})).

Since the probable next best node's cost ($12$) exceeds the current largest cost among root nodes ($11$), Algorithm~\ref{algo1} 
invokes the $\mathtt{get\_next\_assignment}$ procedure (expanded in Algorithm~\ref{algo2}) to compute the next best assignment. 
% proceeds to compute the next best assignment. 
% To start, the sets of conflicting paths undergo sorting in descending order of the respective cost increments. 
% Subsequently, 
It generates a custom ordering of robots (Figure~\ref{fig:run_ex_1}(\subref{run_ex:postpone_b}): first table) by keeping the robots engaged in conflicts ahead of those with non-conflicting paths. 
Consequently, R2 and R4 appear before R1 and R3 in the order. 
%
% Algorithm~\ref{algo2} initially follows the established approach outlined in prior works~\cite{murty1968algorithm, chegireddy1987algorithms} to derive the next best assignment. 
Following the established approach outlined in prior works~\cite{murty1968algorithm, chegireddy1987algorithms, cbsta}, 
the $\mathtt{get\_next\_assignment}$ procedure computes the `include' set corresponding to each robot in the custom order except the last one.
We see that the set of conflicting paths \{R2G3, R4G4\} forms a subset of robot R1's `include' set \{R2G3, R4G4\}.
Thus, the computation of assignment corresponding to the `include' and `omit' sets of R1 is postponed (Figure~\ref{fig:run_ex_1}(\subref{run_ex:postpone_b}): second table). 
%
% However, a crucial modification involves the postponement of assignment computation when a set of conflicting paths 
% forms a subset of the `include' set (Figure~\ref{fig:run_ex_1}(\subref{run_ex:postpone_b}): second table). 
The meta-data of the postponed assignment are saved for possible future reinstatement (Figure~\ref{fig:run_ex_1}(\subref{run_ex:postpone_b}): third table). 
The assignment computed against robot R2 in (Figure~\ref{fig:run_ex_1}(\subref{run_ex:postpone_b}): second table) has a cost of $11$ and it happens to be the second best assignment (Figure~\ref{fig:run_ex_1}(\subref{run_ex:second_asgn})). 
Notably, it is free from collisions, and thus, the desired solution.

Note that in Figure~\ref{fig:run_ex_1}(\subref{run_ex:first_cbs}), the paths R2G3 and R4G4 need an update twice for the same respective set of constraints. 
% Our approach memoizes the path when it is computed for the first time, allowing subsequent reuse.
With path memoization, Algorithm~\ref{algo1} computes both the paths only once, and reuses them during their second update.
Additionally, utilizing the customized heuristic distance-based assignment computation~\cite{our1stPaper}, our approach effectively circumvents the need to compute four actual costs.
The white cells in Figure~\ref{fig:example-prob}(\subref{fig:output}) reflect this optimization.

%%%%%%%%%%%%%%%%%%%%%%%%%%%%%%%%%%%%%%%%%%%%%%%%%%%%%%%%%%%%%%%%%%%%%%%%%%%%%%%%%%%%%%%%%%%%
\subsection{Theoretical Properties}
\label{algo:theo_guaran}
%%%%%%%%%%%%%%%%%%%%%%%%%%%%%%%%%%%%%%%%%%%%%%%%%%%%%%%%%%%%%%%%%%%%%%%%%%%%%%%%%%%%%%%%%%%%

\begin{comment}
Algorithm~\ref{algo1} provides the following guarantees (see Appendix~\ref{appendix:correct_a} for the proofs). 

% completeness 

\begin{theorem}
    \label{theorem:completeness}
    Algorithm~\ref{algo1} is complete.
\end{theorem}

\begin{theorem}
    \label{theorem:correctness}
    
    Algorithm~\ref{algo1} provides an optimal collision-free solution to the multi-robot goal assignment problem (Problem~\ref{prob1}),
    i.e., the path of each robot is free from collision and the sum of individual costs of all robots is minimized.

\end{theorem}
\end{comment}

% making correctness inline by bringing content from appendix
\subsubsection{Correctness.}
\label{appendix:correct_a}

Let us review the following lemmas.

\begin{lemma}
\label{lemma_cost_increase}

In a multi-robot application, if an assignment of robots to goals 
contains a set of conflicting robot-goal pairs (or edges) that escalates the assignment cost,
then any other assignment containing the same set of conflicting robot-goal pairs 
will experience at least an equal escalation in its cost. 

\end{lemma}

\begin{proof}
    For a particular robot-goal assignment in a multi-robot application, 
    consider that there exists a set of conflicting robot-goal pairs 
    that effectively increases the cost of the assignment during the resolution of conflicts. 
    It implies that there must be some robot `$r$' whose path has conflict with few of its fellow robot(s) `$J$', 
    and its path cost increases while attempting to resolve the conflict.
    An increase in the path cost of robot `$r$' means that there did not exist any alternative path for `$r$' that could have prevented the cost increment. 
    Now, in a different assignment, assume that `$r$' and its fellow conflicting robots `$J$' have same assigned goals. 
    In other words, their paths poses the same conflict, and again, there would not be any alternative path for `$r$' that could prevent the cost increase.
    Note that `$r$' can have additional conflicts due to other robots whose assignments got altered. 
    Thus, attempting to resolve these conflicts would result in a cost increase by at least the same value as in the prior assignment. 
\end{proof}

\begin{lemma}
\label{lemma_next_best_asgn}

Given the current best assignment (in terms of cost) of robots to goals in a multi-robot application,
and
a conflict record comprising of sets of conflicting robot-goal pairs (or edges) with corresponding cost increments,
Algorithm~\ref{algo2}'s $\mathtt{get\_next\_assignment}$ procedure computes the next best assignment.

\end{lemma}

\begin{proof}

    % Drawing inspiration from the established approach of~\cite{murty1968algorithm, chegireddy1987algorithms}, 
    % Algorithm~\ref{algo2}'s $\mathtt{get\_next\_assignment}$ procedure uses a conflict-guided approach
    % for the computation of next best assignment of robots to goals, given the current best assignment
    % in multi-robot application.

    % fair
    From the given current best assignment of robots to goals, it is straightforward to derive the set of robots.
    Using the given conflict record and the derived set of robots, 
    Algorithm~\ref{algo2}'s $\mathtt{get\_next\_assignment}$ (or, GNA for short) procedure creates a custom order of robots, 
    such that the robots indulging in any conflict are kept ahead of the non-conflicting ones. 
    For a set of robots and a set of goals, an assignment solution space consists of all the possible assignments.
    % Supported by the proof of \color{blue} Lemma 1 \color{black},
    Together with the current best assignment, 
    the custom order of robots is used for the disjoint partitioning of the assignment solution space, 
    rather than using a naive order (like, $r_1, r_2, r_3, \ldots, r_n$) as done in the existing literature~\cite{murty1968algorithm, chegireddy1987algorithms, cbsta}.
    It has been shown that, irrespective of the chosen order, such a partitioning covers the complete solution space minus the current best assignment~\cite{murty1968algorithm}.

    Two sets of robot-goal edges, namely the `include' set and the `omit' set are used to determine the partitions.
    % Each partition also provides a robot-goal assignment. 
    While determining the partition, an assignment adhering to the `include' and `omit' sets is computed that is locally optimal to the corresponding partition.
    The edges listed in the `include' set are supposed to be part of the assignment, whereas those in the `omit' set must be excluded. 
    % The ordering of the robots along with the given current best assignment play a role in populating these two sets. 
    The population of the `include' set takes place incrementally, and is thus, \emph{influenced} by the ordering of robots.
    The custom order ensures that the `include' set rapidly accumulates the edges consisting of conflicting robots.
    
    The GNA procedure postpones the computation of an assignment 
    if its `include' set contains a known set of conflicting edges that would inevitably escalate the assignment's cost
    during the resolution of conflicts. 
    According to Lemma~\ref{lemma_cost_increase}, such postponement of assignment computation is valid.
    A lower bound cost of the deferred assignment is derived by adding the certain cost increase (due to conflicts) to the cost of the current best assignment.

    After the partitioning of the assignment solution space is complete, during which some assignments, locally optimal to the respective partitions, get computed and the rest assignment computations get postponed, 
    the GNA procedure does an additional step.
    In this step, if a postponed assignment's lower bound cost becomes lesser than the minimum cost of the computed but unprocessed assignments, the GNA procedure revokes the postponement and computes the assignment (Algo~\ref{algo2}: Lines~\ref{algo2:revoke_begin}-\ref{algo2:revoke_end}). 
    This ensures that any eligible assignment is not missed, and thus, the GNA procedure provides the next best assignment in terms of cost.    
\end{proof}

%       B E G I N     T H E O R E M S    - - - - - - -- - - -

\noindent Following are the theorems on the correctness of Algorithm~\ref{algo1}.

% completeness 

\begin{theorem}
    \label{theorem:completeness-appendix}
    Algorithm~\ref{algo1} is complete.
\end{theorem}

\begin{proof}
    To impose constraints and mitigate collisions, Algorithm~\ref{algo1} employs CBS which has been shown to be complete~\cite{cbs_journal}.
    % Algorithm~\ref{algo1} conducts a CBS search on each root node.
    Each root node undergoes a CBS search. 
    Incrementally imposing additional constraints gradually diminishes the number of alternative paths with optimal costs, eventually leading to an inevitable increase in the path's (or node's) cost.
    During the CBS search, if the cost of the tentative next best node appears to exceed the maximum cost among the root nodes, Algorithm~\ref{algo1} proceeds to compute the next best assignment. 
    This process can continue until all possible assignments have been enumerated.
    Thus, the search is exhaustive in both goal assignment and path planning. 
\end{proof}

% soundness 

\begin{theorem}
    \label{theorem:soundness-appendix}

    Algorithm~\ref{algo1} provides an optimal collision-free solution to the multi-robot goal assignment problem (Problem~\ref{prob1}),
    i.e., the path of each robot is free from collision and the sum of individual costs of all robots is minimized.

\end{theorem}

\begin{proof}
    In order to resolve collisions among the robots of a multi-robot application,
    Algorithm~\ref{algo1} applies CBS search, which is shown to be complete and optimal~\cite{cbs_journal}, on each root node (i.e., assignment). 
    The cost of a collision-free solution is always either equal to or greater than the respective root node's cost (proof: Lemma~$1$, \cite{cbs_journal}). 
    During the CBS search in a particular tree, if the cost of the tentative next best node appears to exceed the maximum cost among the root nodes, Algorithm~\ref{algo1} proceeds to compute the next best assignment by invoking the Algorithm~\ref{algo2}'s $\mathtt{get\_next\_assignment}$ procedure.
    From the proof of Lemma~\ref{lemma_next_best_asgn}, it is evident that the assignments are computed by $\mathtt{get\_next\_assignment}$ procedure in increasing cost order. 
    And Algorithm~\ref{algo1} uses the best cost-first expansion order of processing the nodes from the priority queue $OPEN$.
\end{proof}

\subsubsection{Time complexity.}
    Problem~\ref{prob1} has been known to be NP-hard~\cite{Yu_LaValle_2013}. 
    The worst-case scenario is the one in which for a problem instance, all possible assignments get enumerated.
    Algorithm~\ref{algo1} does not claim an improvement over the baseline in terms of the worst-case time complexity.
    Nonetheless, the average case performance of Algorithm~\ref{algo1} is significantly better than the baseline, which is evident from the experimental results shared in Section~\ref{exp:results}.

\begin{figure*}[!t]
    
    % RANDOM_200 
    \begin{subfigure}{\linewidth} % Use the subfigure environment
        \centering
        \includegraphics[scale=0.125]{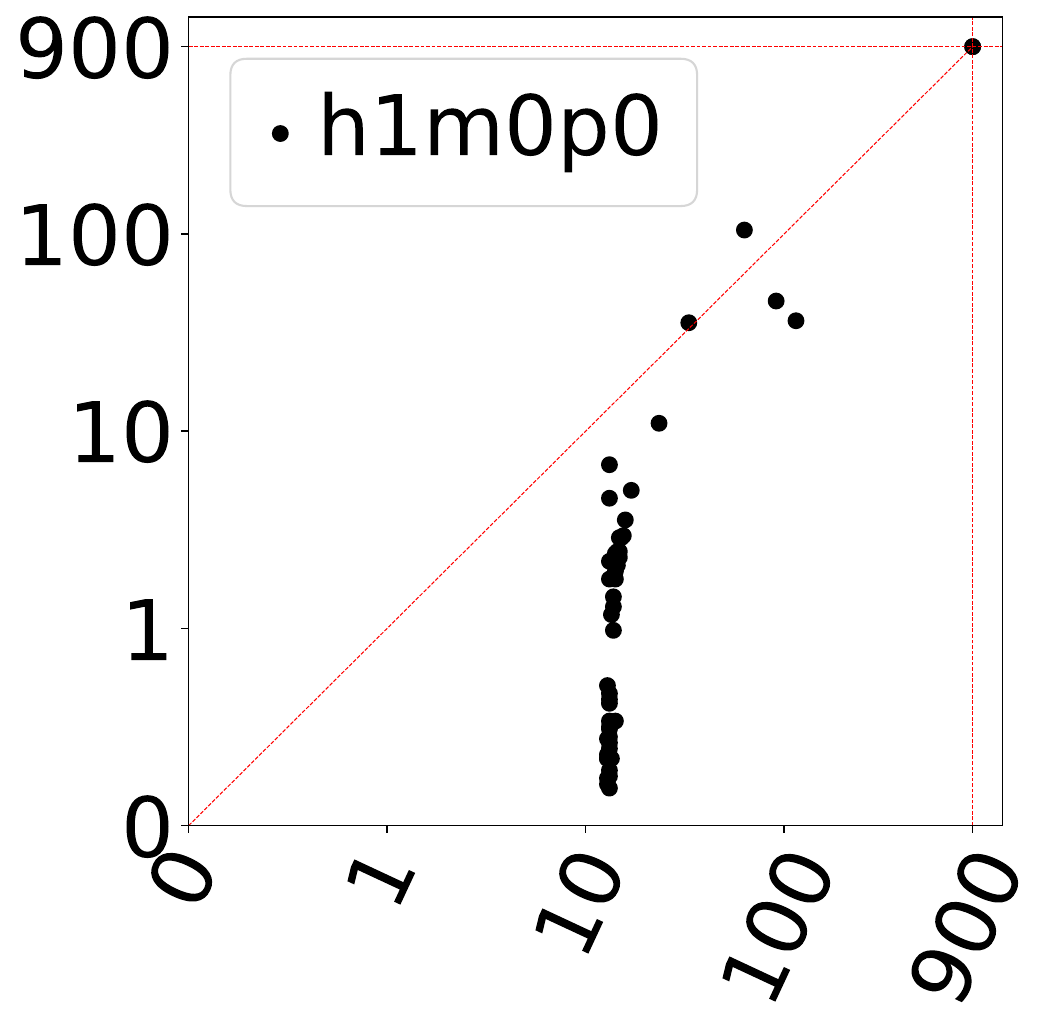}
        \includegraphics[scale=0.125]{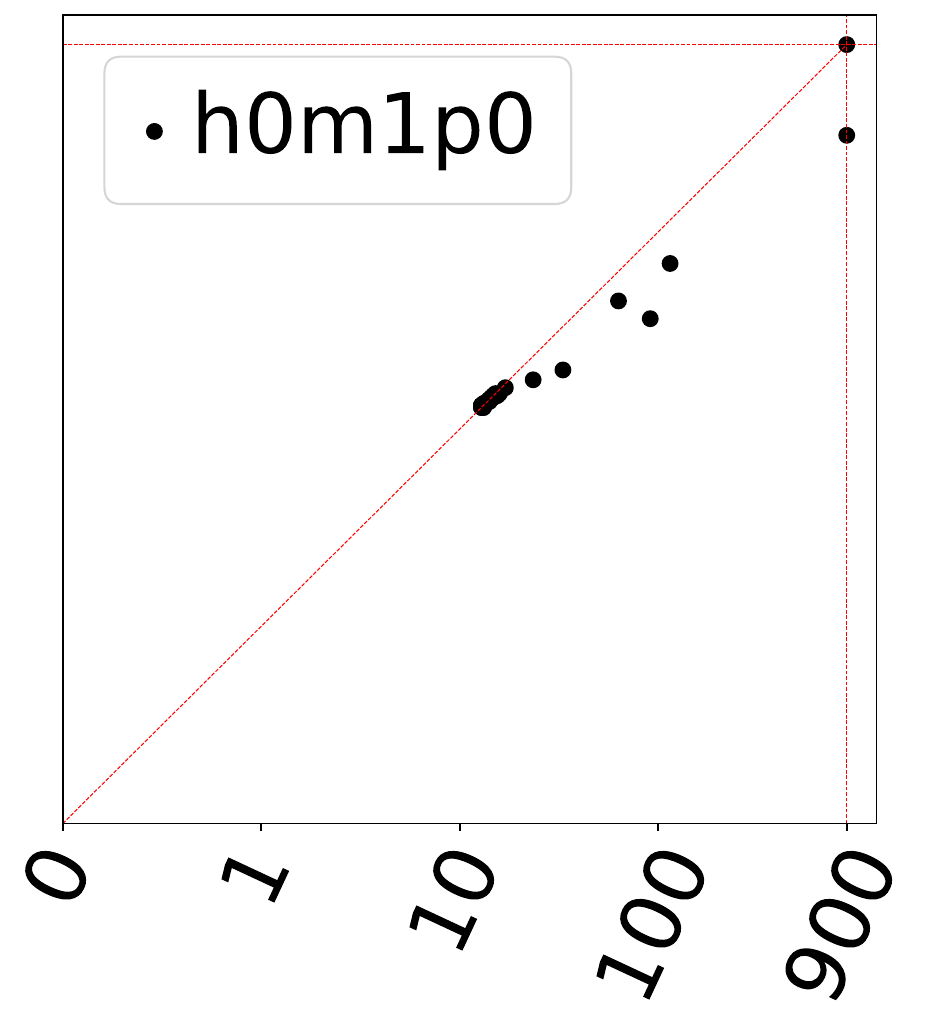}
        \includegraphics[scale=0.125]{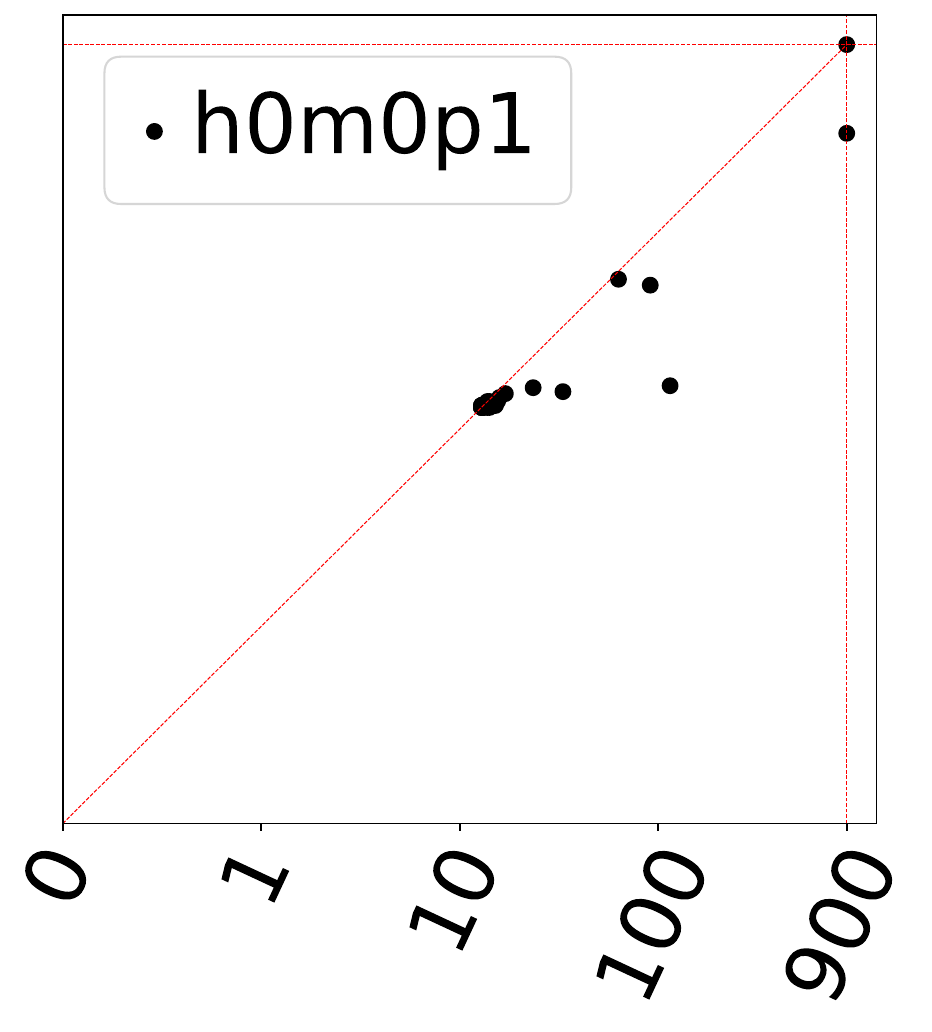}
        \includegraphics[scale=0.125]{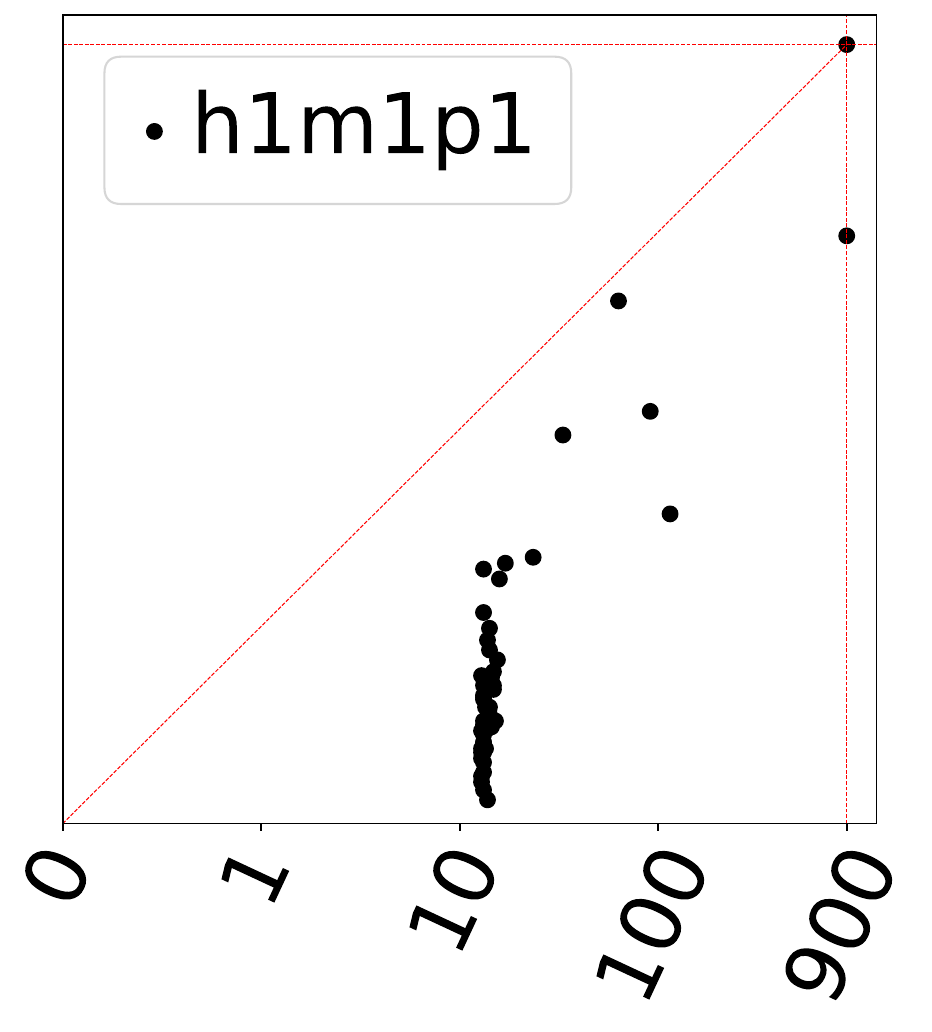} 
        \hspace*{0.2cm}
        \includegraphics[scale=0.125]{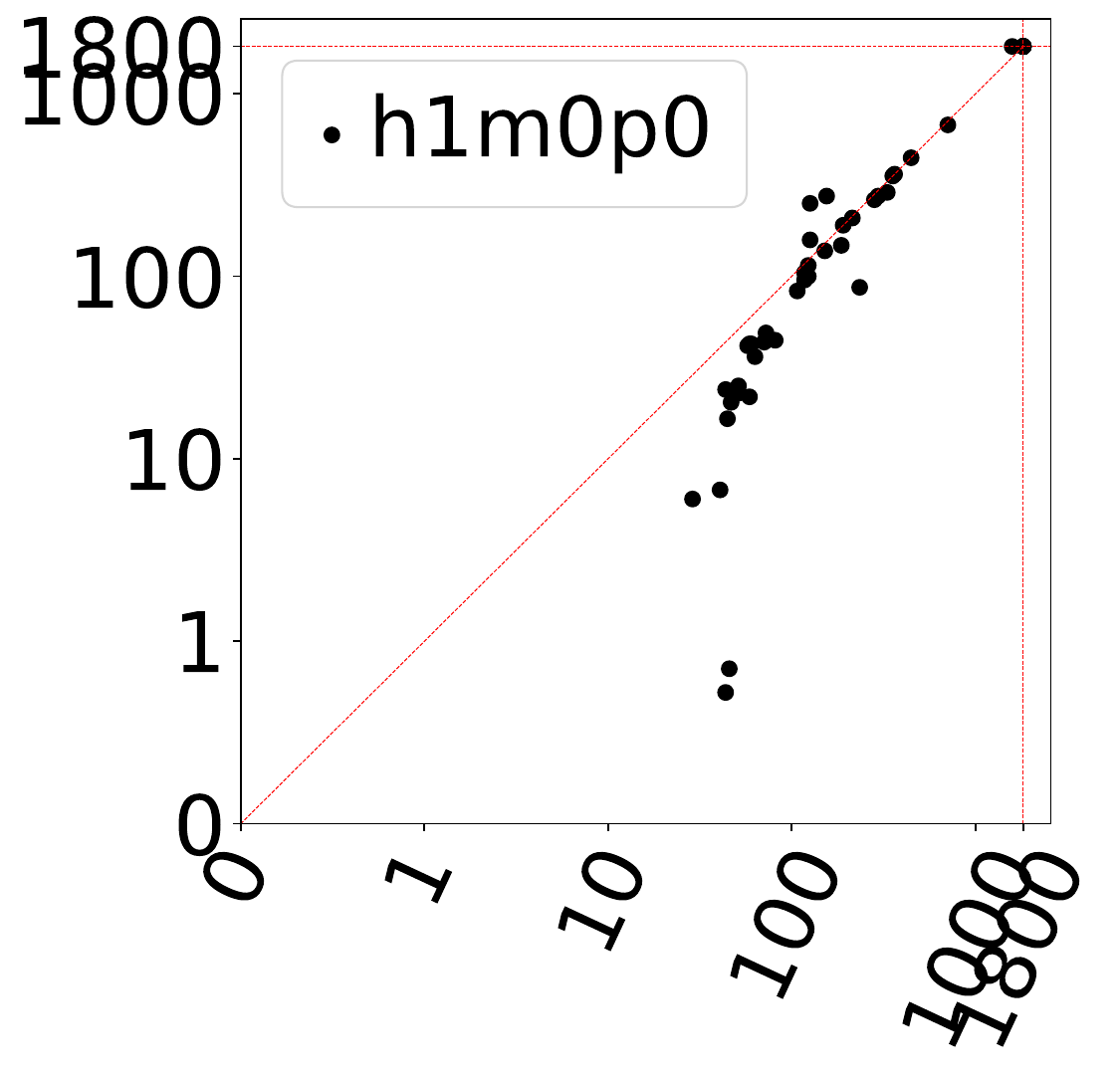}
        \includegraphics[scale=0.125]{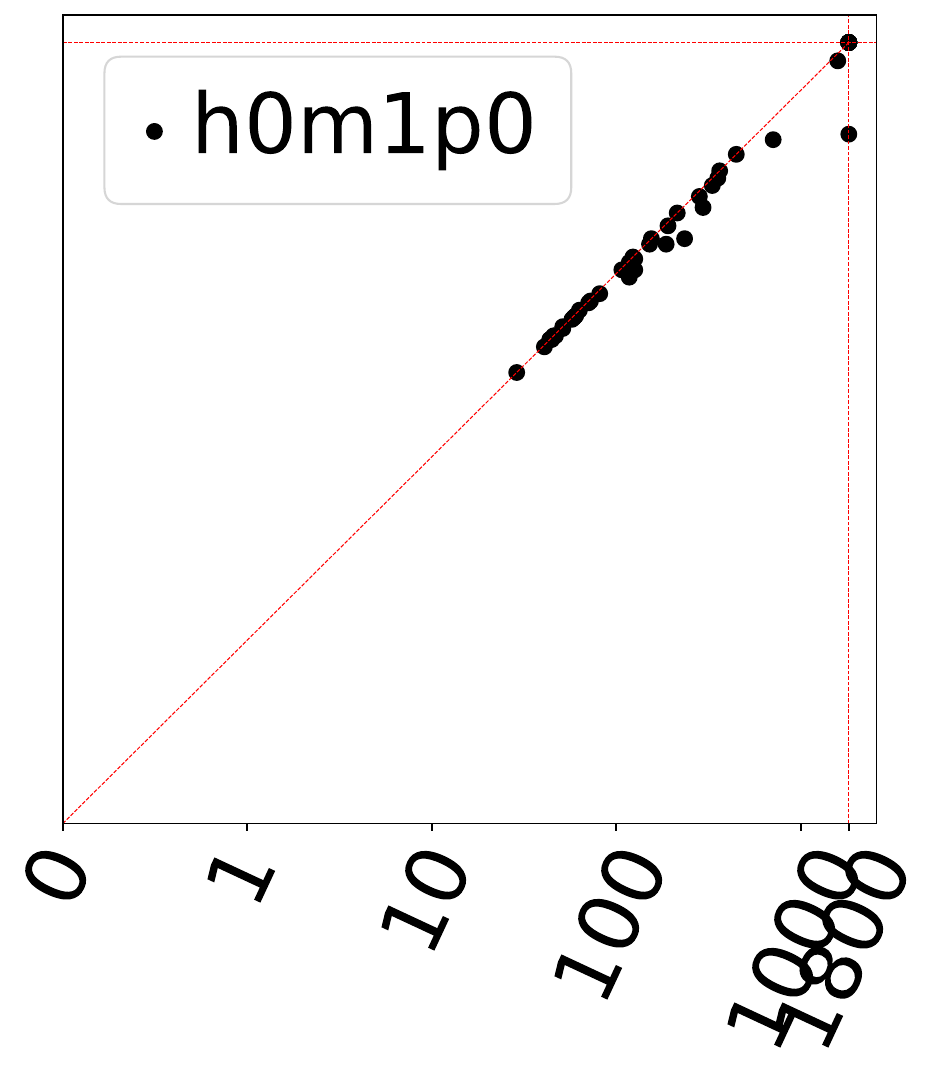}
        \includegraphics[scale=0.125]{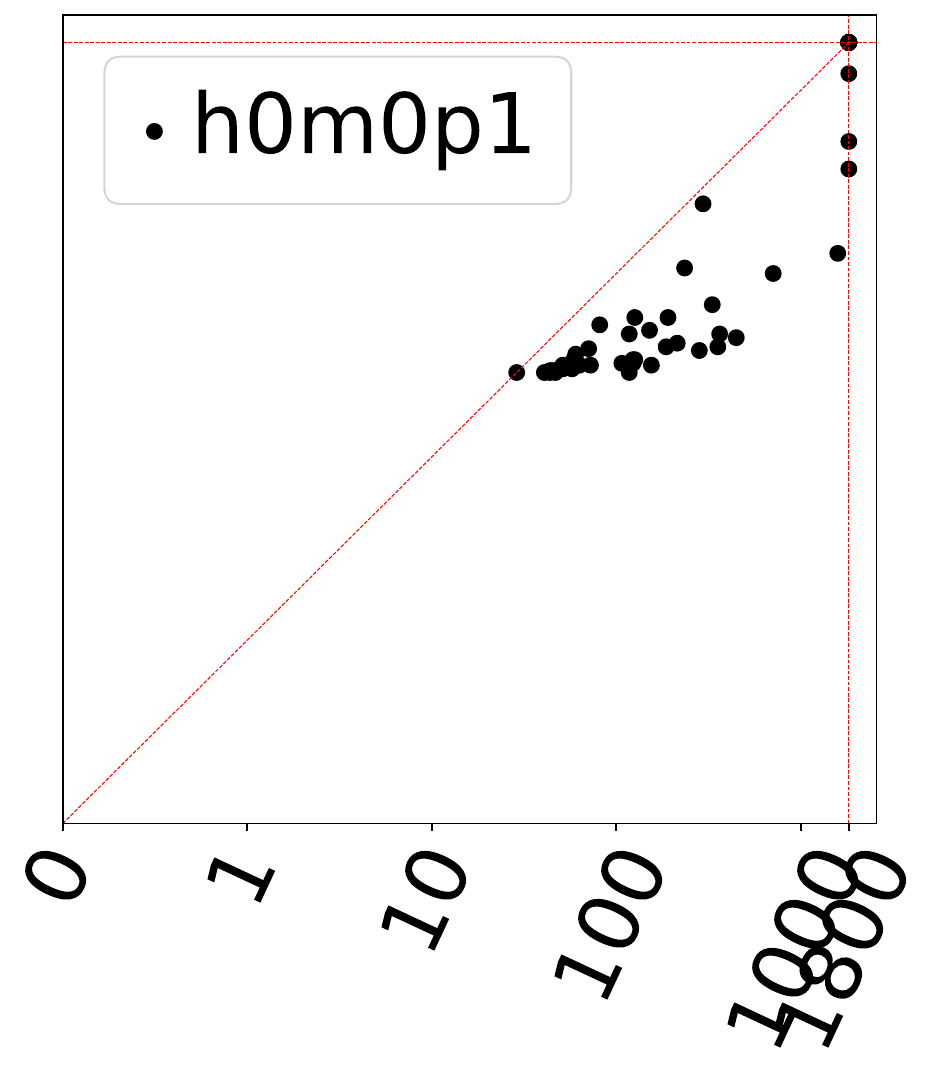}
        \includegraphics[scale=0.125]{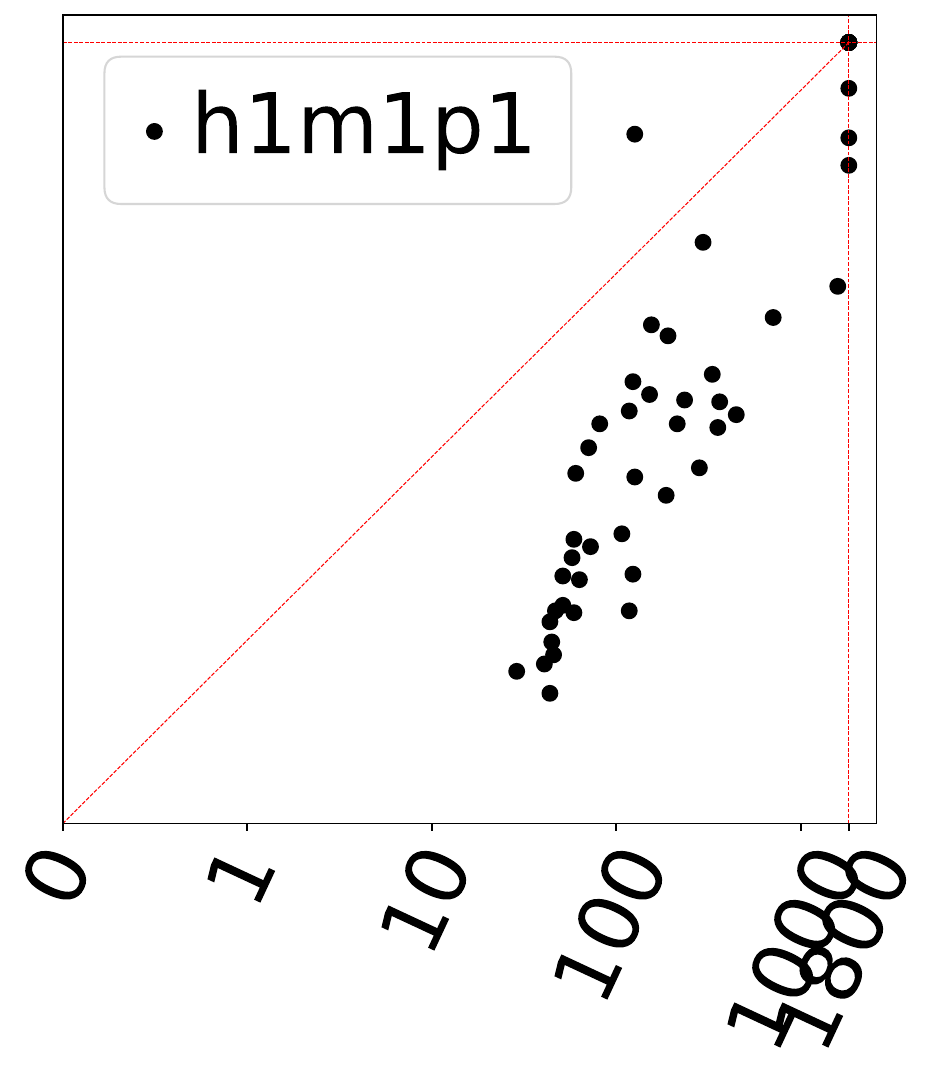}
        
        \caption{Random ($200 \times 200$)}
        \label{fig:scatter-random}
    \end{subfigure}

    % WAREHOUSE
    \begin{subfigure}{\linewidth} % Use the subfigure environment
        \centering
        \includegraphics[scale=0.125]{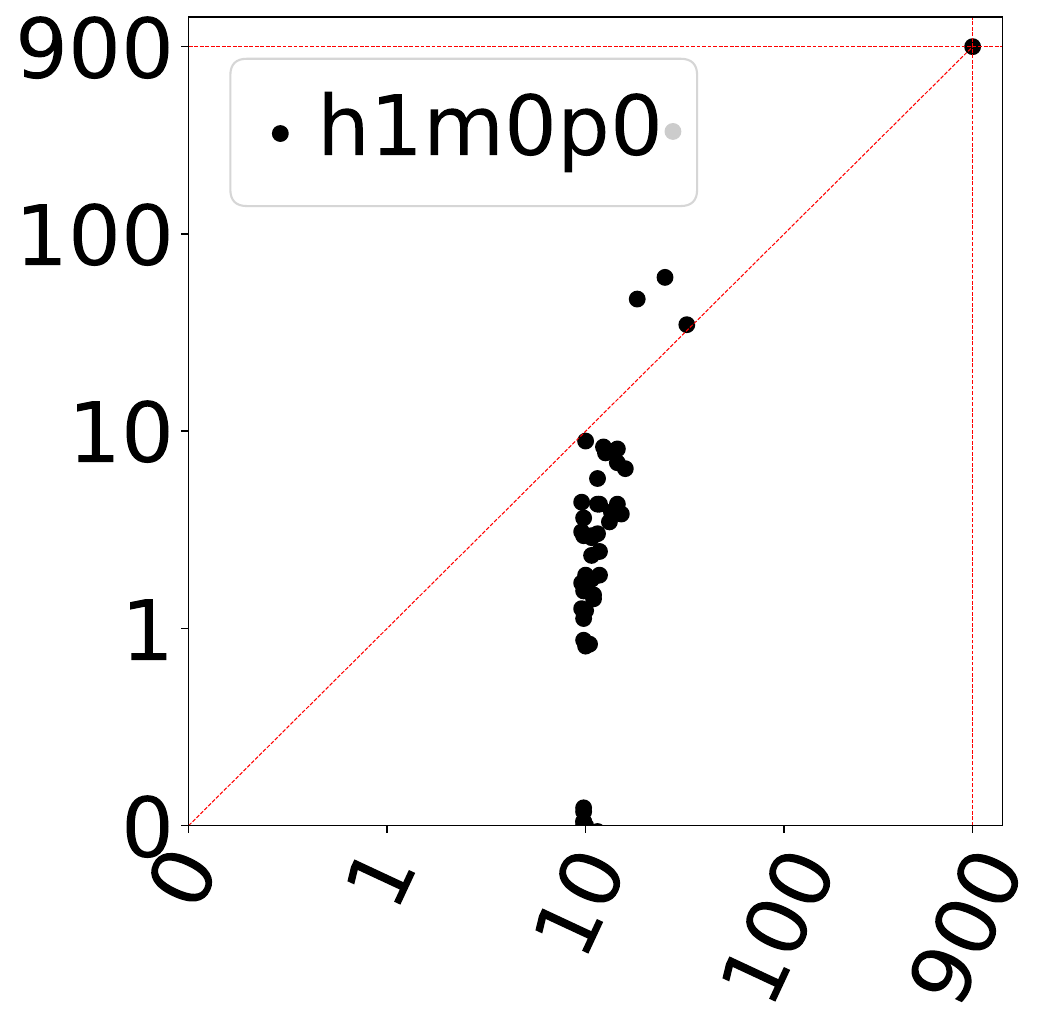}
        \includegraphics[scale=0.125]{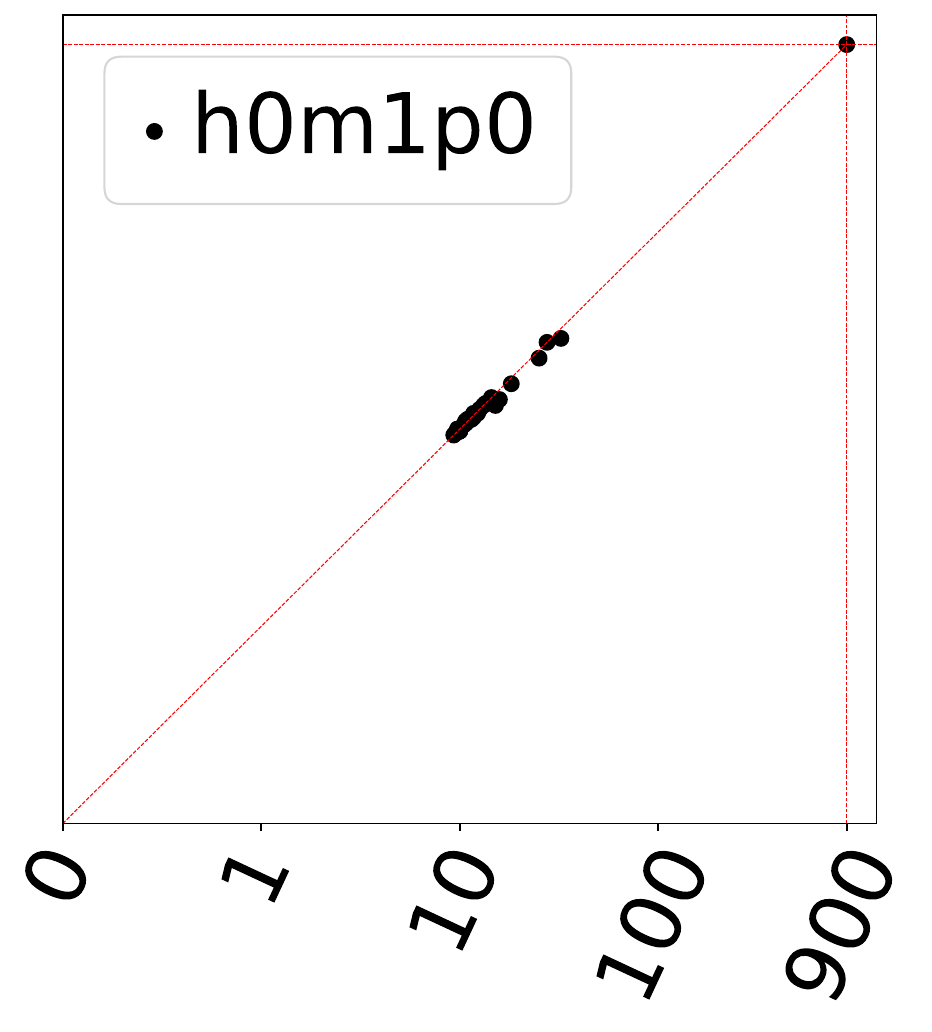}
        \includegraphics[scale=0.125]{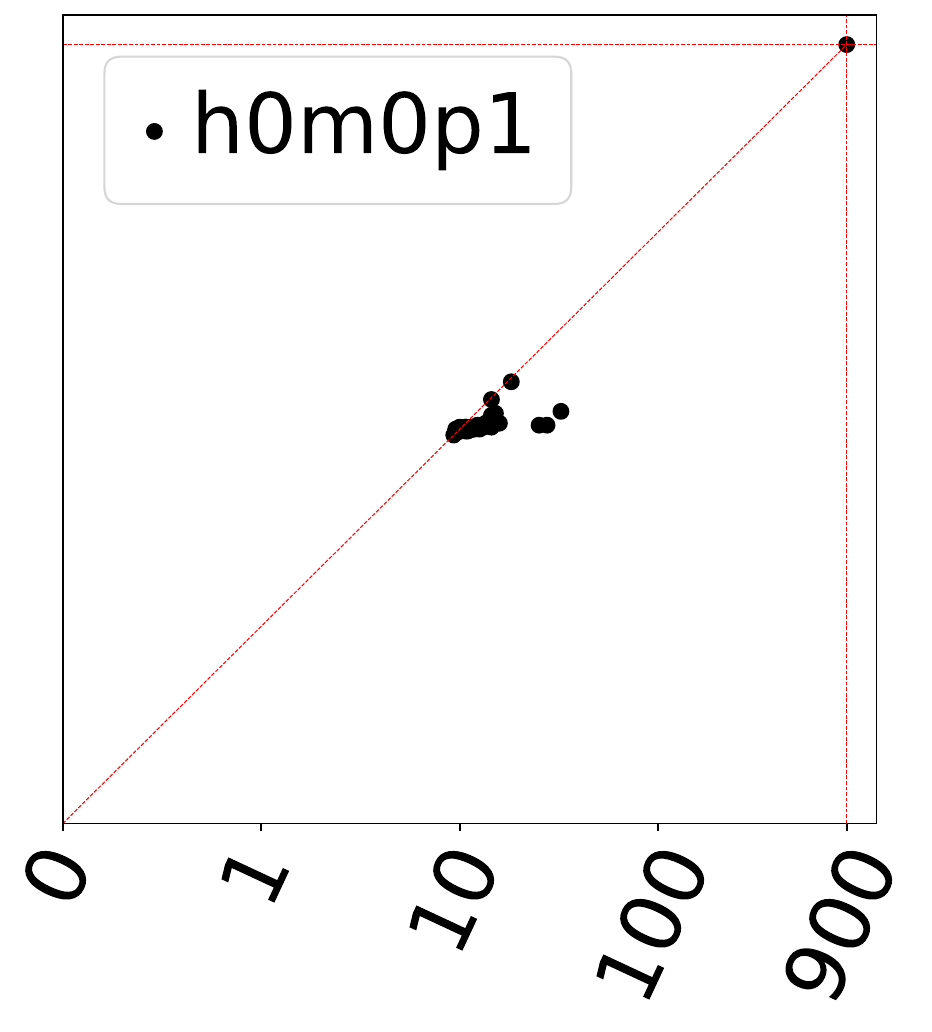}
        \includegraphics[scale=0.125]{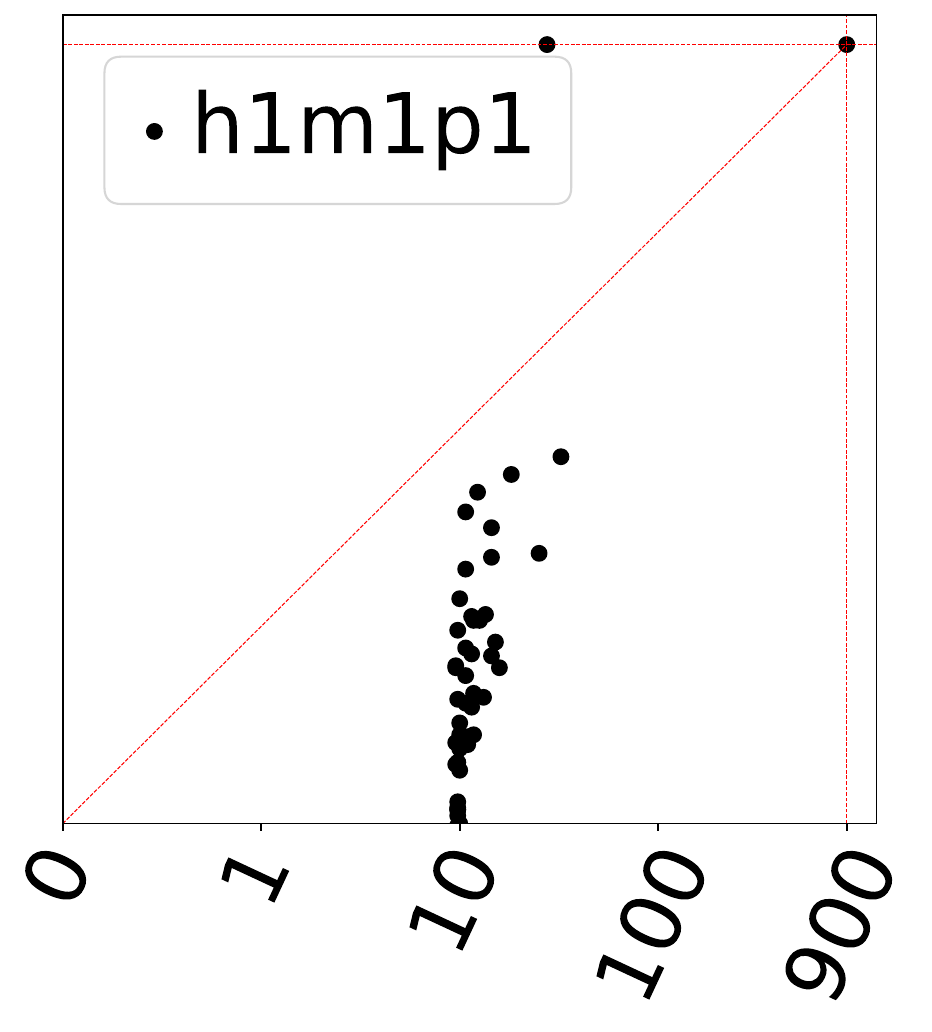}
        \hspace*{0.2cm}
        \includegraphics[scale=0.125]{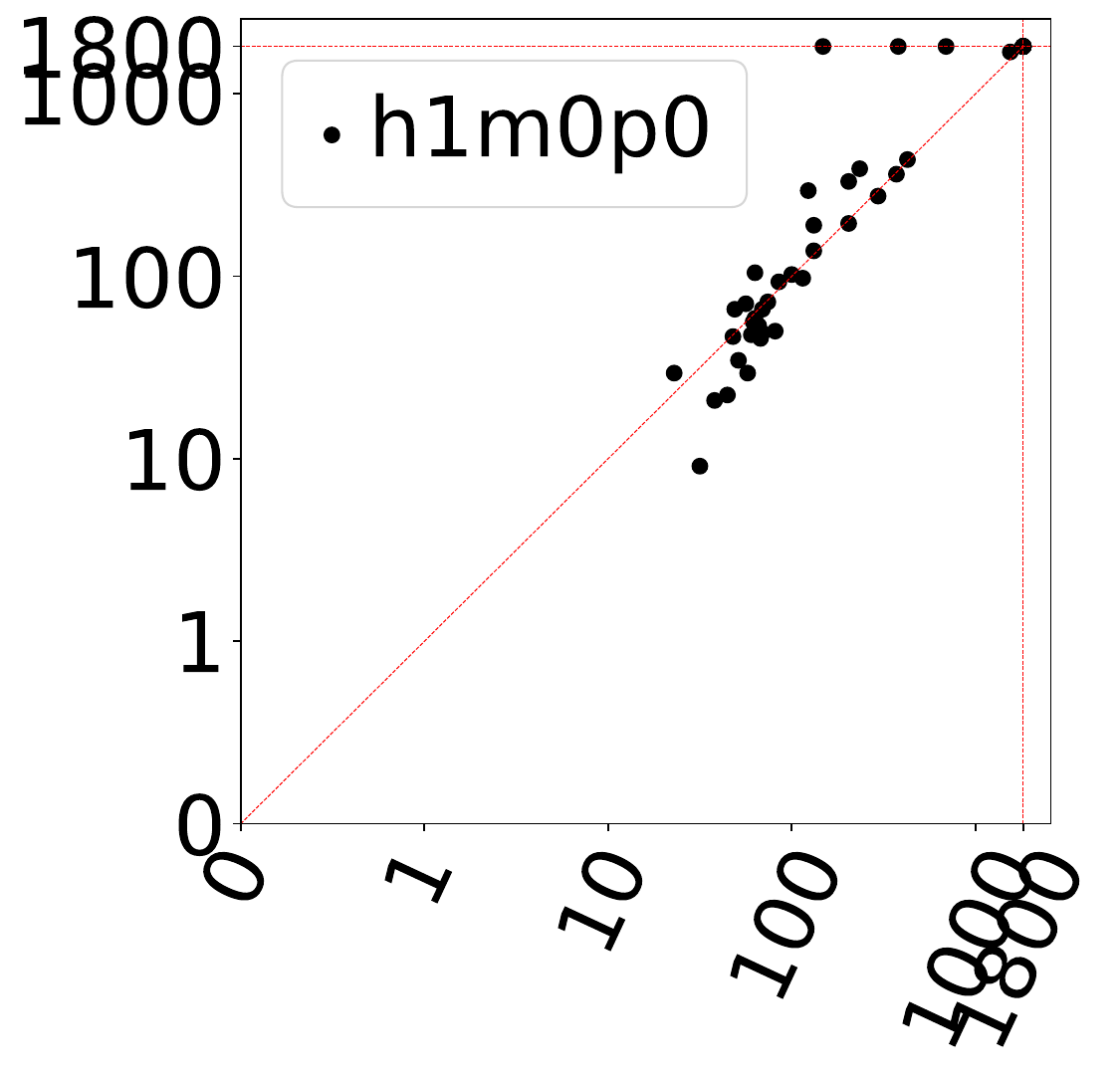}
        \includegraphics[scale=0.125]{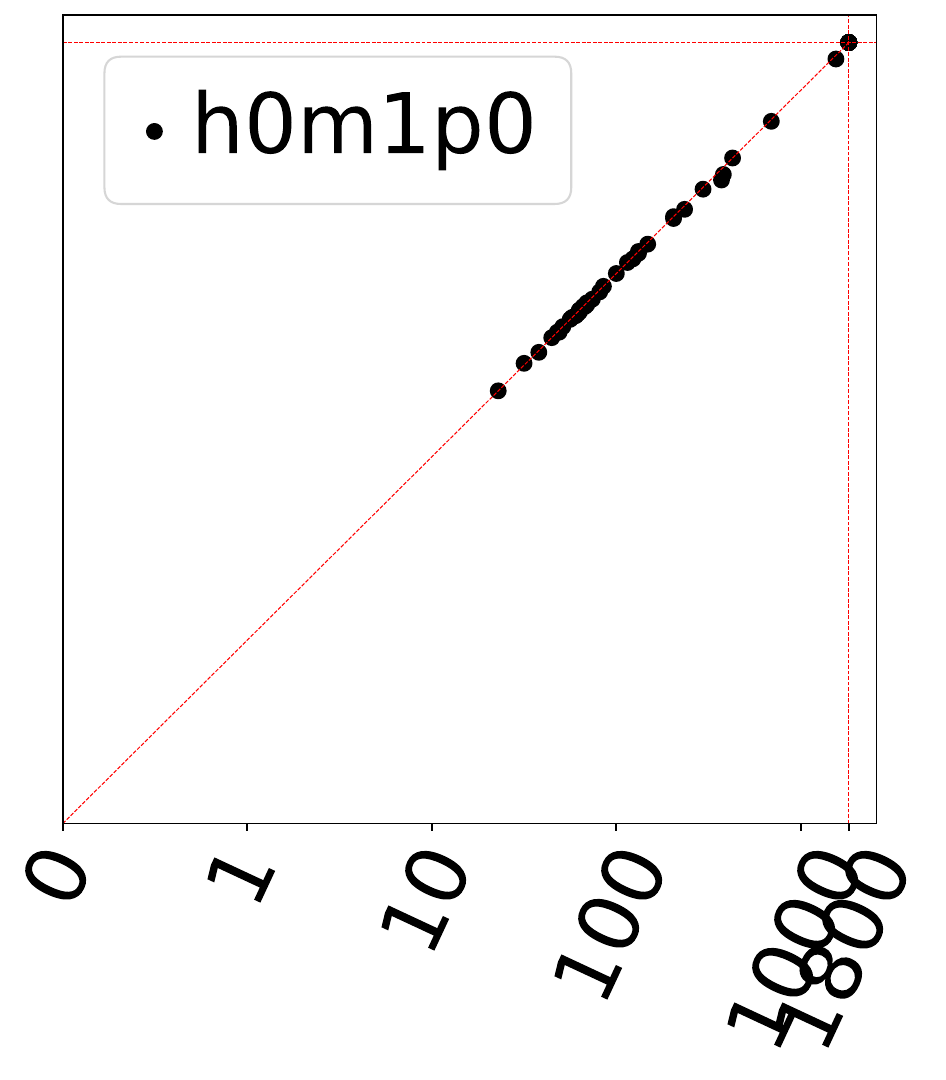}
        \includegraphics[scale=0.125]{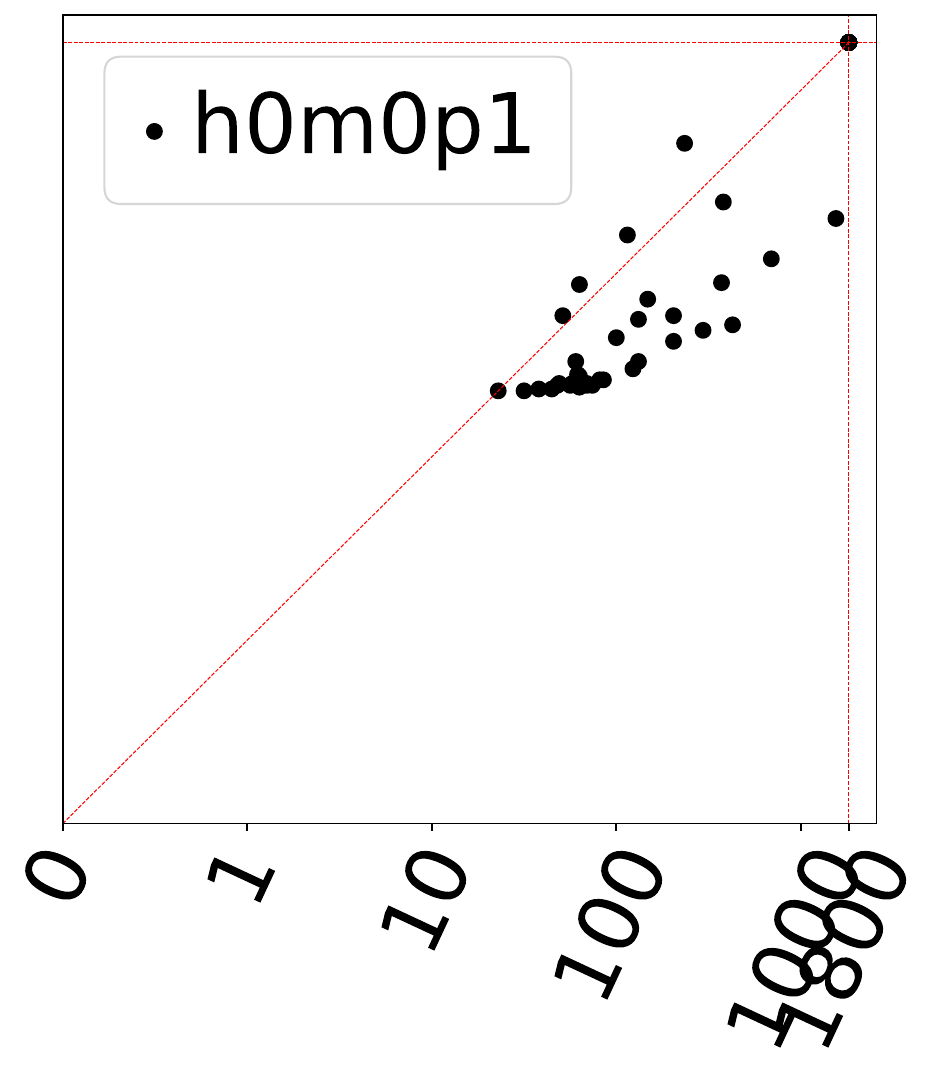}
        \includegraphics[scale=0.125]{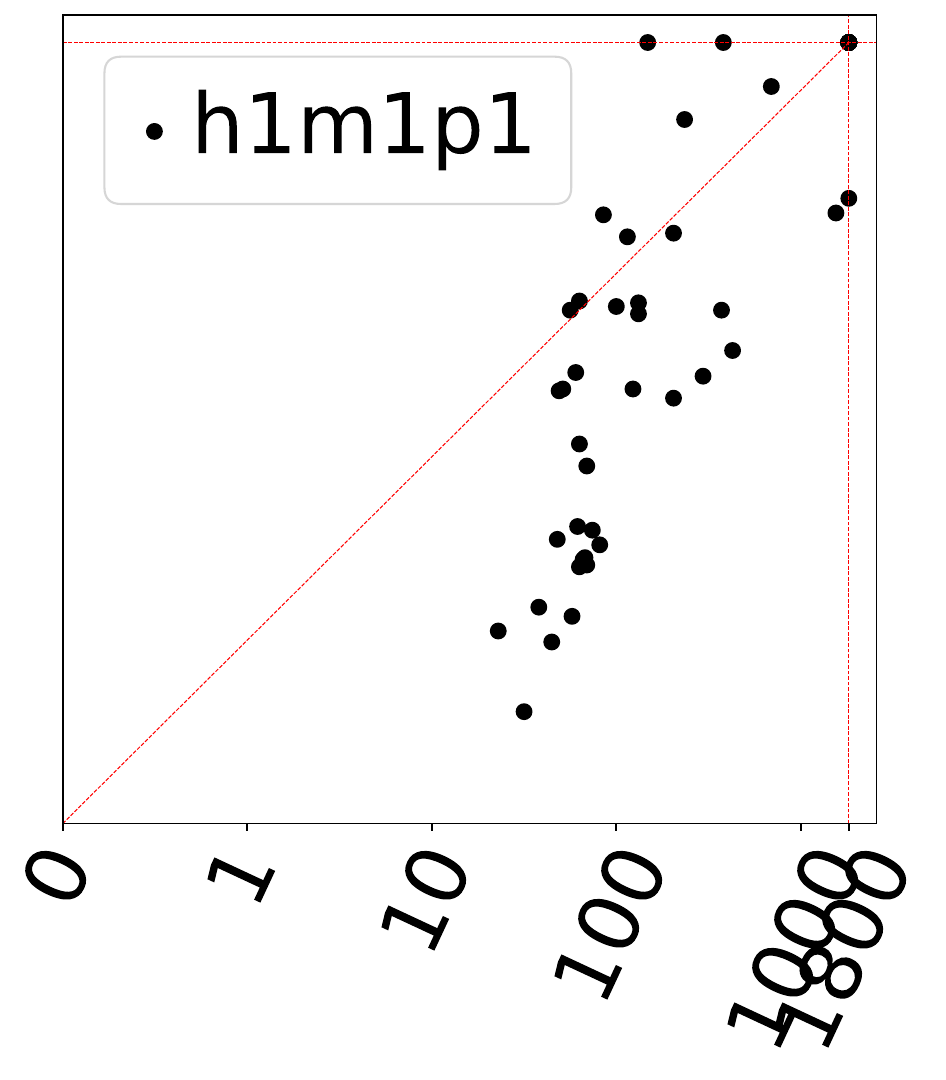}
        
        \caption{Warehouse ($123 \times 321$)}
        \label{fig:scatter-warehouse}
    \end{subfigure}

    % PARIS
    \begin{subfigure}{\linewidth} % Use the subfigure environment
        \centering
        \includegraphics[scale=0.125]{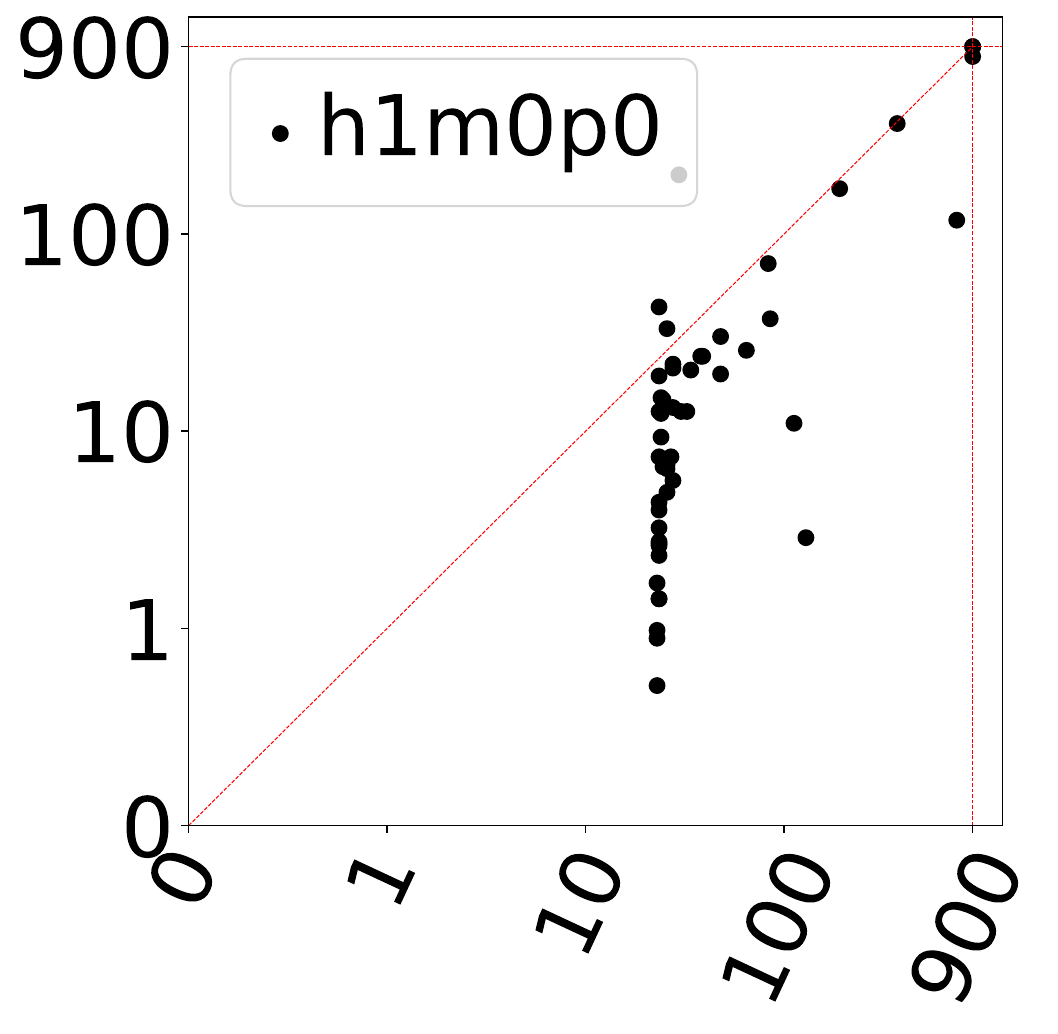}
        \includegraphics[scale=0.125]{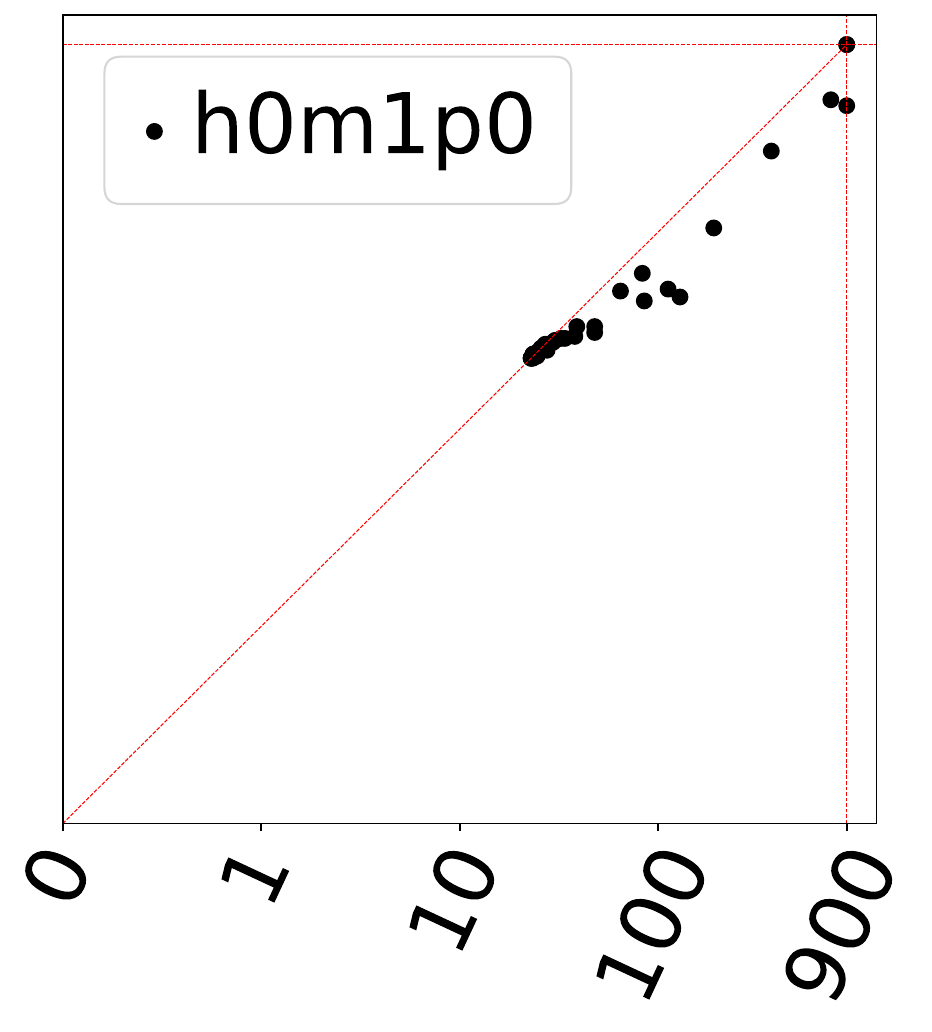}
        \includegraphics[scale=0.125]{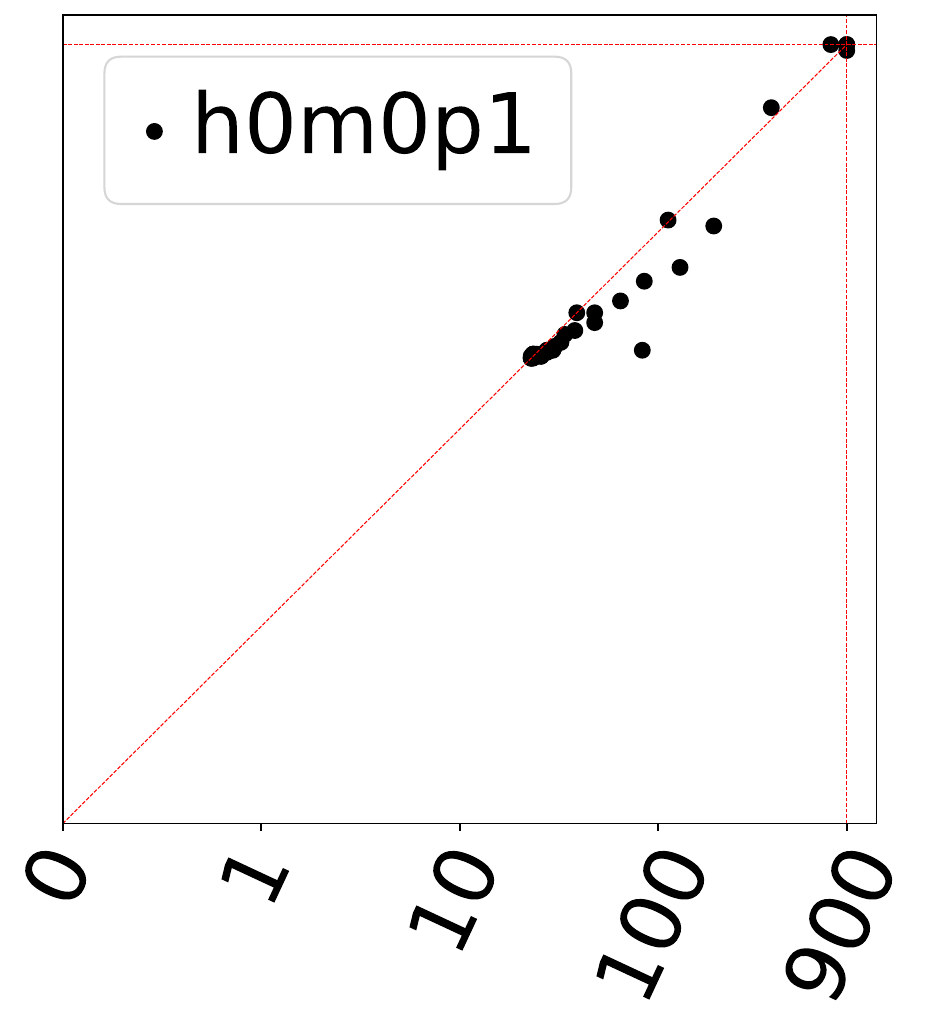}
        \includegraphics[scale=0.125]{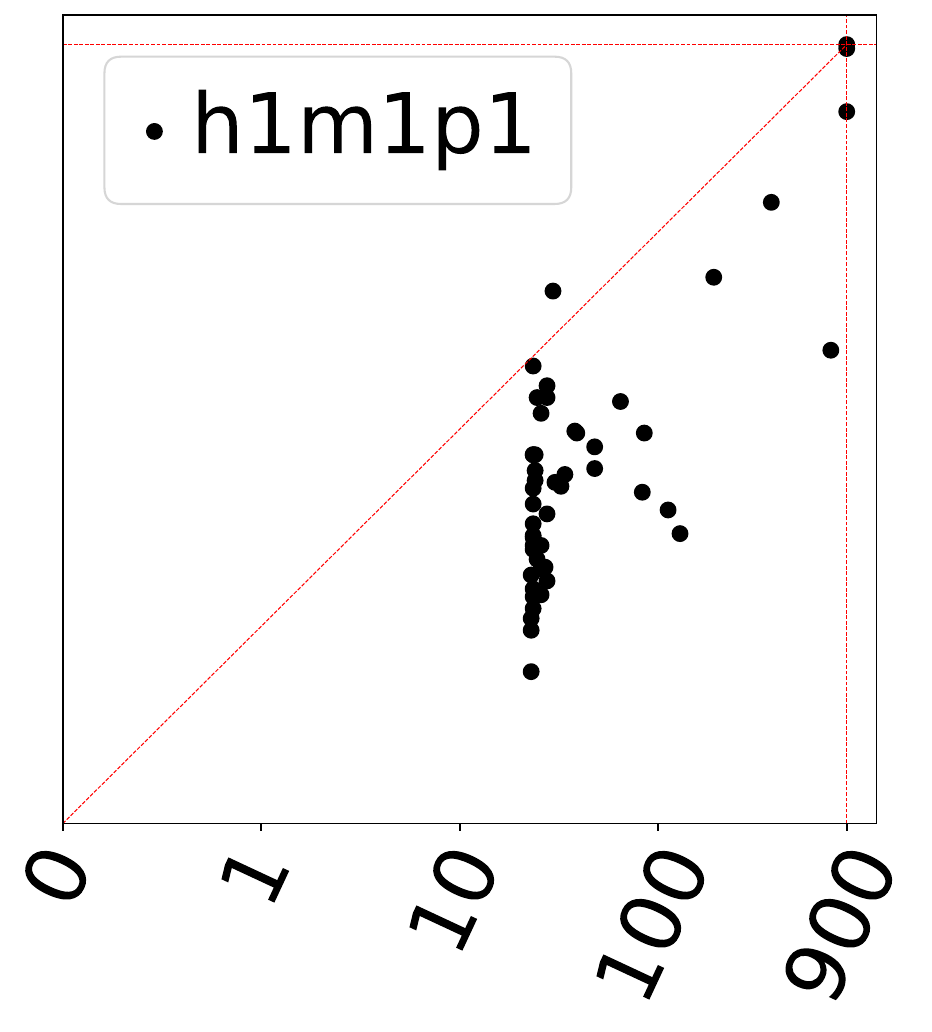}
        \hspace*{0.2cm}
        \includegraphics[scale=0.125]{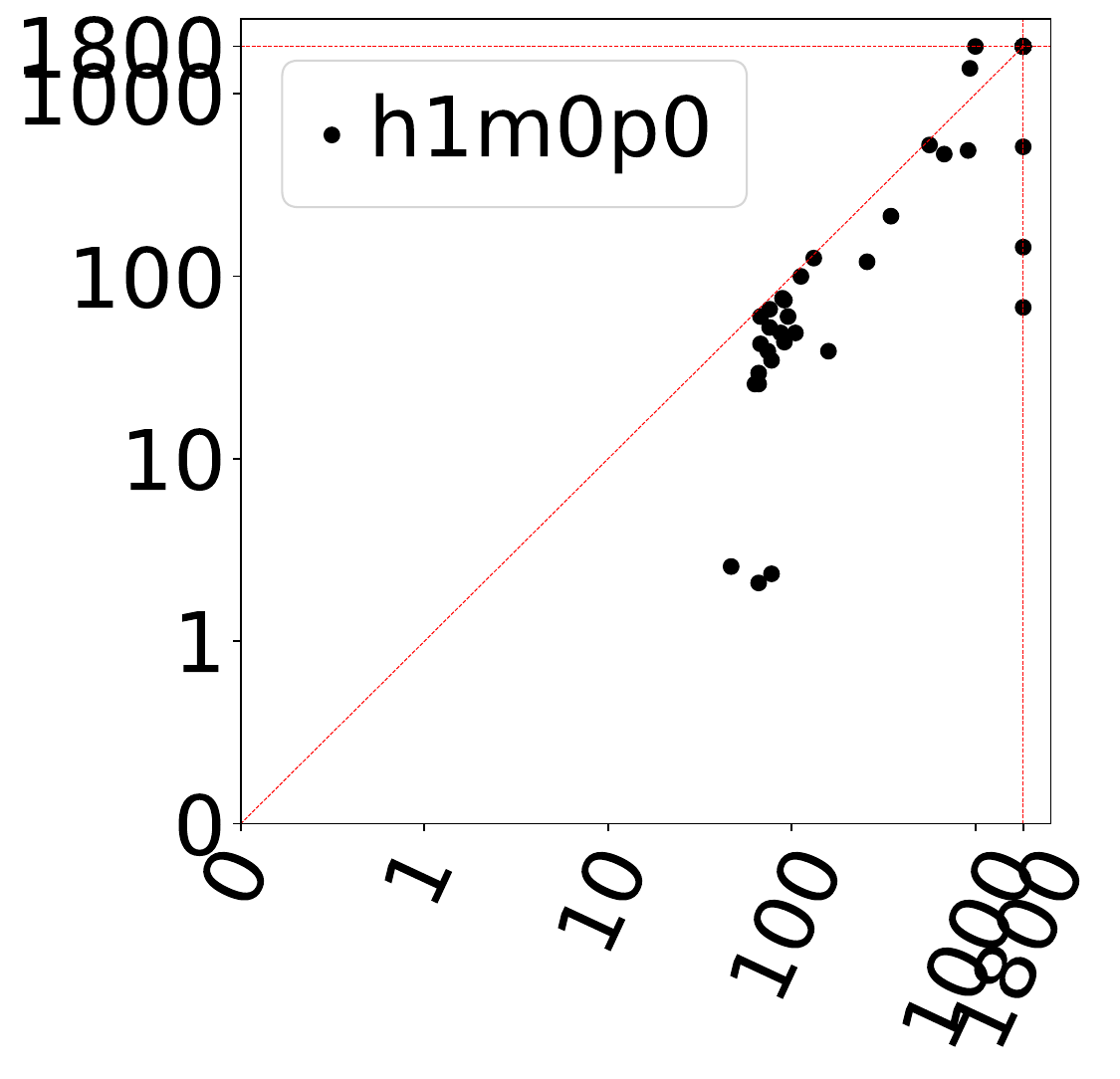}
        \includegraphics[scale=0.125]{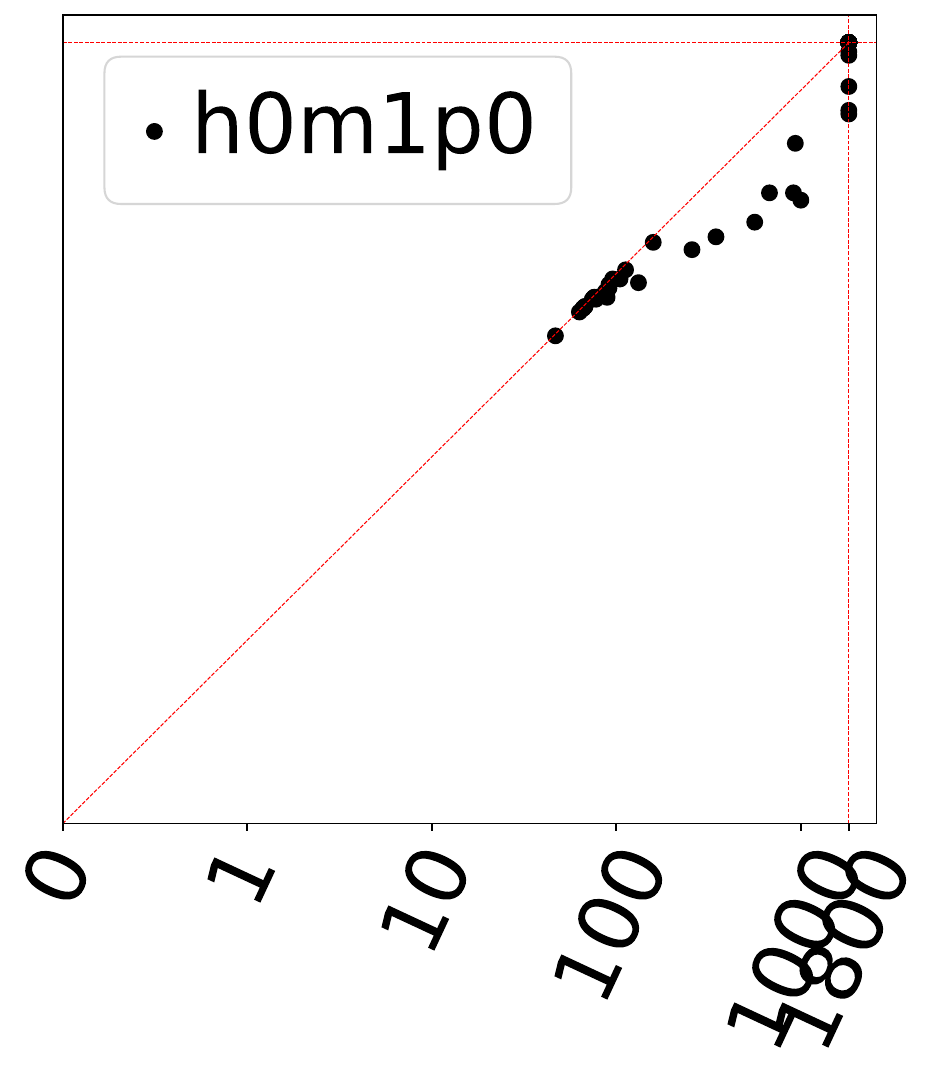}
        \includegraphics[scale=0.125]{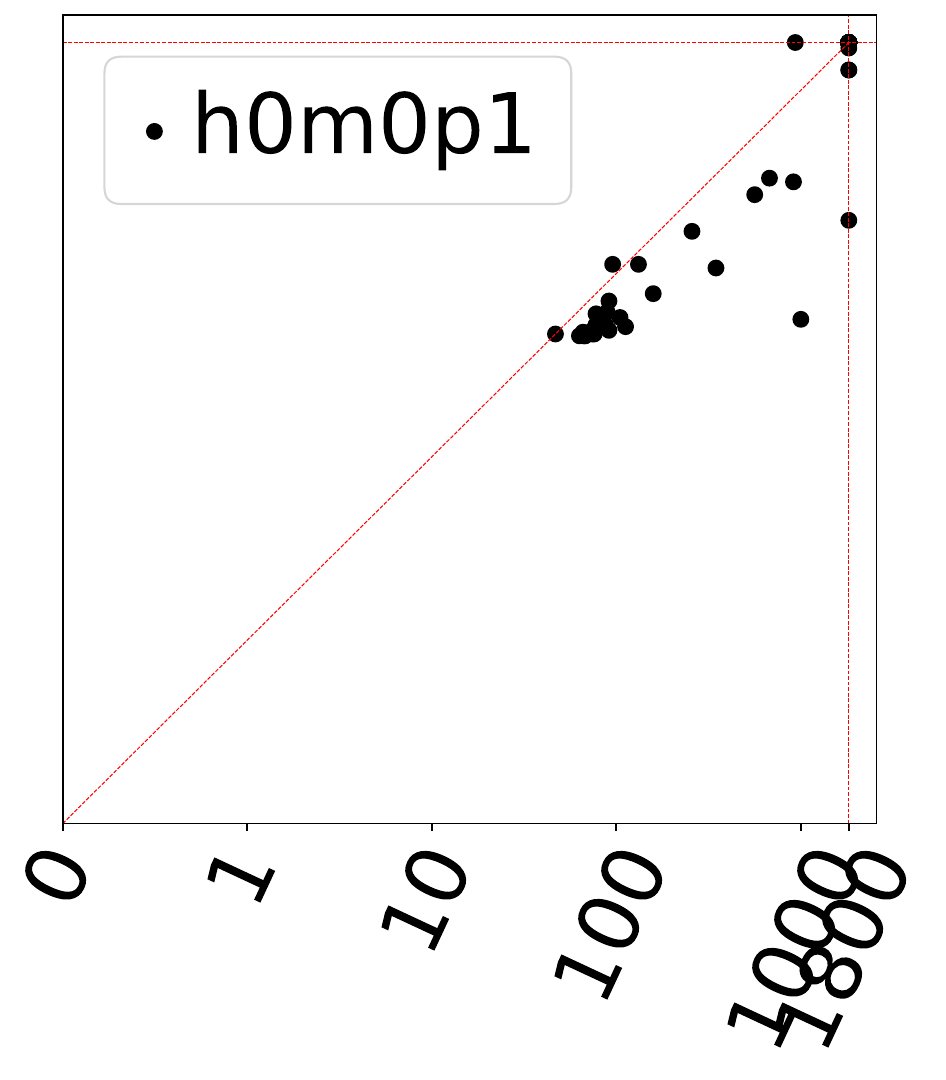}
        \includegraphics[scale=0.125]{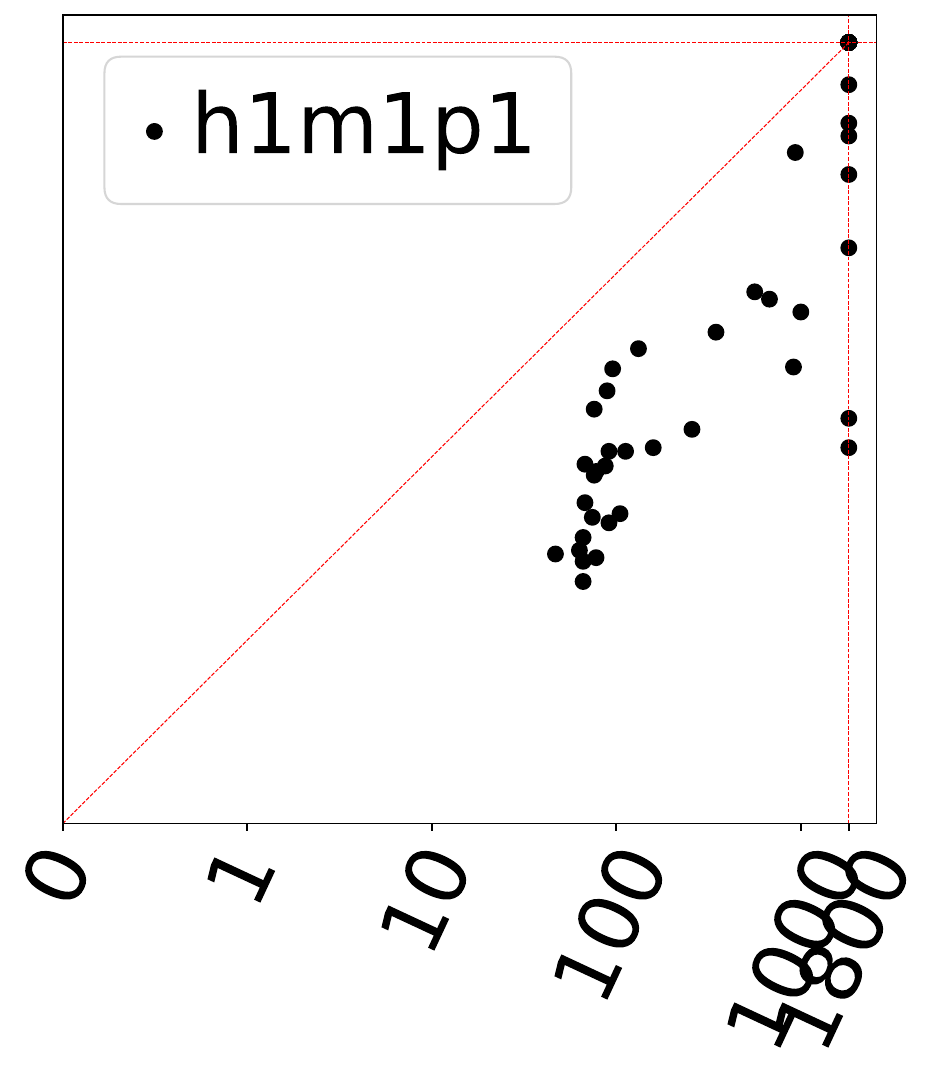}
        
        \caption{Paris ($256 \times 256$)}
        \label{fig:scatter-paris}
    \end{subfigure}

    % % % SYDNEY
    % \begin{subfigure}{\linewidth} % Use the subfigure environment
    %     \centering
    %     \includegraphics[scale=0.125]{plots_b3/scatter_syd_50R_b3_h1m0p0.pdf}
    %     \includegraphics[scale=0.125]{plots_b3/scatter_syd_50R_b3_h0m1p0.pdf}
    %     \includegraphics[scale=0.125]{plots_b3/scatter_syd_50R_b3_h0m0p1.pdf}
    %     \includegraphics[scale=0.125]{plots_b3/scatter_syd_50R_b3_h1m1p1.pdf}
    %     \hspace*{0.2cm}
    %     \includegraphics[scale=0.125]{plots_b3/scatter_syd_100R_b3_h1m0p0.pdf}
    %     \includegraphics[scale=0.125]{plots_b3/scatter_syd_100R_b3_h0m1p0.pdf}
    %     \includegraphics[scale=0.125]{plots_b3/scatter_syd_100R_b3_h0m0p1.pdf}
    %     \includegraphics[scale=0.125]{plots_b3/scatter_syd_100R_b3_h1m1p1.pdf}
        
    %     \caption{Sydney ($256 \times 256$)}
    %     \label{fig:scatter-sydney}
    % \end{subfigure}

    % DEN
    \begin{subfigure}{\linewidth} % Use the subfigure environment
        \centering
        \includegraphics[scale=0.125]{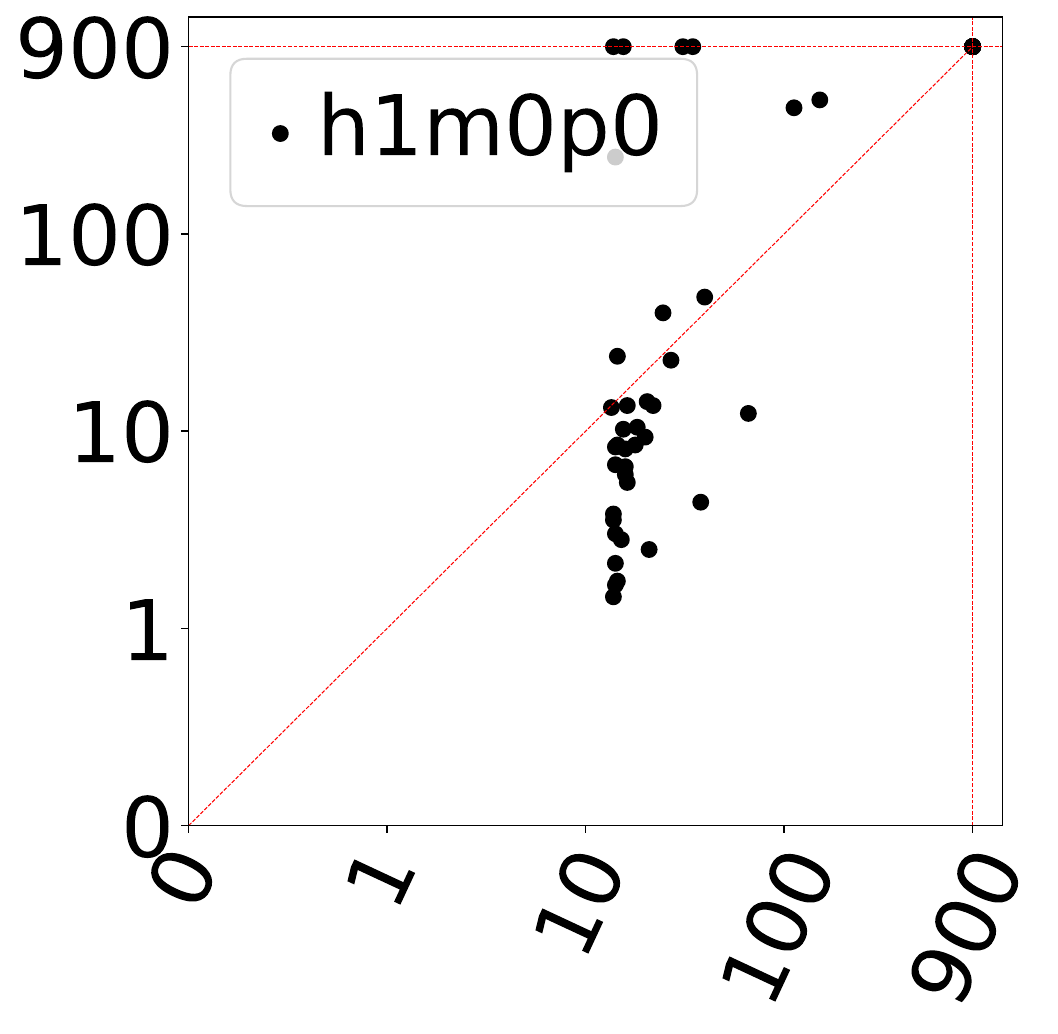}
        \includegraphics[scale=0.125]{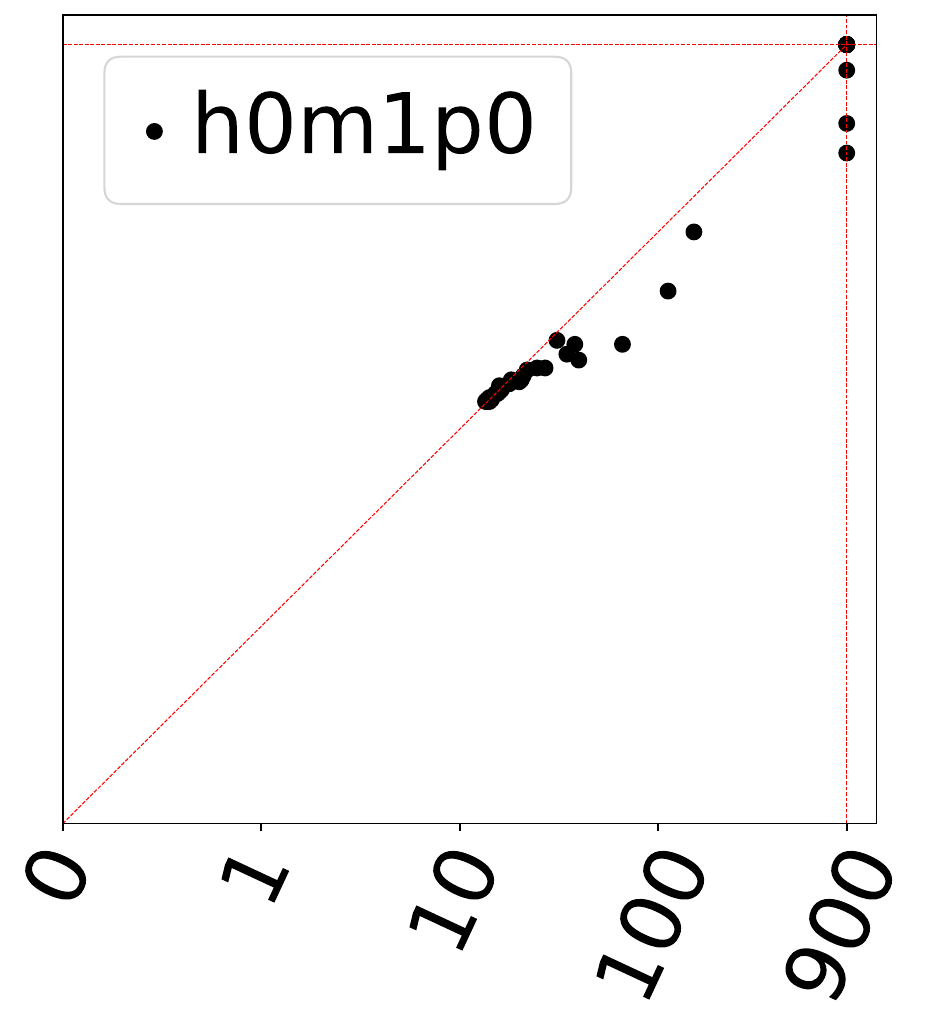}
        \includegraphics[scale=0.125]{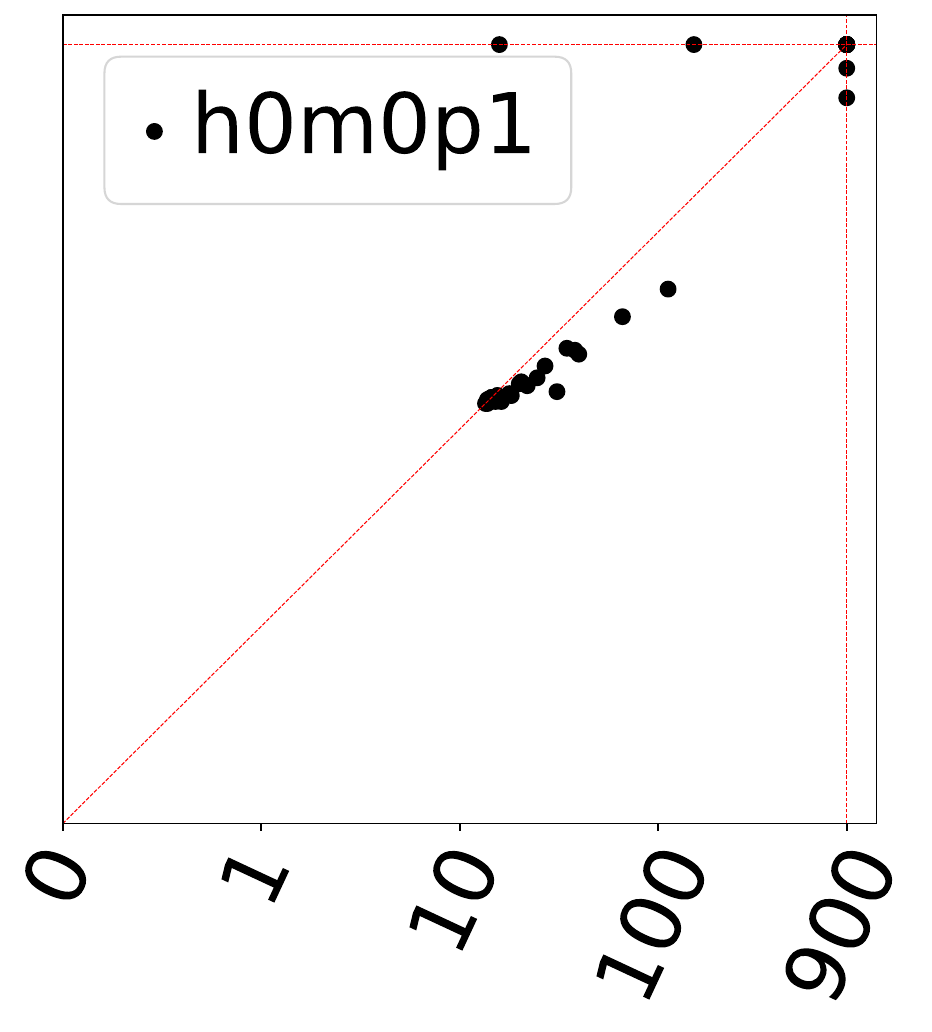}
        \includegraphics[scale=0.125]{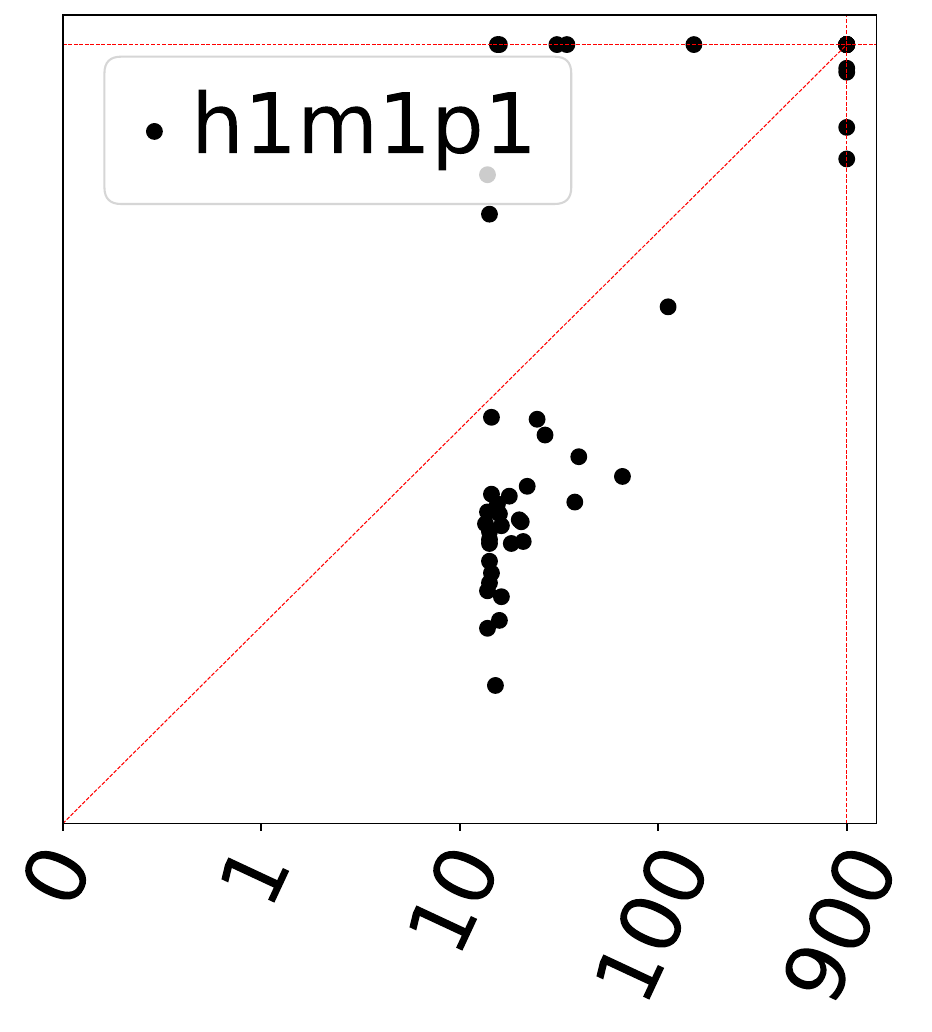}
        \hspace*{0.2cm}
        \includegraphics[scale=0.125]{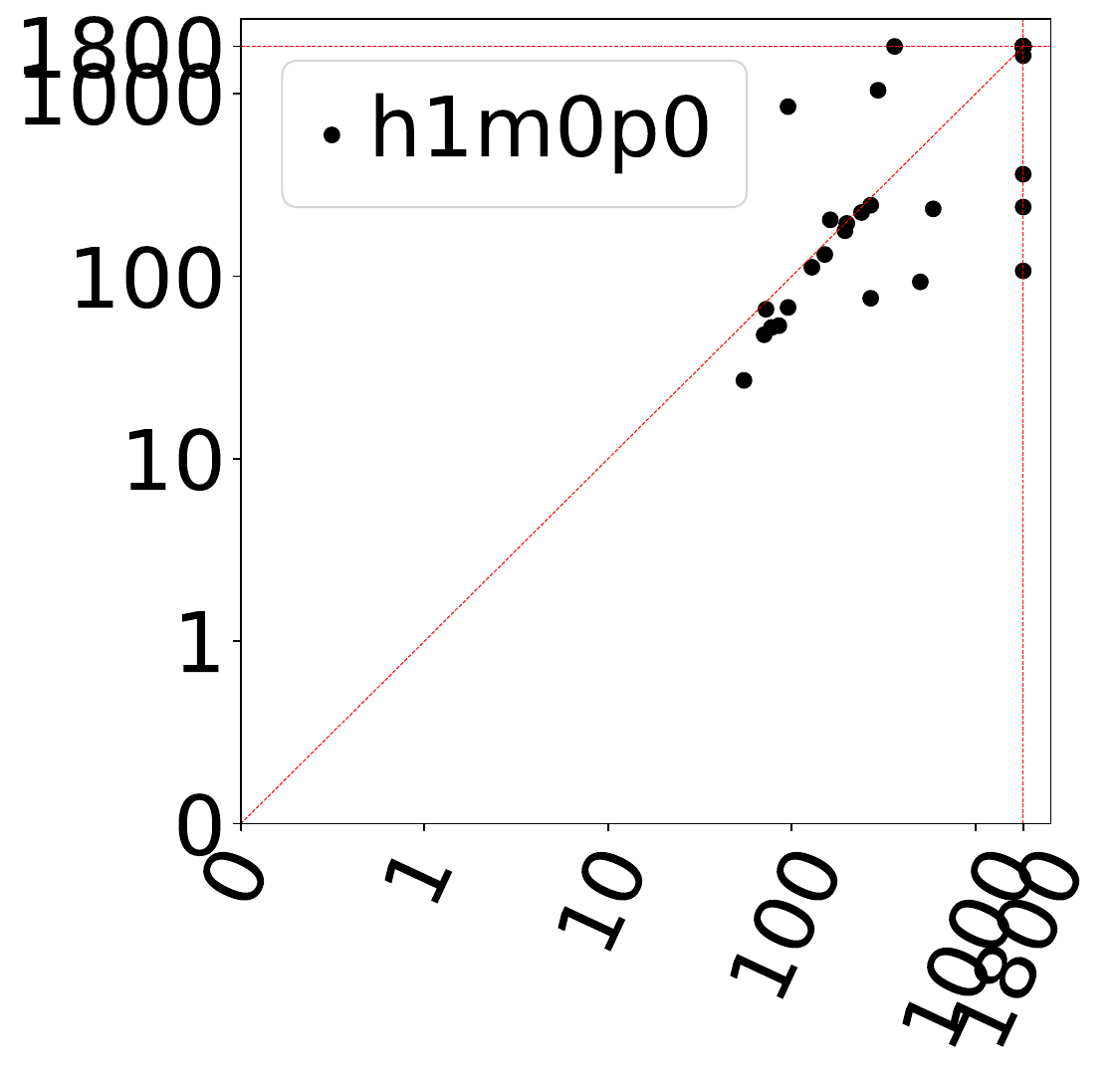}
        \includegraphics[scale=0.125]{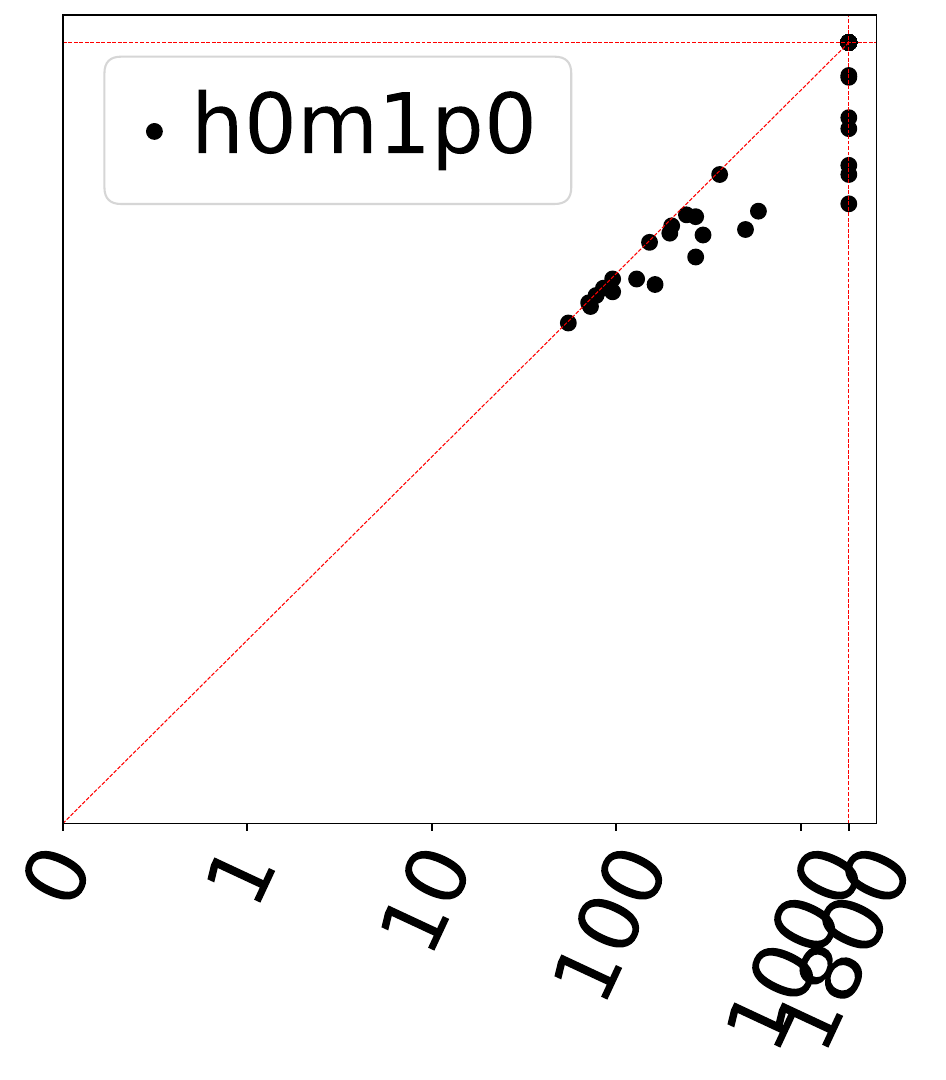}
        \includegraphics[scale=0.125]{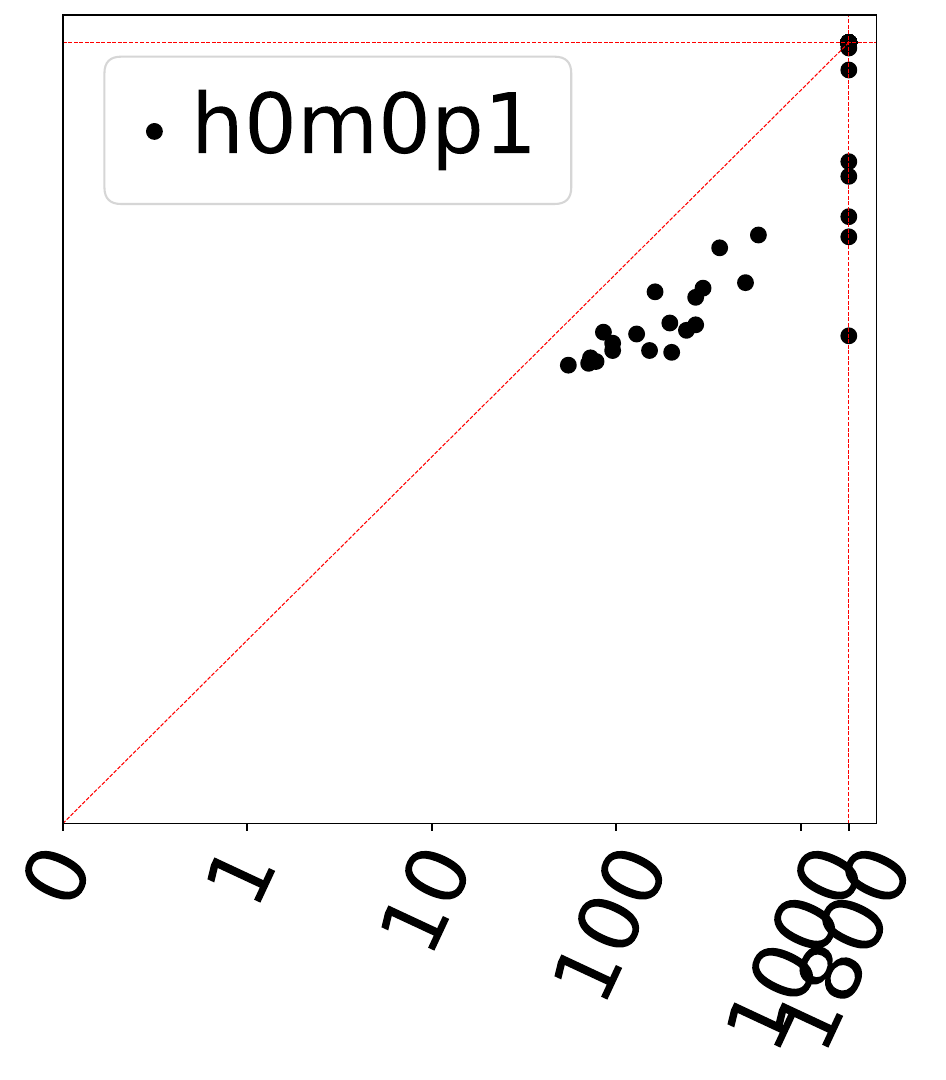}
        \includegraphics[scale=0.125]{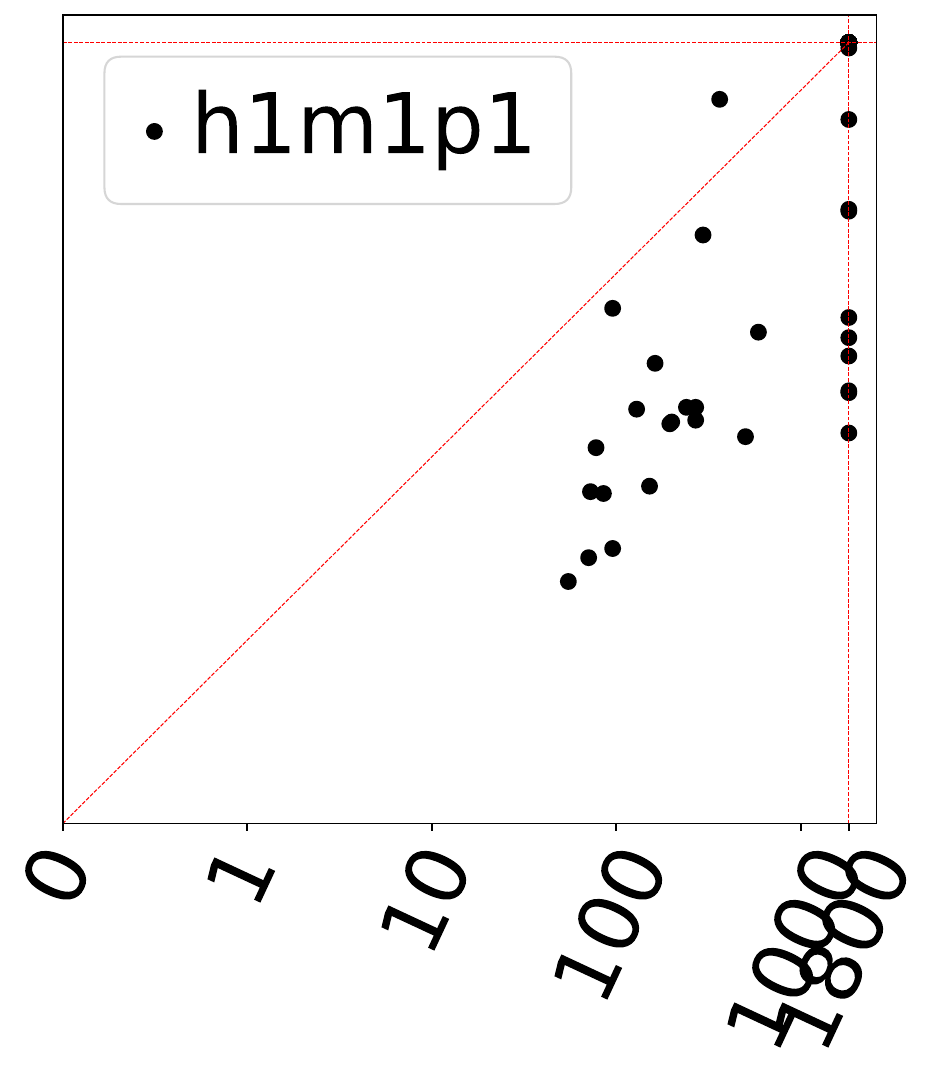}
        
        \caption{Den ($257 \times 256$)}
        \label{fig:scatter-den}
    \end{subfigure}

    \caption{Computation Time Comparison with Baseline for Various Workspaces 
    % \\ (\textit{X-axis:} $\log{(\text{Runtime(\si{\second})})}$ of (p0h0c0), \textit{Y-axis:} $\log{(\text{Runtime(\si{\second})})}$ 
    \\ \hspace*{1.3cm} (\textit{X-axis:} Runtime(\si{\second}), \textit{Y-axis:} Runtime(\si{\second}))
     \hspace*{1.3cm} (Leftmost four plots: $R = 50$, Rightmost four plots: $R = 100$)  
    }
    \label{fig:scatter_plots}
\end{figure*}

\section{Evaluation}
\label{sec:eval}

% \input{table_1}

% VIOLIN PLOTS

\begin{figure}[!ht]
    
    % RANDOM_200
    \begin{subfigure}{\columnwidth} % Use the subfigure environment
        \includegraphics[scale=0.27]{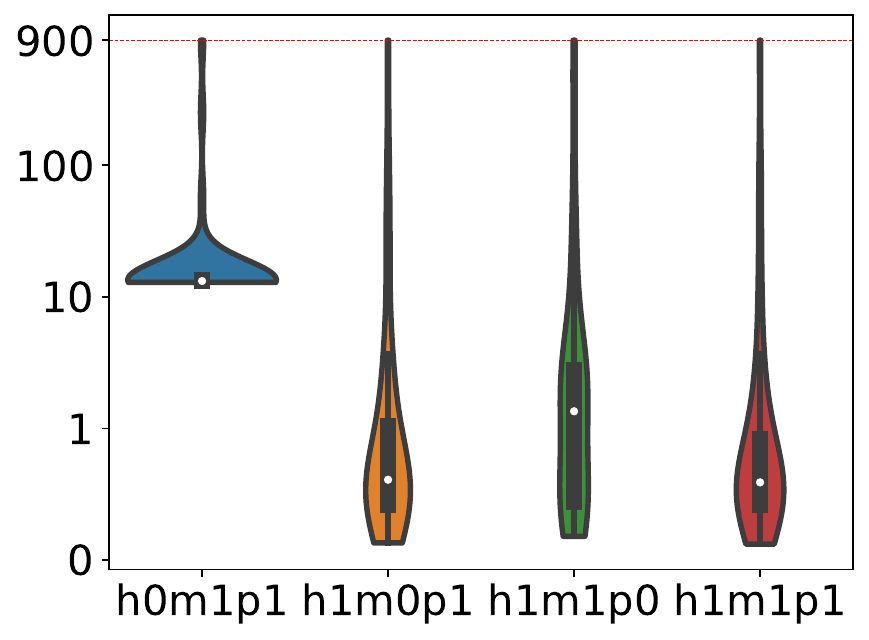}
        \includegraphics[scale=0.27]{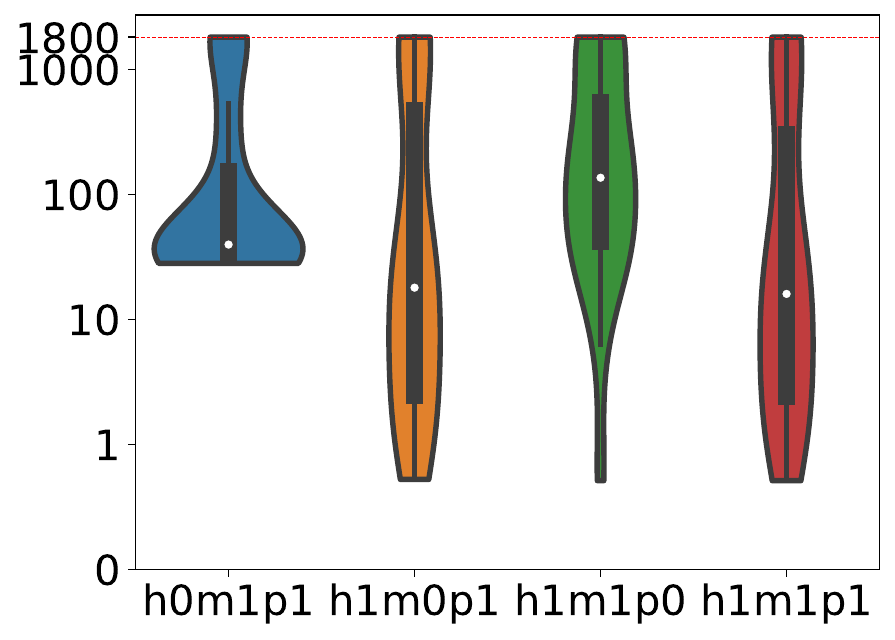}
        \caption{Random ($200 \times 200$)}
        \label{fig:violin-random}
    \end{subfigure}

    % WAREHOUSE
    \begin{subfigure}{\columnwidth} % Use the subfigure environment
        \includegraphics[scale=0.27]{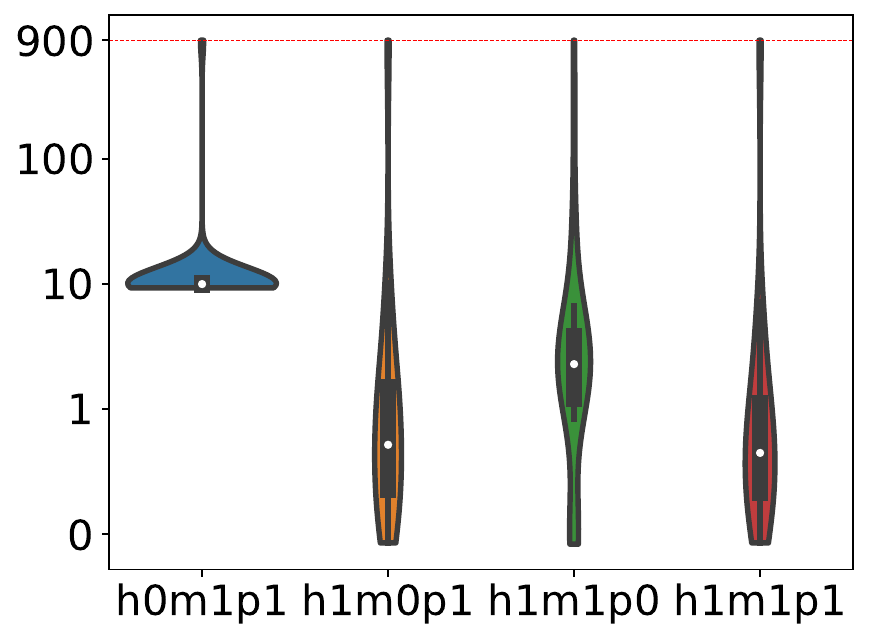}
        \includegraphics[scale=0.27]{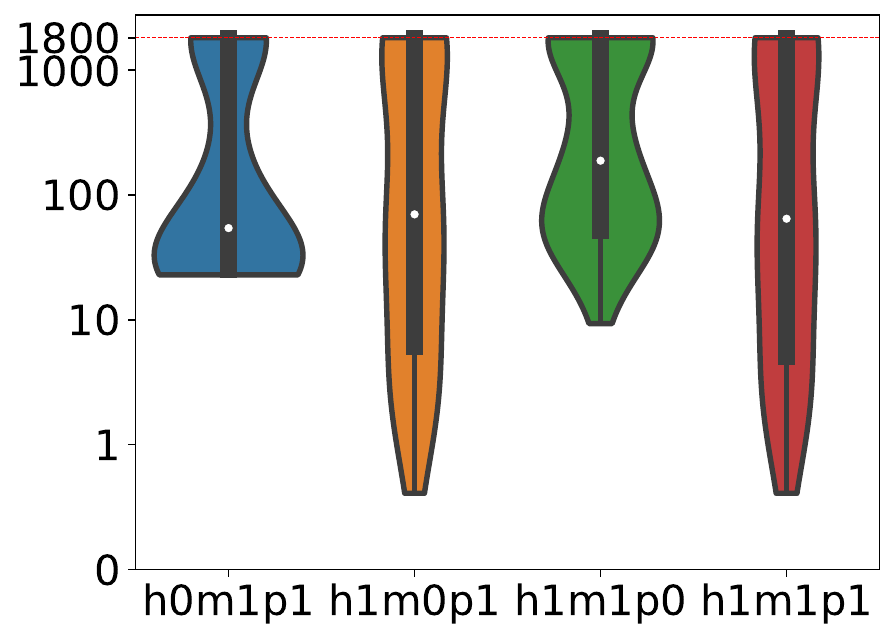}
        \caption{Warehouse ($123 \times 321$)}
        \label{fig:violin-warehouse}
    \end{subfigure}

    % PARIS
    \begin{subfigure}{\columnwidth} % Use the subfigure environment
        \includegraphics[scale=0.27]{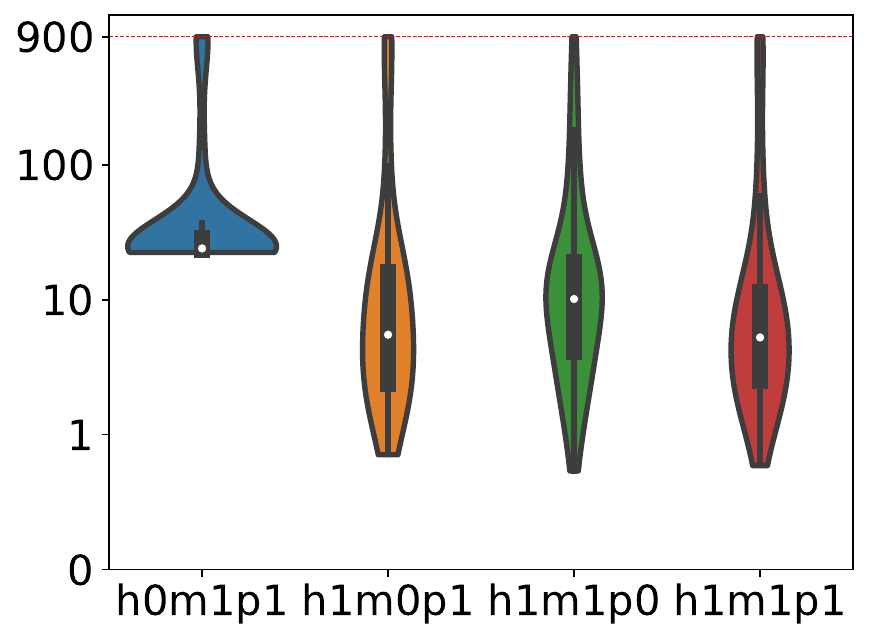}
        \includegraphics[scale=0.27]{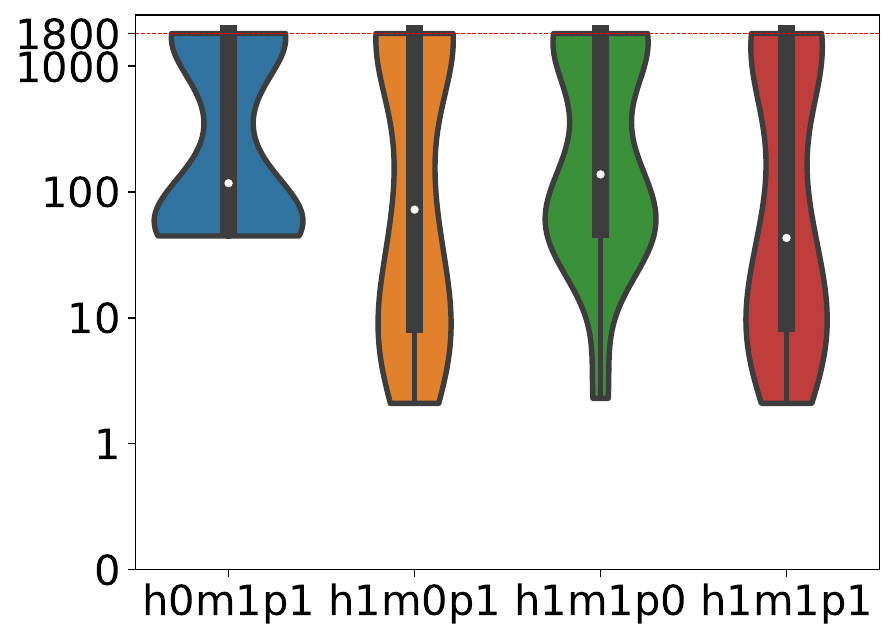}
        \caption{Paris ($256 \times 256$)}
        \label{fig:violin-paris}
    \end{subfigure}

    % % SYDNEY
    % \begin{subfigure}{\columnwidth} % Use the subfigure environment
    %     \includegraphics[scale=0.27]{plots_a2/violin_syd_50R_a2.pdf}
    %     \includegraphics[scale=0.27]{plots_a2/violin_syd_100R_a2.pdf}
    %     % \includegraphics[scale=0.18]{plots/v1.pdf}
    %     \caption{Sydney ($256 \times 256$)}
    %     \label{fig:violin-sydney}
    % \end{subfigure}

    % DEN
    \begin{subfigure}{\columnwidth} % Use the subfigure environment
        \includegraphics[scale=0.27]{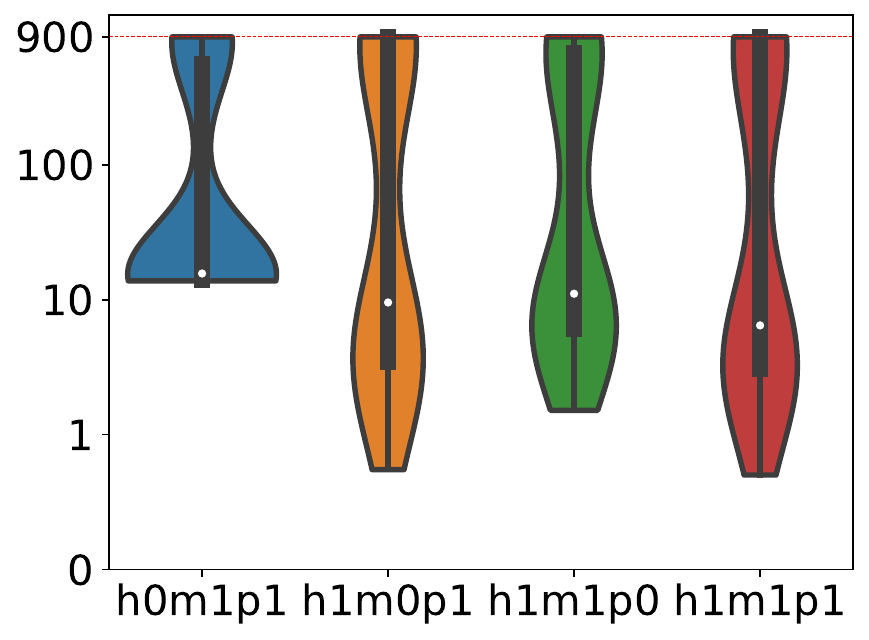}
        \includegraphics[scale=0.27]{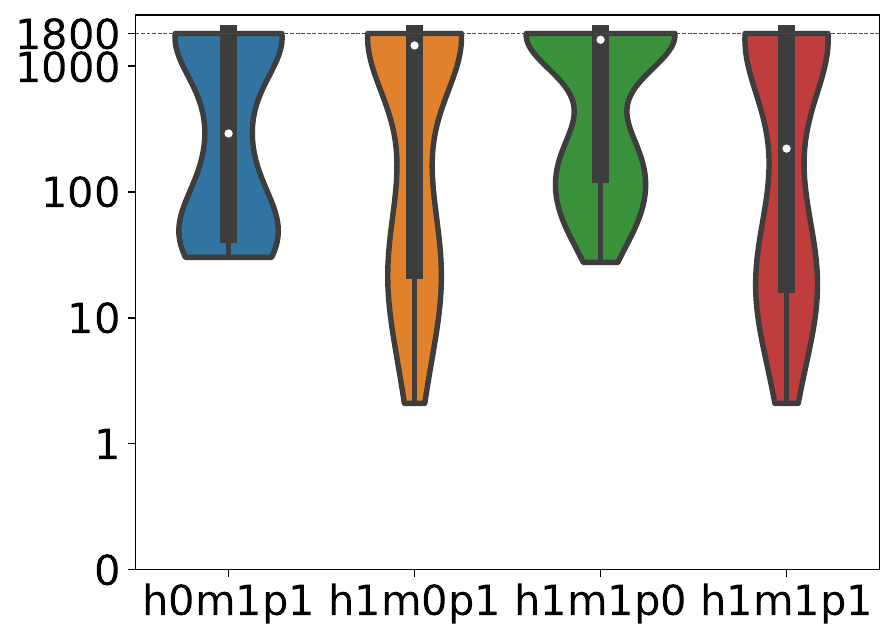}
        \caption{Den ($257 \times 256$)}
        \label{fig:violin-den}
    \end{subfigure}

    \caption{Ablation Study for Various Workspaces 
    \\ \hspace*{1.3cm} (\textit{X-axis:} Approaches, \textit{Y-axis:} Runtime(\si{\second})) 
    \\ \hspace*{1.3cm} (Left: $R = 50$, Right: $R = 100$)}
    \label{fig:violin_plots}
\end{figure}

% SCATTER PLOTS (in split mode) 

\begin{figure}[!ht]
        % robots
    \begin{subfigure}{0.32\columnwidth} % Use the subfigure environment
        \centering
        \includegraphics[scale=0.05]{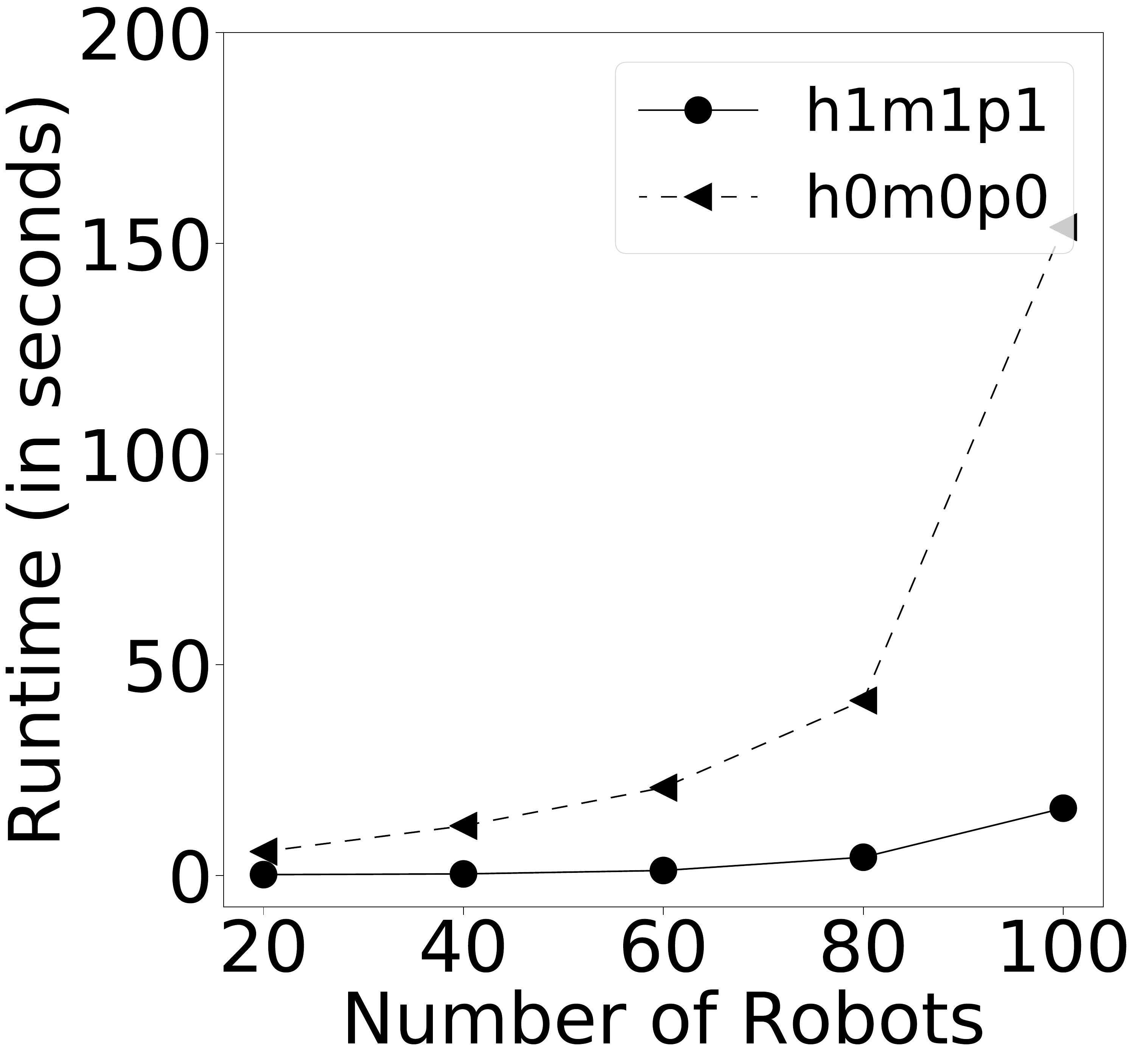}
        
        \caption{Varying $R$}
        \label{fig:vary-r}
    \end{subfigure}
    \begin{subfigure}{0.32\columnwidth} % Use the subfigure environment
        \centering
        \includegraphics[scale=0.05]{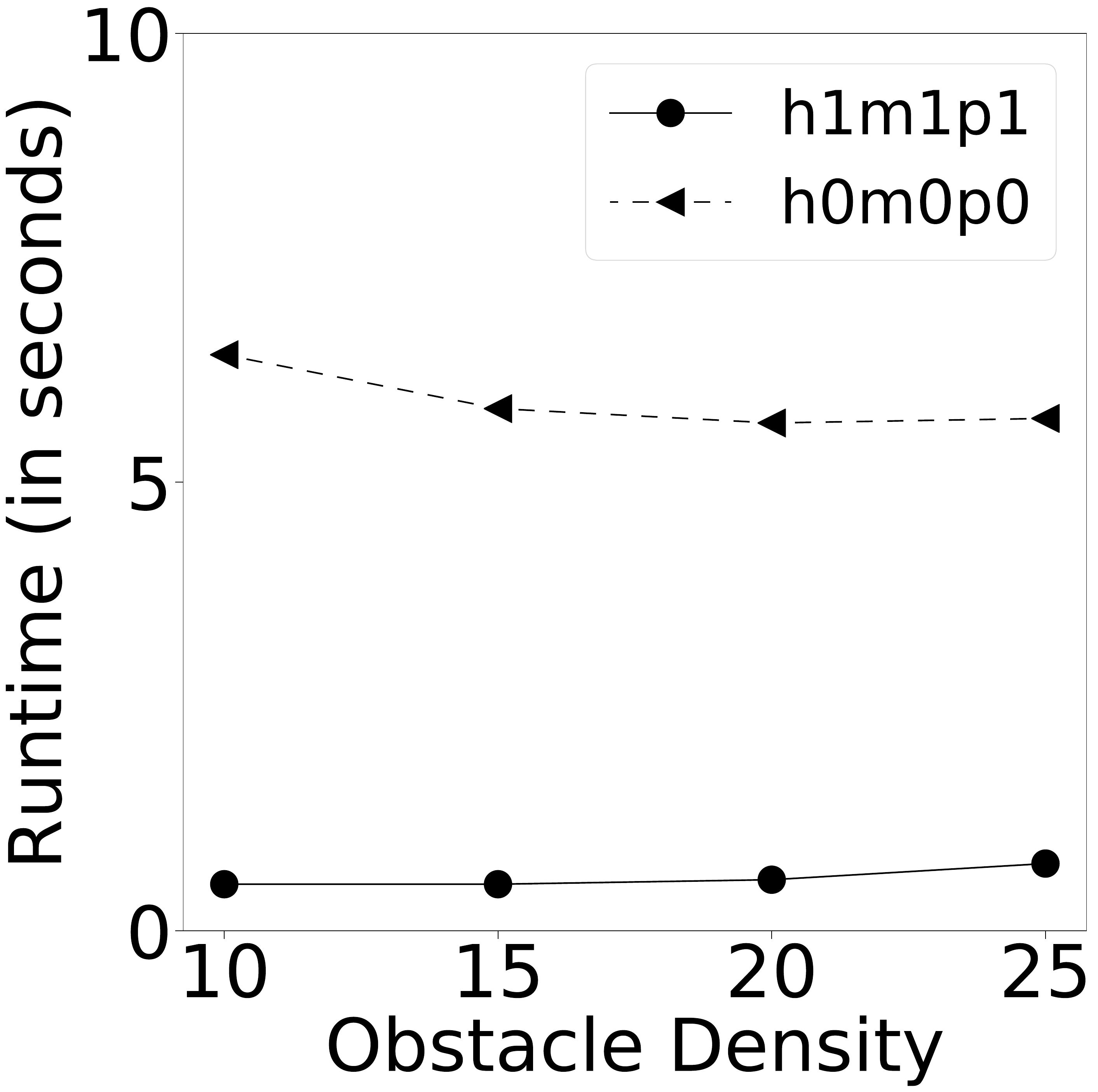}
        
        \caption{Varying $OD$}
        \label{fig:vary-od}
    \end{subfigure}
    \begin{subfigure}{0.32\columnwidth} % Use the subfigure environment
        \centering
        \includegraphics[scale=0.05]{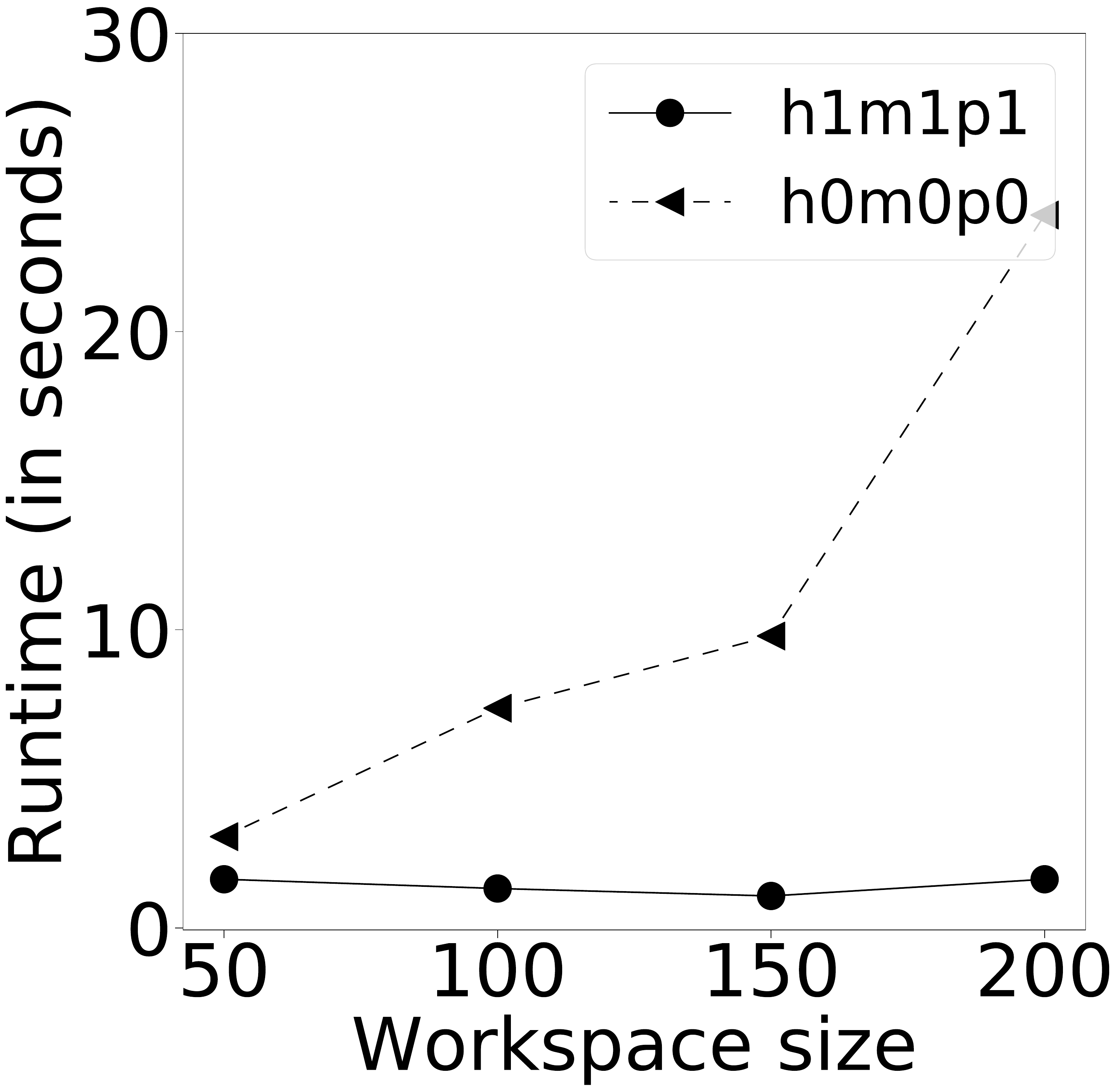}
        
        \caption{Varying $WS$ size}
        \label{fig:vary-ws}
    \end{subfigure}
% \caption{Scalability Comparison between Baseline ($\mathtt{h0m0p0}$) and Our Approach ($\mathtt{h1m1p1}$)  
%     \\ (a) $OD = 20\%$, $WS: 200\times200$
%      (b) $R = 50$, $WS: 100\times100$
%      (c) $R = 50$, $OD: 20\%$
%     }

\caption{Scalability Comparison between Baseline \\ \hspace*{1.2cm} ($\mathtt{h0m0p0}$) and Our Approach ($\mathtt{h1m1p1}$)  
    \\ \hspace*{1.2cm} (a) $OD = 20\%$, $WS: 200\times200$
   \\ \hspace*{1.2cm} (b) $R = 50$, $WS: 100\times100$
    \\ \hspace*{1.2cm} (c) $R = 50$, $OD: 20\%$
    }
    
    \label{fig:vary-param-plots}
\end{figure}

In this section, we present the results obtained from our experimental evaluation of Algorithm~\ref{algo1}.

\subsection{Experimental Setup}
\label{exp:setup}

% \subsubsection{Baseline.}

To assess Algorithm~\ref{algo1}, we consider the state-of-the-art CBS-TA~\cite{cbsta} as the baseline.
Our approach consists of three pivotal components: 
(a) \textbf{H}euristic distance-based assignment computation, denoted by `$\mathtt{h}$', computes the assignments efficiently by trying to avoid the exhaustive computation of all robot-goal paths (unconstrained). 
(b) Path \textbf{m}emoization, denoted by `$\mathtt{m}$', caches the paths computed under certain constraints for future reuse, and
(c) Assignment computation \textbf{p}ostponement, denoted by `$\mathtt{p}$', uses the encountered conflicts to postpone assignment computations,
%
%
% To understand the forthcoming plots, we use the notations $h$, $m$, and $p$. 
We use the notations $\mathtt{h}$, $\mathtt{m}$, and $\mathtt{p}$ in the forthcoming plots. 
Here are a few examples of what they represent:
$\mathtt{h0m0p0}$: denotes the baseline in which none of the three enhancements are present,
$\mathtt{h0m0p1}$: denotes the partially enhanced version consisting of only the postponement of assignment computation, and
$\mathtt{h1m1p1}$: denotes our approach consisting of all the three enhancements. 
We implement the baseline and our proposed algorithm in Python. 
The source code of the implementation is submitted as supplementary material.

\subsubsection{Benchmarks and Evaluation Metrics.}

We evaluate Algorithm~\ref{algo1} on benchmark workspaces~\cite{sturtevant2012benchmarks, stern2019mapf} as well as a randomly generated workspace.
% We use the following four evaluation metrics: 
% (a)~\textbf{Runtime}: the execution time, 
% (b)~\textbf{Unconst. Path Comp.}: the number of unconstrained robot-goal paths computed,
% (c)~\textbf{Const. Path Comp.}: the number of constrained paths combuted, and 
% (d)~\textbf{Assignment Comp.}: the number of assignments computed.
%
We use runtime to be the evaluation metric.

%
% We also report the \textbf{Speedup} that Algorithm~\ref{algo1} achieves over the baseline algorithm.

We run all the experiments in a desktop machine with Intel\textregistered \ Core\textsuperscript{TM} i$7$-$8700$ CPU\,@\,$3.20$\,GHz processor, $32$\,GB RAM, and Ubuntu $20.04$ OS. 
We run each experiment for $50$ times to report the results.
% and report the mean and standard deviation for the evaluation metrics. 

\subsection{Experimental Results}
\label{exp:results}

Throughout our experiments, we take $900 \si{\second}$ and $1800 \si{\second}$ as the timeouts for the multi-robot systems having $50$ and $100$ robots, respectively. 
% We provide additional results in Appendix~\ref{subsec:expr_results_appendix}.

\subsubsection{Algorithm~\ref{algo1} vs. Baseline.}
In Figure~\ref{fig:scatter_plots}, we present the runtime scatter plots for the comparison between different approaches.
% to compare our approach (p1m1h1) as well as the other variations that implement only a single enhancement with the baseline (p0m0h0).
The X-axis represents the runtime ($\si{\second}$) of the baseline ($\mathtt{h0m0p0}$), while the Y-axis displays the runtime ($\si{\second}$) of the labeled approach.
% our approach ($\mathtt{h1m1p1}$) alongside three other variations that implement only a single enhancement.
The number of robots is denoted by $R$.
The first four scatter plots in each row are for $R = 50$ whereas the last four are for $R = 100$.
The red horizontal and vertical lines indicate the timeout. 
The crossover point of the runtimes is depicted by the red diagonal line.
We observe that across all cases, our approach ($\mathtt{h1m1p1}$) outperforms the baseline ($\mathtt{h0m0p0}$) and the other variations that implement only a single enhancement.

% \zb{Will add text for Table 1 ... }

\subsubsection{Ablation Study.}
In Figure~\ref{fig:violin_plots}, we use comparative violin plots to perform an ablation study of the three key enhancements that we incorporate into our approach.
We present the comparative violin plots for two scenarios within each workspace: one with $50$ robots (left side plot) and another with $100$ robots (right side plot).
The X-axis represents the following approaches: $\mathtt{h0m1p1}$, $\mathtt{h1m0p1}$, $\mathtt{h1m1p0}$, and $\mathtt{h1m1p1}$ (see Section~\ref{exp:setup} for their interpretation).
% \begin{itemize}
%     \item $\mathtt{h0m1p1}$: removal of heuristic-based assignment computation from our approach, 
%     \item $\mathtt{h1m0p1}$: removal of path memoization from our approach,
%     \item $\mathtt{h1m1p0}$: removal of assignment postponement from our approach, and
%     \item $\mathtt{h1m1p1}$: our complete approach with no removal.
% \end{itemize}
%
%
The Y-axis represents the runtime ($\si{\second}$) of the four approaches.
The red horizontal line marks the timeout. 
The white dot within the embedded box plot in a violin indicates the median. 
The density of the data points in reflected by the width of the violin. 
Across all cases, the median and the width of the violins collectively show that our approach ($\mathtt{h1m1p1}$), without removal of any enhancement, has the most optimized performance.

\subsubsection{Scalability Analysis.}
In Figure~\ref{fig:vary-param-plots}, we vary three parameters, namely, the number of robots $R$, the obstacle density $OD$, which represents the percentage of workspace cells that are occupied by obstacles, and the workspace size to compare the runtime of our approach ($\mathtt{h1m1p1}$) with that of baseline ($\mathtt{h0m0p0}$).

\section{Conclusion}
\label{sec:conclusion}

We have presented a centralized algorithm to solve the multi-robot goal assignment problem while optimizing the total cost of movement of all the robots. 
The paths assigned to the robots by our algorithm are free from collisions. 
We have considered the established CBS-TA as the baseline for the evaluation of our algorithm.
% Our approach scales well up to a hundred robots in large workspaces. 
Our experimental results, in particular, the scatter plots, reflect that our approach outperforms the baseline by an order of magnitude for almost all the cases.

Our algorithm can also be applied to optimize makespan while solving the collision-free multi-robot goal assignment problem. 
To achieve this, the assignment computation module, currently focusing on optimizing total cost, needs to be replaced by the assignment computation module that optimizes makespan~\cite{lbap, gross1959bottleneck}. 
% Moreover, the CBS algorithm needs to undergo corresponding adjustments to prioritize the selection of the best node based on makespan rather than total cost. 
% Furthermore, all decision points in our algorithm that currently rely on total cost must be adapted to prioritize makespan instead.
% \is{This is trivial, need not be mentioned.}

%\isb{Mention how this work will be applicable to makespan.}

\bibliography{references.bib}    

\begin{thebibliography}{23}
\providecommand{\natexlab}[1]{#1}

\bibitem[{Aakash and Saha(2022)}]{our1stPaper}
Aakash; and Saha, I. 2022.
\newblock It Costs to Get Costs! {A} Heuristic-Based Scalable Goal Assignment Algorithm for Multi-Robot Systems.
\newblock In \emph{ICAPS}, 2--10. {AAAI} Press.

\bibitem[{Chegireddy and Hamacher(1987)}]{chegireddy1987algorithms}
Chegireddy, C.~R.; and Hamacher, H.~W. 1987.
\newblock Algorithms for Finding k-Best Perfect Matchings.
\newblock \emph{Discrete applied mathematics}, 18(2): 155--165.

\bibitem[{Chen et~al.(2021)Chen, Alonso-Mora, Bai, Harabor, and Stuckey}]{Chen21}
Chen, Z.; Alonso-Mora, J.; Bai, X.; Harabor, D.~D.; and Stuckey, P.~J. 2021.
\newblock Integrated Task Assignment and Path Planning for Capacitated Multi-Agent Pickup and Delivery.
\newblock \emph{IEEE Robotics and Automation Letters}, 6(3): 5816--5823.

\bibitem[{Das, Nath, and Saha(2021)}]{DasNS21}
Das, S.~N.; Nath, S.; and Saha, I. 2021.
\newblock OMCoRP: An Online Mechanism for Competitive Robot Prioritization.
\newblock In \emph{ICAPS}, 112--121.

\bibitem[{Eppstein(2016)}]{eppstein2016encyclopedia}
Eppstein, D. 2016.
\newblock Encyclopedia of Algorithms, Chapter k-Best Enumeration.
\newblock \emph{Springer}, 680: 1003--1006.

\bibitem[{Fulkerson, Glicksberg, and Gross(1953)}]{lbap}
Fulkerson, D.~R.; Glicksberg, I.~L.; and Gross, O.~A. 1953.
\newblock \emph{A Production-Line Assignment Problem}.
\newblock Santa Monica, California: The Rand Corporation.

\bibitem[{{Gonzalez-de-Santos} et~al.(2017){Gonzalez-de-Santos}, Ribeiro, Fernandez-Quintanilla, Lopez-Granados, Brandstoetter, Tomic, Pedrazzi, Peruzzi, Pajares, Kaplanis, Perez-Ruiz, Valero, del Cerro, Vieri, Rabatel, and Debilde}]{Gonzalez-de-Santos17}
{Gonzalez-de-Santos}, P.; Ribeiro, A.; Fernandez-Quintanilla, C.; Lopez-Granados, F.; Brandstoetter, M.; Tomic, S.; Pedrazzi, S.; Peruzzi, A.; Pajares, G.; Kaplanis, G.; Perez-Ruiz, M.; Valero, C.; del Cerro, J.; Vieri, M.; Rabatel, G.; and Debilde, B. 2017.
\newblock Fleets of robots for environmentally-safe pest control in agriculture.
\newblock \emph{Precision Agriculture}, 18: 574--614.

\bibitem[{Grippa et~al.(2019)Grippa, Behrens, Wall, and Bettstetter}]{GrippaBWB19}
Grippa, P.; Behrens, D.~A.; Wall, F.; and Bettstetter, C. 2019.
\newblock Drone delivery systems: job assignment and dimensioning.
\newblock \emph{Auton. Robots}, 43(2): 261--274.

\bibitem[{Gross(1959)}]{gross1959bottleneck}
Gross, O. 1959.
\newblock The Bottleneck Assignment Problem.
\newblock Technical Report P-1620, The Rand Corporation, Santa Monica, California.

\bibitem[{H{\"o}nig et~al.(2018)H{\"o}nig, Kiesel, Tinka, Durham, and Ayanian}]{cbsta}
H{\"o}nig, W.; Kiesel, S.; Tinka, A.; Durham, J.; and Ayanian, N. 2018.
\newblock Conflict-Based Search with Optimal Task Assignment.
\newblock In \emph{AAMAS}, 757--765.

\bibitem[{Kuhn(1955)}]{kuhn1955hungarian}
Kuhn, H.~W. 1955.
\newblock The Hungarian Method for the Assignment Problem.
\newblock \emph{Naval research logistics quarterly}, 2(1-2): 83--97.

\bibitem[{Li et~al.(2021)Li, Tinka, Kiesel, Durham, Kumar, and Koenig}]{TKDKK21}
Li, J.; Tinka, A.; Kiesel, S.; Durham, J.~W.; Kumar, T. K.~S.; and Koenig, S. 2021.
\newblock Lifelong Multi-Agent Path Finding in Large-Scale Warehouses.
\newblock In \emph{AAAI}, 11272--11281.

\bibitem[{Ma and Koenig(2016)}]{ma2016optimal}
Ma, H.; and Koenig, S. 2016.
\newblock Optimal Target Assignment and Path Finding for Teams of Agents.
\newblock In \emph{AAMAS}, 1144--1152.

\bibitem[{Murty(1968)}]{murty1968algorithm}
Murty, K.~G. 1968.
\newblock An Algorithm for Ranking all the Assignments in Order of Increasing Cost.
\newblock \emph{Operations Research}, 16(3): 682--687.

\bibitem[{Sharon et~al.(2012)Sharon, Stern, Felner, and Sturtevant}]{cbs_conference}
Sharon, G.; Stern, R.; Felner, A.; and Sturtevant, N.~R. 2012.
\newblock Conflict-Based Search For Optimal Multi-Agent Path Finding.
\newblock In \emph{AAAI}, 563--569. {AAAI} Press.

\bibitem[{Sharon et~al.(2015)Sharon, Stern, Felner, and Sturtevant}]{cbs_journal}
Sharon, G.; Stern, R.; Felner, A.; and Sturtevant, N.~R. 2015.
\newblock Conflict-Based Search for Optimal Multi-Agent Pathfinding.
\newblock \emph{Artificial Intelligence}, 219: 40--66.

\bibitem[{Stern et~al.(2019)Stern, Sturtevant, Felner, Koenig, Ma, Walker, Li, Atzmon, Cohen, Kumar, Boyarski, and Bartak}]{stern2019mapf}
Stern, R.; Sturtevant, N.~R.; Felner, A.; Koenig, S.; Ma, H.; Walker, T.~T.; Li, J.; Atzmon, D.; Cohen, L.; Kumar, T. K.~S.; Boyarski, E.; and Bartak, R. 2019.
\newblock Multi-Agent Pathfinding: Definitions, Variants, and Benchmarks.
\newblock \emph{SoCS}, 151--158.

\bibitem[{Sturtevant(2012)}]{sturtevant2012benchmarks}
Sturtevant, N. 2012.
\newblock Benchmarks for Grid-Based Pathfinding.
\newblock \emph{Transactions on Computational Intelligence and AI in Games}, 4(2): 144 -- 148.

\bibitem[{Tian et~al.(2009)Tian, Yang, Qi, and Yang}]{Tian09}
Tian, Y.-T.; Yang, M.; Qi, X.-Y.; and Yang, Y.-M. 2009.
\newblock Multi-robot task allocation for fire-disaster response based on reinforcement learning.
\newblock In \emph{2009 International Conference on Machine Learning and Cybernetics}, volume~4, 2312--2317.

\bibitem[{Turpin, Michael, and Kumar(2013)}]{6630671}
Turpin, M.; Michael, N.; and Kumar, V. 2013.
\newblock Concurrent assignment and planning of trajectories for large teams of interchangeable robots.
\newblock In \emph{ICRA}, 842--848.

\bibitem[{Turpin et~al.(2013)Turpin, Mohta, Michael, and Kumar}]{Turpin-RSS-13}
Turpin, M.; Mohta, K.; Michael, N.; and Kumar, V. 2013.
\newblock Goal Assignment and Trajectory Planning for Large Teams of Aerial Robots.
\newblock In \emph{RSS}. Berlin, Germany.

\bibitem[{Turpin et~al.(2014)Turpin, Mohta, Michael, and Kumar}]{DBLP:journals/arobots/TurpinMMK14}
Turpin, M.; Mohta, K.; Michael, N.; and Kumar, V. 2014.
\newblock Goal assignment and trajectory planning for large teams of interchangeable robots.
\newblock \emph{Auton. Robots}, 37(4): 401--415.

\bibitem[{Yu and LaValle(2013)}]{Yu_LaValle_2013}
Yu, J.; and LaValle, S. 2013.
\newblock Structure and Intractability of Optimal Multi-Robot Path Planning on Graphs.
\newblock \emph{Proceedings of the AAAI Conference on Artificial Intelligence}, 27(1): 1443--1449.

\end{thebibliography}
%%%%%%%%%%%%%%%%%%%%%%%%%%%%%%%%%

% ~\newpage ~\newpage ~\newpage
% ~\newpage 
% \input{appendix}

\end{document}